\documentclass [12pt] {article}

\usepackage[utf8]{inputenc}
\usepackage{amssymb} 
\usepackage{amsmath,amsfonts} 

\usepackage{mathrsfs} 
\usepackage{blindtext}
\usepackage[utf8]{inputenc}  
\usepackage[english]{babel}  
\usepackage{graphicx}
\usepackage{physics} 
\usepackage[toc,page]{appendix} 
\usepackage{url}
\usepackage{mathtools}
\usepackage{setspace}
\usepackage{subfig}
\usepackage{apacite}
\usepackage{pdfpages}
\usepackage{afterpage}
\usepackage{enumitem}
\usepackage{float}
\usepackage{enumerate}
\usepackage{enumitem, kantlipsum}
\usepackage{multirow}
\usepackage{inconsolata}
\usepackage{subcaption}
\usepackage{natbib}
\usepackage{threeparttable}
\usepackage{multirow}
\usepackage{rotating}
\usepackage{fullpage}
\usepackage{tabularx}
\usepackage{array}
\usepackage{graphicx}
\usepackage{subcaption}
\usepackage{arydshln}
\usepackage{pdflscape}
\usepackage{environ}
\usepackage{adjustbox}
\usepackage{booktabs}
\usepackage{siunitx}
\usepackage{upgreek}

\usepackage{hyperref}
\hypersetup{
    colorlinks=true,
    linkcolor=blue,
    filecolor=magenta,
    citecolor=black, 
    urlcolor=cyan
}

\newcommand{\comment}[1]{}

\usepackage[margin=1in]{geometry}
\newlength\mystoreparindent

\usepackage{chngcntr}
\counterwithout{equation}{section} 

\newtheorem{theorem}{Theorem}

\newtheorem{defi}{Definition}
\newtheorem{lem}{Lemma}
\newtheorem{assum}{Assumption}

\newtheorem{prop}{Proposition}

\newtheorem{corol}{Corollary}
\newenvironment{proof}[1][Proof]{\begin{trivlist}
\item[\hskip \labelsep {\bfseries #1}]}{\end{trivlist}}

\numberwithin{thm}{section}
\numberwithin{prf}{section}
\numberwithin{equation}{section}
\counterwithin{table}{section}
\counterwithin{figure}{section}

\newcommand*{\QEDB}{\hfill\ensuremath{\square}}

\newcolumntype{Y}{>{\centering\arraybackslash}X}

\comment{
\usepackage{fancyhdr}
\pagestyle{fancy}
\fancyhf{}
\rhead{Dissertation}
\lhead{}
\rfoot{Page \thepage}
}

\makeatletter
\newcommand*\bigcdot{\mathpalette\bigcdot@{.5}}
\newcommand*\bigcdot@[2]{\mathbin{\vcenter{\hbox{\scalebox{#2}{$\m@th#1\bullet$}}}}}
\makeatother

\usepackage{hyperref}
\hypersetup{
    colorlinks=true,
    linkcolor=black,
    filecolor=black,      
    urlcolor=black,
}
\usepackage{listings}
\usepackage{color}
\definecolor{mygreen}{RGB}{28,172,0}
\definecolor{mylilas}{RGB}{170,55,241}
\newcommand{\mcA}{\mathcal{A}}
\newcommand{\mcB}{\mathcal{B}}
\newcommand{\mcC}{\mathcal{C}}

\newcommand{\mcE}{\mathcal{E}}
\newcommand{\mcF}{\mathcal{F}}
\newcommand{\mcG}{\mathcal{G}}
\newcommand{\mcH}{\mathcal{H}}
\newcommand{\mcI}{\mathcal{I}}
\newcommand{\mcJ}{\mathcal{J}}
\newcommand{\mcK}{\mathcal{K}}
\newcommand{\mcL}{\mathcal{L}}
\newcommand{\mcM}{\mathcal{M}}
\newcommand{\mcN}{\mathcal{N}}
\newcommand{\mcO}{\mathcal{O}}

\newcommand{\mcR}{\mathcal{R}}
\newcommand{\mcS}{\mathcal{S}}

\newcommand{\mcU}{\mathcal{U}}
\newcommand{\mcV}{\mathcal{V}}
\newcommand{\mcW}{\mathcal{W}}
\newcommand{\mcX}{\mathcal{X}}

\newcommand{\mbE}{\mathbb{E}}

\newcommand{\mbP}{\mathbb{P}}
\newcommand{\mbR}{\mathbb{R}}
\newcommand{\mbV}{\mathbb{V}}

\newcommand{\ve}{\varepsilon}
\newcommand{\vp}{\varphi}
\newcommand{\vr}{\varrho}
\newcommand{\vs}{\varsigma}
\newcommand{\vt}{\vartheta}

\newcommand{\sumin}{\sum_{i=1}^n}

\newcommand{\diag}{\textup{diag}}
\newcommand{\mineig}{\textup{mineig}}
\newcommand{\traceoperator}{\textup{tr}}
\newcommand{\cl}{\textup{cl}}

\newcommand{\BC}{\textup{BC}}
\newcommand{\MSE}{\textup{MSE}}
\newcommand{\CER}{\textup{CER}}
\newcommand{\BA}{\textup{BA}}

\begin{document}

\title{\Large{The moment is here: a generalized class of estimators for fuzzy regression discontinuity designs}\thanks{I give special thanks to Sami Stouli for his extensive comments, suggestions, conversations, and guidance on this paper. I also thank Xavier D'Haultfoeuille, Vincent Han, Arthur Lewbel, Adam McCloskey, Claudia Noack, Senay Sokullu, Frank Windmeijer, and participants in seminars at University College London, the University of Bristol, the Bristol Econometric Study Group, the European Summer Meeting of the Econometric Society for helpful comments, suggestions, and conversations.  Code for estimation and inference in R and Python is available at \href{https://github.com/stuart-lane/LambdaFRD}{https://github.com/stuart-lane/LambdaFRD}.}}
\author{\large{Stuart Lane}\thanks{School of Economics, University of Bristol, Bristol, BS8 1TU, UK, stuart.lane@bristol.ac.uk \\
\hspace{10mm}I gratefully acknowledge support from UKRI through postdoctoral funding under the UK government’s Horizon Europe funding guarantee (EP/Y004159/1).}}
\date{27/11/2026}

\maketitle

\begin{abstract}
I show that the conventional fuzzy regression discontinuity (FRD) estimator does not possess finite integer moments under suitable regularity conditions, regardless of local-polynomial degree and kernel function. The estimator can be highly dispersed in small samples or when the treatment probability discontinuity is small. I present a generalized class of FRD estimators, indexed by a single tuning parameter, that preserves all finite moments from the error distribution and nests both conventional FRD and sharp estimators. Simple data-driven rules for selecting the tuning parameter are developed. In simulations, these estimators weakly dominate the conventional FRD estimator in median absolute error, control root mean squared error, and have competitive median bias. Bias-corrected confidence intervals for the new class provide both lower median absolute coverage error and reduced median length relative to conventional bias-corrected confidence intervals. Estimator stability and performance are demonstrated using data on HIV treatment and retention in South Africa. \par\smallskip
\end{abstract}

\thispagestyle{empty}

\section{Introduction}\label{sec introduction}

Regression discontinuity (RD) designs are important tools for treatment effect estimation and causal inference when the probability of receiving treatment depends on an observed running variable, and is discontinuous at some known cutoff \citep{thistlethwaite1960regression, hahn2001identification}. If treatment receipt is a deterministic function of the running variable, the design is sharp (SRD). If instead the probability of receiving treatment jumps at the cutoff by less than one, so that receipt is not deterministic, the design is fuzzy (FRD). These methods have been applied to a range of topics, including education economics and educational outcomes \citep{angrist1999using}, industrial organization \citep{busse20061}, political economy \citep{lee2008randomized}, environmental economics \citep{salman2022paris}, the effects of social/financial aid programs \citep{ludwig2007does} and medical applications \citep{bor2017treatment}.

The treatment effect at the cutoff in the fuzzy design is identified by the ratio of the discontinuities of the outcome and treatment conditional expectation functions \citep{hahn2001identification}. The conventional estimator uses four local-polynomial regressions, applied separately to the outcome and treatment conditional expectation functions, on either side of the cutoff. Here, the numerator and denominator are treated as separate SRD objects; this is justified asymptotically via the delta method, but ratio estimators potentially lack finite moments in finite samples if the denominator can be close to zero with sufficient probability (e.g., the just-identified linear instrumental variables (IV) estimator under joint error normality \citep{staiger1997instrumental, chao2013expository, andrews2019weak}). Despite a significant theoretical literature discussing optimal choices regarding the degree of polynomial regression \citep{gelman2019high}, choice of kernel \citep{cheng1997automatic, imbens2008regression} and bandwidth selection \citep{imbens2012optimal, arai2016optimal,calonico2020optimal}, this particular issue of finite-sample instability has received limited attention; \citet{feir2016weak} and \citet{noack2024bias} both develop identification-robust test statistics, but they do not consider estimation. These tests build on ideas from the identification-robust testing literature in linear IV \citep{anderson1949}.

I show that the conventional FRD estimator does not have finite integer moments in finite samples under suitable regularity conditions, regardless of the degree of local-polynomial regression, the kernel function, or the bandwidth used. This follows as the density of the denominator is bounded away from zero in an open neighborhood of the origin. Consequently, the sampling distribution has heavy tails, and may provide a poor approximation of the asymptotic normal distribution. This can lead to unstable point estimators and inaccurate inference in small samples or with a small discontinuity in the treatment receipt probability at the cutoff. This result provides an analogue to the non-existence of finite integer moments in finite samples for the ratio-IV estimator \citep{sawa1972finite}, as the uniform-kernel local-linear estimator is numerically equivalent to a ratio-IV estimator \citep{hahn2001identification, imbens2008regression}.

To address this problem, I introduce a class of estimators indexed by a single new tuning parameter $\lambda$, which nests the conventional FRD estimator. I show that a continuum of estimators within this class preserves all finite moments of the error distribution, so if the errors have finite moments of all orders (e.g., errors are normal, Laplace, etc.), then so will the new estimators. This is achieved as appropriate selections of $\lambda$ will bound the denominator away from zero. These estimators are computationally simple, provide a local analogue of the $k$-class estimators from linear IV \citep{theil1958economic, nagar1959bias, kinal1980existence}, and given suitable conditions on $\lambda$, are asymptotically equivalent to the conventional FRD estimator. The construction of the class builds on a generalized IV representation of the conventional FRD estimator. These estimators perform well across a large parameter grid in simulations with simple data-driven choices for $\lambda$, dominating the conventional estimator on median absolute error, controlling root mean squared error in accordance with the moment theorems of this paper, while also having similar median bias (the median bias ordering between the conventional FRD and $\lambda$-class estimators depends on the strength of the endogenous selection and the curvature of the structural function, but the difference is typically not large in either direction). Here, $\lambda$ is selected to remove the leading-order ratio bias of the estimator; when the uniform kernel is used, this selection rule is equivalent to the Fuller correction from linear IV \citep{fuller1977some}, which is known to lead to substantially improved finite-sample behavior in that model \citep[e.g.,][]{harding2016finite}. 

These estimators have the same asymptotic smoothing bias as conventional FRD estimators. Therefore, I construct a bias-corrected $\lambda$-class estimator and confidence interval, allowing for pointwise asymptotically valid inference without undersmoothing. The new $\lambda$-class bias-corrected confidence intervals perform well in simulations; with the appropriate choice of $\lambda$ (modified from before, as bias correction introduces its own additional ratio bias), and the coverage-error-optimal bandwidth, they reduce median absolute coverage error by over a third relative to the conventional bias-corrected confidence interval \citep{calonico2014robust}, and by at least half in the low and medium endogeneity designs, while having shorter median length in every configuration considered, with an average reduction of around 25\%. Under strong endogenous selection, this comes at the cost of some undercoverage.

The construction of the new class results from a generalized IV representation of the conventional FRD estimator. Specifically, I show that the aforementioned numerical equivalence of the local-linear uniform-kernel estimator to an IV estimator is generalizable to all FRD estimators, regardless of the choice of polynomial order or kernel function. Therefore, all conventional FRD estimators, including the standard local-linear triangular-kernel estimator, are implementable as a single weighted IV regression, and this leads naturally to a local analogue of the $k$-class estimators. This class has advantages over other potential forms of moment-restoring regularizations such as Tikhonov regression or trimming, as the class is equivariant to shifts in the treatment effect, can remove leading-order ratio bias, and requires only a single tuning parameter (which is unitless, scale-invariant, and dependent only on the running variable and kernel function). This generalization is also of independent interest, and deepens the already established links between the linear IV and the FRD literatures.

I demonstrate the stability of the new class in an empirical application, looking at the effect of antiretroviral therapy (ART) initiation on retention in HIV care in South Africa, using data from \citet{bor2017treatment}. The small estimated first-stage effects and relatively small sample sizes make this an ideal application to study the sensitivity of FRD estimators. Conventional FRD estimators and confidence intervals (and their bias-corrected counterparts) show a substantial degree of sensitivity to bandwidth selection, whereas the new $\lambda$-class estimators and confidence intervals exhibit much greater stability across the range of considered bandwidths. While the bias-corrected point estimators are similar when using common data-driven bandwidths, the new bias-corrected $\lambda$-class confidence intervals are tighter than conventional bias-corrected confidence intervals by 20--27\% using coverage-error-optimal bandwidths.

In Section \ref{sec model}, I introduce the model and show that the conventional FRD estimator does not have finite integer moments in finite samples. In Section \ref{sec lambda class}, I introduce the $\lambda$-class of FRD estimators and show that it preserves all finite moments from the error distribution in finite samples. Section \ref{sec asymptotics} provides asymptotic equivalence results, constructs bias-corrected $\lambda$-class estimators and confidence intervals, and develops selection rules for $\lambda$. Section \ref{sec monte carlo} gives a Monte Carlo study, and Section \ref{sec empirical} gives an empirical application with the \citet{bor2017treatment} dataset looking at the effects of early ART initiation on retention in care. Proofs of theorems and supporting lemmas, and additional simulations, are presented in the Appendix.

\section{Model and the moment problem}\label{sec model}

Consider the model
\begin{align}
    Y_i = m(X_i) + \tau(X_i) D_i + U_i, \label{eq main model}
\end{align}
where $Y_i\in\mbR$ is the outcome of interest, $X_i\in \mcX\subseteq \mbR$ is the running variable, $m(\cdot) : \mcX \to \mbR$ is some unknown nuisance function, $\tau(\cdot) : \mcX \to \mbR$ is the treatment effect function, $D_i\in\{0,1\}$ is the treatment received, and $U_i\in\mbR$ is an unobserved error.
\begin{assum} \label{as dgp}
    (i) $\{(Y_i,X_i,D_i)\}_{i=1}^n$ is an i.i.d. sample. (ii) $X_i$ has density $f_X(x)$ for each $x\in \mcX$. (iii) (a) $\mbE[U_i | X_i = x] = 0$ for every $x \in \mcX$; (b) $\sup_{x \in \mcX}\mbE[|U_i|^r|X_i = x] < \infty$ for all $r \in (0,r^*)$, for some $r^* \in (1,\infty]$, and $\mbE[|U_i|^r|X_i = x] = \infty$ for every $x \in \mcX$ and for all $r \geq r^*$; (c) the conditional density $f_{U|X}(\cdot|x)$ exists, and for any compact $\mcK \subset \mbR$, $\inf_{(u,x)\in \mcK \times \mcX}f_{U|X}(u|x) > 0.$
\end{assum}
Parts (i)--(ii) of Assumption \ref{as dgp} are standard, assuming an i.i.d. sample with a continuous running variable. Part (iii) provides the moment thresholds for the conditional distribution of $U_i$ given $X_i$. The errors may possess finite moments of all orders (e.g., normal errors), or may have a finite mean but infinite variance (e.g., $t$ with $1 < \nu \leq 2$ degrees of freedom). Errors also have unbounded support with strictly positive density over compact subsets of $\mbR$. Define $\pi: \mcX\to [0,1]$ such that $\pi(x) = \mbP\{D_i = 1 | X_i = x\} = \mbE[D_i | X_i = x]$ for each $x\in\mcX$, assumed continuous everywhere except at a single known point $x_0 \in \mcX$, referred to as the cutoff. Further, define $\mu(x) = \mbE[Y_i|X_i = x] = m(x) + \tau(x)\pi(x)$ and the one-sided limits for $s \in \{+,-\}$,
\begin{align*}
    \mu_s(x_0) = \lim_{x\to x_0^s}\mbE[Y_i|X_i = x],\,\,\,
    \pi_s(x_0) = \lim_{x\to x_0^s}\mbE[D_i|X_i = x].
\end{align*}

\begin{assum}\label{as limits}
    (i) $\mu_+(x_0)$, $\mu_-(x_0)$, 
    $\pi_+(x_0)$, and $\pi_-(x_0)$ exist. (ii) $\pi_+(x_0) \neq \pi_-(x_0)$. 
\end{assum}
Assumption \ref{as limits} ensures that the conditional probability of receiving treatment is discontinuous at $x_0$, and therefore the cutoff is well-defined. Given Assumptions \ref{as dgp} and \ref{as limits}, and continuity of $m(\cdot)$ and $\tau(\cdot)$ at $x_0$, define the treatment effect of interest as $\tau = \tau(x_0)$, the treatment effect at the cutoff. Then, \citet{hahn2001identification} show that $\tau$ is identified by
\begin{align}\label{eq tau}
    \tau = \frac{\tau^Y}{\tau^D}, \,\,\, \tau^Y = \mu_+(x_0) - \mu_-(x_0), \,\,\, \tau^D = \pi_+(x_0) - \pi_-(x_0),
\end{align}
where Assumption \ref{as limits} immediately gives $\tau^D \neq 0$.

\subsection{The conventional estimator and non-existence of finite integer moments}\label{subsec standard est}
The conventional approach to estimating \eqref{eq tau} uses local-polynomial regressions to separately estimate each of $\mu_+(x_0)$, $\mu_-(x_0)$, $\pi_+(x_0)$, and $\pi_-(x_0)$. For $s \in \{+,-\}$, the local-polynomial estimators of degree $p\geq 0$ are
\begin{equation} \label{eq main estimators}
    \begin{aligned}
        \hat{\mu}_{p,s}(x_0) &= \frac{\displaystyle\sum_{i\in \mcN_s}Y_i K_{h}(X_i-x_0)\omega_{p,s,i}}{\displaystyle\sum_{i\in \mcN_s} K_{h}(X_i-x_0)\omega_{p,s,i}}, \,\,\, \hat{\pi}_{p,s}(x_0) &= \frac{\displaystyle\sum_{i\in \mcN_s}D_i K_{h}(X_i-x_0)\omega_{p,s,i}}{\displaystyle\sum_{i\in \mcN_s} K_{h}(X_i-x_0)\omega_{p,s,i}}, 
    \end{aligned}
\end{equation}
where $K(\cdot)$ is some kernel function, $K_h(\upsilon) = K(\upsilon/h)/h$ for fixed bandwidth $h > 0$, and for the intervals $\mcX_{h,+} = [x_0,x_0+h]$ and $\mcX_{h,-} = [x_0-h,x_0)$, define the sets $\mcN_+ = \{i: X_i \in \mcX_{h,+}\}$ and $\mcN_- = \{i: X_i \in \mcX_{h,-} \}$.\footnote{Taking $h$ as fixed is natural for finite-sample analysis. The non-existence of moments is not a consequence of bandwidth choice, and is not remedied via appropriate bandwidth selection. Simulations provide clear evidence of the moment problem for the conventional FRD estimator with the common MSE- and coverage-error-rate-optimal bandwidths \citep{imbens2012optimal, calonico2020optimal}. I relax the fixed $h$ assumption in Section \ref{sec asymptotics} for asymptotic analysis.} Denote $\mcX_h = \mcX_{h,-} \cup \mcX_{h,+}$. The generalized kernel weights in \eqref{eq main estimators} are constructed as $\omega_{p,s,i} = e_{p,1}' \mcS^{-1}_{p,s}\mcH_{p,i}$, where $e_{\ell,j}$ is a $(\ell+1)$-vector with $j$-th element 1 and 0 otherwise, $\mcH_{\ell,i}$ is a $(\ell+1)$-vector with $j$-th element $(X_i - x_0)^{j-1}$ and $\mcS_{\ell,s} = \sum_{i \in \mcN_s} K_h(X_i - x_0) \mcH_{\ell,i}  \mcH_{\ell,i}'$. These four local-polynomial estimators are then combined to form the estimator
\begin{align}\label{eq tau_hat}
    \hat{\tau}_p = \frac{\hat{\tau}_p^Y}{\hat{\tau}_p^D}, \,\,\, \textup{with} \,\,\, \hat{\tau}_p^Y = \hat{\mu}_{p,+}(x_0) - \hat{\mu}_{p,-}(x_0), \,\,\, \hat{\tau}_p^D = \hat{\pi}_{p,+}(x_0) - \hat{\pi}_{p,-}(x_0).
\end{align}
Define the effective sample to be the set of individuals $i \in \mcN_h$ for $\mcN_h = \mcN_+ \cup \mcN_-$, with $n_s = |\mcN_s|$ for $s \in \{+,-\}$ and $n_h = n_+ + n_-$. Write $k_{p,s,i} = K_{h}(X_i-x_0)\omega_{p,s,i}$ and then for each $i \in \mcN_h$, denote $k_{p,i} = k_{p,+,i}1\{i \in \mcN_+\} - k_{p,-,i}1\{i \in \mcN_-\}$ (with $k_{p,i} = 0$ for $i \notin \mcN_h$) such that $\hat{\tau}_p^W = \sum_{i \in \mcN_h} k_{p,i}W_i$, for $W \in \{Y,D\}$. Stacking $Y_i$, $D_i$ and $\mcH_{p,i}'$ for $i \in \mcN_h$ gives the $n_h\times 1$ vectors $Y$ and $D$, and $n_h\times(p+1)$ matrix $\mcH_p$. Also, let $K = \diag(K_h(X_i - x_0))$ for $i \in \mcN_h$.\footnote{For notational clarity, $K(\cdot)$ denotes the scalar kernel function, $K_h(\cdot)$ denotes the scaled version, and $K=\diag(K_h(X_i - x_0))$ denotes the associated diagonal matrix.} Then, \eqref{eq tau_hat} can also be expressed as
\begin{equation}\label{eq tau p matrix}
    \hat{\tau}_p = \frac{e_{p,1}'(\mcS_{p,+}^{-1}\mcH_{p,+}'K_+Y_+ - \mcS_{p,-}^{-1}\mcH_{p,-}'K_-Y_-)}{e_{p,1}'(\mcS_{p,+}^{-1}\mcH_{p,+}'K_+D_+ - \mcS_{p,-}^{-1}\mcH_{p,-}'K_-D_-)},
\end{equation}
where for generic vector or matrix $B$ with $i$-th element or row $B_i'$, $B_s$ denotes just the stacked elements or rows of $B$ relating to individuals $i \in \mcN_s$ for $s \in \{+,-\}$. I treat $p \geq 0$ as fixed, as is common in practice.\footnote{Researchers often fix $p = 1$ following \citet{fan1996local}, \citet{hahn2001identification} and \citet{gelman2019high}, and sometimes use higher degrees as a robustness check.}
\begin{assum} \label{as kernel}
    (i) $K(\cdot)$ is a symmetric second-order kernel with support $[-1,1]$, strictly positive on $(-1,1)$, and $L_K$-Lipschitz on $[-1,1]$. (ii) $K(\cdot) \in  \mcC_b^{\vs}((-1,1)\setminus \{0\})$ for some $\vs \geq 1$, for $\mcC_b^\vs(\mcA)$ the space of $\vs$-times continuously differentiable functions with bounded derivatives on $\mcA$. (iii) $K(\upsilon) \leq C_K$ for every $\upsilon\in[-1,1]$ for some $C_K < \infty$.
\end{assum}
All common kernels used for FRD estimation (the triangular, uniform, and Epanechnikov kernels) satisfy Assumption \ref{as kernel}. Part (ii) requires $\vs$-times continuous differentiability with bounded derivatives for $\vs \geq 1$; while this is not usually explicitly assumed, it is satisfied for $\vs=\infty$ for all kernels used empirically, and should be taken as a technical assumption rather than a meaningful restriction in practice. Differentiability of the kernel gives differentiability of $\hat{\tau}_p^D$ (see Lemma \ref{lem differentiability}), and is required as I show that admissible first-stage configurations $(x,d)$ exist such that $\hat{\tau}_p^D = 0$, and then use geometric arguments to show that the density of $\hat{\tau}_p^D$ is bounded away from zero in an open neighborhood of $\hat{\tau}_p^D = 0$. This leads to the expectation integral diverging. The punctured interval in part (ii) allows for the triangular kernel $K(\upsilon) = (1 - |\upsilon|)1\{|\upsilon| \leq 1\}$, which is the most popular choice as it has optimal properties for boundary estimation \citep{cheng1997automatic}, but is not differentiable at $\upsilon = 0$. 
\begin{assum}\label{as finite sample}
    (i) $f_X(x)$ is bounded away from both 0 and $\infty$ for all $x \in \mcX_h$. Further, $m(\cdot)$ and $\tau(\cdot)$ are continuous over $\mcX_h$.
    (ii) $D_i = f_d(X_i, U_i, v_i)$ for some measurable function $f_d(\cdot)$, where $v_1,\hdots,v_n$ are mutually independent and independent of $(X,U)$. Let $\pi_U(x, u)= \mbP\{D_i = 1|X_i = x, U_i = u\}$. Then, for every $B > 0$, there exists some $\pi_B >0$ such that  $\pi_B \leq \pi_U(x,u) \leq 1-\pi_B$ for all $x \in \mcX_h$, $|u| \leq B$.
    (iii) There exist $i,j \in \mcN_+$, $i \neq j$, such that $D_i = 0$ and $D_j =1$.
    (iv) There exists $C_{\mcS}>0$ such that the minimum eigenvalue of $\mcS_{p,s}$ satisfies $\mineig(\mcS_{p,s}) \geq C_{\mcS}$ for $s \in \{+,-\}$.  
    (v) $\mbP\{X_i \notin \mcX_h\} > 0$.
\end{assum}
Part (i) mildly strengthens typical regularity conditions to the bandwidth window $\mcX_h$ instead of an arbitrary open neighborhood of $x_0$ usually imposed for asymptotic results. This is necessary for finite-sample results. Part (ii) imposes a local overlap condition; given $X_i$ and any finite realization of $U_i$, the probability of receiving treatment is bounded away from both 0 and 1. Part (iii) rules out perfect compliance within the observed effective sample (an observationally sharp design), and excludes the pathological case of homogeneous treatment status. For example, if $D_i = 0$ for every $i \in \mcN_h$, this trivially means $\hat{\tau}_p^D = 0$, and so $\hat{\tau}_p$ is undefined (this is theoretically compatible with the data generating process assumption in part (ii)). Part (iv) assumes that the sample kernel moment matrices are well-conditioned, implying $n_s \geq p+1$. Also note that $\sum_{i\in \mcN_s} K_{h}(X_i-x_0)\omega_{p,s,i} = 1$ under Assumptions \ref{as kernel} and \ref{as finite sample}, so the denominators can be dropped from the estimators in \eqref{eq main estimators}. Parts (iii) and (iv) of Assumption \ref{as finite sample} restrict the observed sample, not the data generating process. Let $\mcE$ denote the event that the observed effective sample satisfies both parts. Under Assumptions \ref{as dgp}--\ref{as kernel} and \ref{as finite sample}(i)--(ii), $C_{\mcS}$ can be chosen small enough such that $\mbP\{\mcE\} > 0$ for $n \geq 2\max \{p+1,2\}$. Throughout the finite-sample analysis of the conventional FRD estimator, all probabilities, densities, and expectations are taken under the law conditional on $\mcE$, but for notational convenience I leave this implicit. This strengthens the empirical relevance of the subsequent moment theorems, as the non-existence of finite moments is not the result of degenerate samples, but rather occurs for samples considered well-behaved. Without part (v), certain technical arguments requiring symmetry, as introduced in the following definition, may not hold. This condition is satisfied whenever the support of $X_i$ strictly contains $\mcX_h$.
\begin{defi}\label{def symmetry}
The first-stage effective sample $(x^*,d^*)$ is symmetric if there exists a bijection $\Pi: \mcN_+ \to \mcN_-$ such that $x_i^* - x_0 = -(x_{\Pi(i)}^* - x_0)$ and $d_i^* = d_{\Pi(i)}^*$ for each $i\in\mcN_+$, $\Pi(i)\in\mcN_-$.
\end{defi}
\begin{assum}\label{as gradient}
    (i) For admissible $d$, define $g_d : \mcX_h^{n_h} \to \mbR$, $g_d(x) = \hat{\tau}_p^D(x,d)$. Then, there exists an admissible symmetric first-stage effective sample $(x^*,d^*)$, lying in the interior of the admissible configuration set for fixed $d^*$, and with $x^* \in (\mcX_h^\circ \setminus \{x_0\})^{n_h}$ such that $\norm{\nabla_x g_{d^*}(x^*)}_2 > 0$, where $\mcA^\circ$ denotes the interior of set $\mcA$, $\nabla_x$ denotes the gradient with respect to $x$ and $\norm{\cdot}_2$ is the Euclidean norm. (ii) The density $f_{\hat{\tau}_p^D}(t_d)$ of $\hat{\tau}_p^D$ exists, and there exist constants $\delta_d>0$ and $C_D<\infty$ such that $\sup_{|t_d|<\delta_d} f_{\hat{\tau}_p^D}(t_d)\le C_D$.
\end{assum}
Assumption \ref{as gradient}(i) is a local nondegeneracy condition, which requires that $\hat{\tau}_p^D$ varies with at least one component of $x$ at some admissible interior symmetric first-stage configuration $(x^*,d^*)$. This condition is natural when either $p \geq 1$ or the kernel is not uniform, as both induce dependence of $\hat{\tau}_p^D$ on $x$.\footnote{For the $p=0$ local-constant estimator with uniform kernel, $\hat{\tau}_0^D = n_+^{-1}\sum_{i \in \mcN_+}D_i - n_-^{-1}\sum_{i \in \mcN_-}D_i$ is unconditionally discrete, and therefore cannot satisfy Assumption \ref{as gradient}. However, $\mbP\{n_+ = 2, n_- = 2, \sum_{i \in \mcN_+} D_i = 1, \sum_{i \in \mcN_-} D_i= 1\} > 0$ for $n \geq 4$ under Assumptions \ref{as dgp} and \ref{as finite sample}, and so $\mbP\{\hat{\tau}_0^D = 0\} > 0$. As $Y_i$ is continuous conditional on $(X,D)$ and $\hat{\tau}_0^Y = n_+^{-1}\sum_{i \in \mcN_+}Y_i - n_-^{-1}\sum_{i \in \mcN_-}Y_i$, $\mbP\{\hat{\tau}_0^Y \neq 0|X,D\} = 1$, so $\mbP\{\hat{\tau}_0^D = 0, \hat{\tau}_0^Y \neq 0\} = \mbE\left[1\{\hat{\tau}_0^D = 0\}\cdot \mbP\{\hat{\tau}_0^Y \neq 0|X,D\}\right] = \mbP\{\hat{\tau}_0^D = 0\} > 0$ implies $\mbP\{|\hat{\tau}_0| = \infty\} > 0$, so $\mbE[|\hat{\tau}_0|^r] = \infty$ for all $r > 0$.} The assumption is imposed to rule out pathological cases where the gradient vanishes. The punctured interior $\mcX_h^\circ \setminus \{x_0\}$ is again used as Assumption \ref{as kernel} does not require differentiability of $K(\upsilon)$ at $\upsilon \in \{-1,0,1\}$. Part (ii) bounds the density of $\hat{\tau}_p^D$ from above near zero, and is used only to establish finiteness of moments of order $r < 1$. Part (i) treats the treatment vector as fixed for geometric arguments, whereas part (ii) considers the density of the first-stage estimator conditional on the preceding finite-sample design assumptions. These assumptions are sufficient to show that the conventional FRD estimator has no finite integer moments.

\begin{theorem} \label{thm no moments}
Let Assumptions \ref{as dgp}--\ref{as gradient} hold and $r > 0$. Then, $\mbE[|\hat{\tau}_{p}|^r] < \infty$ if and only if $r < 1$, $\mbE[|\hat{\tau}_{p}|^r] = \infty$ otherwise.
\end{theorem}
Under suitable regularity conditions, Theorem \ref{thm no moments} holds for arbitrary local-polynomial order, kernel function, and admissible bandwidths (with the special case for the local-constant uniform-kernel estimator, which cannot satisfy Assumption \ref{as gradient}, discussed above). Even with normal errors, which have finite moments of all orders, $\hat{\tau}_p$ does not have a finite mean, because the density of the denominator $\hat{\tau}_p^D$ of $\hat{\tau}_p$ is bounded away from zero over some open neighborhood containing zero. This causes $\mbE[|\hat{\tau}_p^Y/\hat{\tau}_p^D|^r]$ to behave like $C^r\cdot\int_0^a x^{-r}dx$ for some $C \neq 0$, which will diverge if $r \geq 1$ for any $a>0$. Due to this, the estimator may be highly dispersed in small samples, which can lead to potentially inaccurate inference.

\section{A generalized class of FRD estimators}\label{sec lambda class}

It is well known \citep[see, e.g.,][]{imbens2008regression} that when (\ref{eq tau_hat}) is estimated using $p = 1$ and the uniform kernel, the conventional FRD estimator is numerically equivalent to the IV estimator $\tilde{\tau}_{IV}$ with instrument $Z_i = 1\{X_i \geq x_0\}$ in the regression representation
\begin{equation}\label{eq simple_reg}
    Y_i = V_i'\delta + \tau D_i + \epsilon_i,
\end{equation}
where $V_i' = ( 1, \, Z_i(X_i - x_0), \, (1-Z_i)(X_i - x_0) )$ is a vector of included exogenous regressors, $\delta = (\delta_0, \, \delta_1, \, \delta_2)'$, $\epsilon_i = U_i + (m(X_i) - V_i'\delta) + (\tau(X_i) -\tau)D_i$, and \eqref{eq simple_reg} is restricted to just the effective sample $i\in \mcN_h$ \citep[e.g.,][]{hahn2001identification}. The instrument $Z_i$ is the treatment assignment as determined by the cutoff rule for individual $i$. This gives the IV estimator
\begin{align}\label{eq tau iv imbens}
    \tilde{\tau}_{IV} = (Z'M_{V}D)^{-1}Z'M_{V}Y,
\end{align}
where $Z$ is the $n_h\times 1$ vector with $i$-th element $Z_i$, $V$ is the $n_h\times 3$ matrix with $i$-th row $V_i'$, and for some general $m \times q$ matrix $B$, $M_B = I_m - P_B$ for $P_B = B(B'B)^{-1}B'$ and $I_m$ the $m \times m$ identity matrix. The estimator $\tilde{\tau}_{IV}$ is computationally simple and typically numerically similar to $\hat{\tau}_1$ with triangular kernel, and only requires one pass through the data instead of computing four separate local-polynomial regressions. For this reason, $\tilde{\tau}_{IV}$ is also used in practice.

To my knowledge, the fact that the estimator in \eqref{eq tau iv imbens} can be generalized has not been explicitly noted. Numerous papers such as \citet{hahn2001identification}, \citet{imbens2008regression} and \citet{noack2024bias} state the equivalence for the local linear estimator with uniform kernel. However, the equivalence holds for any kernel function and degree of local-polynomial regression. For general $K(\cdot)$ and $p$, denote the weighted data by $\tilde{Y} = K^{1/2}Y$, $\tilde{D} = K^{1/2}D$, $\tilde{Z} = K^{1/2}Z$, $\tilde{\epsilon} = K^{1/2}\epsilon$ and $\tilde{V}_p = K^{1/2}V_p$, where $V_p$ has $i$-th row $V_{p,i}' = (1, \, Z_i\tilde{H}_{p,i}', \,  (1-Z_i)\tilde{H}_{p,i}')$ for $\tilde{H}_{p,i}$ the $p$-vector with $j$-th element $(X_i - x_0)^j$. Thus, $\tilde{Y}$, $\tilde{D}$, $\tilde{Z}$ and $\tilde{\epsilon}$ are $n_h$-vectors, and $\tilde{V}_p$ is an $n_h\times (2p+1)$ matrix. The weighted regression representation becomes
\begin{equation}\label{eq simple_reg p_degree}
    \tilde{Y} = \tilde{V}_{p}\delta + \tau\tilde{D}+ \tilde{\epsilon},
\end{equation} 
with $\delta\in \Delta \subset \mbR^{2p+1}$ (where $\Delta$ is assumed to be compact), and the IV estimator of $\tau$ in \eqref{eq simple_reg p_degree} with weighted instrument vector $\tilde{Z}$ is $\hat{\tau}_{IV,p} = (z'\tilde{D})^{-1} z'\tilde{Y}$, for $z = M_{\tilde{V}_p}\tilde{Z}$.

\begin{prop} \label{prop equivalence iv}
Let Assumptions \ref{as dgp}--\ref{as finite sample} hold. Then $\hat{\tau}_{IV,p} \equiv \hat{\tau}_p$.
\end{prop}
The well-known numerical equivalence between \eqref{eq tau_hat} and \eqref{eq tau iv imbens} when using a uniform kernel and local linear regression is generalizable to any kernel function and degree of local-polynomial regression. Therefore, conventional FRD estimators including the standard local-linear triangular-kernel estimator can be represented as a single weighted local IV estimator rather than the ratio of differences of four separate local-polynomial regressions. So, $\hat{\tau}_{IV,p}$ offers a computationally simple way of computing FRD estimators, and deepens the links between the FRD and linear IV literatures. However, numerical equivalence implies that $\hat{\tau}_{IV,p}$ will also lack finite integer moments in finite samples.
\begin{corol}\label{cor tau_iv}
    Let Assumptions \ref{as dgp}--\ref{as gradient} hold and $r > 0$. Then, $\mbE[|\hat{\tau}_{IV,p}|^r] < \infty$ if and only if $r < 1$, $\mbE[|\hat{\tau}_{IV,p}|^r] = \infty$ otherwise. 
\end{corol}
Corollary \ref{cor tau_iv} follows immediately by combining Theorem \ref{thm no moments} and Proposition \ref{prop equivalence iv}, and provides an analogue to the lack of finite integer moments of the just-identified IV estimator in linear IV models with joint normal errors (see, e.g., \citet{chao2013expository}).

\subsection{The $\lambda$-class of generalized FRD estimators}
The equivalence results of the previous section show that the conventional FRD estimator admits an IV representation for general local-polynomial orders and kernel functions. This representation opens the possibility of regularizing the FRD estimator in a manner analogous to the $k$-class of linear IV estimators \citep{theil1958economic}, which contains a continuum of estimators with finite moments \citep[e.g.,][]{fuller1977some}. 

Recall $z = M_{\tilde{V}_p}\tilde{Z}$, and define $\Gamma_p = z'z$, $P_z = zz'/\Gamma_p$ and $Q = M_{\tilde{V}_p} - P_z = M_{[\tilde{V}_p,\tilde{Z}]}$. $P_z$ projects onto the instrument after partialling out the side-specific local polynomials, and $Q$ is the residual-maker matrix for the included side-specific local polynomials and the instrument. For $\lambda \in [0,1]$, define the $\lambda$-class estimator
\begin{align}\label{eq lambda class estimators}
    \hat{\tau}_{\lambda,p} = \frac{\hat{\tau}_{\lambda,p}^Y}{\hat{\tau}_{\lambda,p}^D}, \,\,\, \hat{\tau}_{\lambda,p}^W = \tilde{D}'(P_z + (1-\lambda)Q)\tilde{W}, \,\,\, W\in \{Y,D\}.
\end{align}
Estimator $\hat{\tau}_{\lambda,p}$ can be expressed in terms of the numerator $\hat{\tau}_p^Y$ and denominator $\hat{\tau}_p^D$ of $\hat{\tau}_p$.
\begin{prop} \label{prop equivalence 2sls}
Let Assumptions \ref{as dgp}--\ref{as finite sample} hold. Then, for $W \in \{Y,D\}$,
\begin{align}\label{eq tau lambda expanded}
    \hat{\tau}_{\lambda,p}^W = \lambda \Gamma_p\hat{\tau}_p^D\hat{\tau}_p^W + (1-\lambda)\tilde{D}'M_{\tilde{V}_p}\tilde{W},
\end{align}
If $\lambda = 1$, then $\hat{\tau}_{1,p} \equiv \hat{\tau}_{p}$.
\end{prop}
When $\lambda = 1$, the second term in \eqref{eq tau lambda expanded} vanishes and $\Gamma_p\hat{\tau}_p^D$ cancels, which leaves the conventional FRD estimator $\hat{\tau}_{1,p} = \hat{\tau}_p^Y/\hat{\tau}_p^D$, and $\lambda = 0$ gives the ordinary least squares estimator of $\tau$ in \eqref{eq simple_reg p_degree}, which is the conventional SRD estimator for \eqref{eq main model} under full compliance with treatment assignment, such that $Z_i = D_i$ for every $i$. The class $\hat{\tau}_{\lambda,p}$ for $\lambda\in[0,1]$ therefore defines a continuum of estimators connecting the conventional FRD and SRD estimators. 
\begin{assum}\label{as deltasample}
There exists a set $\mcN_+(\delta,\kappa) \subseteq \mcN_+$ with $|\mcN_+(\delta,\kappa)| = \max\{2p+1, 2\} =: n_p$, such that:
     (i) $|X_i - X_j| \geq \delta$ for each $i\neq j$, $i,j \in \mcN_+(\delta,\kappa)$, for some $\delta > 0$;
     (ii) $K_h(X_i - x_0) \geq \kappa$ for each $i \in \mcN_+(\delta, \kappa)$, for some $\kappa > 0$; 
     (iii) $D_i = 1$ and $D_j = 0$ for some $i, j \in \mcN_+(\delta, \kappa)$.
\end{assum}
Assumption \ref{as deltasample} requires mild structure on the observed effective sample, stating that above the cutoff, there are at least $\max\{2p+1,2\}$ individuals with distinct values $X_i$ with non-negligible kernel weights (i.e., not boundary points if the kernel is boundary-vanishing, such as the triangular or Epanechnikov kernels), and treatment variation among these individuals. Part (iii) mildly strengthens Assumption \ref{as finite sample}(iii), as two individuals with heterogeneous treatment statuses must have distinct $X_i$ values. This minimal structure will hold in scenarios of practical relevance (for $p=1$, this requires just three observations in $\mcN_+$ with distinct $X_i$, non-negligible kernel weights, and varying treatments between two of the three), and is sufficient to preserve all finite moments of the error distribution when $\lambda < 1$. This is a condition on the realized effective sample, with constants $\delta$ and $\kappa$ treated as fixed.\footnote{Probabilities, densities, and expectations are now taken under the law conditional on $\mcE'$, where $\mcE'$ is the event that the realized effective sample satisfies Assumptions \ref{as finite sample}(iv) and \ref{as deltasample}. For notational convenience, I again leave this implicit. As $\mcE' \subseteq \mcE$, the proof of Theorem \ref{thm no moments} also establishes the non-existence of finite integer moments of $\hat{\tau}_p$ conditional on $\mcE'$ when Assumption \ref{as gradient}(i) is satisfied on $\mcE'$.} 

\begin{assum}\label{as tail persistence}
For every $x \in \mcX_h$, $d \in \{0,1\}$, and $r \geq r^*$, $\mbE[|U_i|^r|X_i = x, D_i = d] = \infty$.
\end{assum}
Assumption \ref{as tail persistence} states that any heavy tails of $U_i$ persist after additionally conditioning on $D_i$, so conditioning on either the treated or untreated group does not alter the moment threshold $r^*$ from Assumption \ref{as dgp}(iii)(b). Persistence of finite moments for $r < r^*$ follows from Assumptions \ref{as dgp}(iii)(b) and \ref{as finite sample}(ii).

\begin{theorem} \label{thm lambda moments}
    Let Assumptions \ref{as dgp}--\ref{as finite sample} and \ref{as deltasample}--\ref{as tail persistence} hold, and fix $\lambda \in [0,1)$ and $r > 0$. Then,
    $\mbE[|\hat{\tau}_{\lambda,p}|^r] < \infty$ for $r < r^*$, $\mbE[|\hat{\tau}_{\lambda,p}|^r] = \infty$ otherwise.
\end{theorem}
The estimator $\hat{\tau}_{\lambda,p}$ preserves the finite moments of the error distribution, so will have finite moments of all orders with normal errors.\footnote{With a binary outcome, Assumption \ref{as dgp}(iii)(c) fails, so Theorem \ref{thm no moments} does not apply directly. Without full rigor, suppose $\eta \leq \mbP\{Y_i = 1 | X_i, D_i\} \leq 1 - \eta$ and $\eta \leq \mbP\{D_i = 1 | X_i\} \leq 1 - \eta$ for some $\eta > 0$ on $\mcX_h$. Take a symmetric configuration $(x^*,d^*)$ satisfying Assumption \ref{as gradient}(i), pick $i^* \in \mcN_+$ with $k_{p,+,i^*}(x^*) \neq 0$, and set $y^*_{i^*} = 1$ and $y^*_j = 0$ otherwise, so that $\hat{\tau}_p^Y(x^*,y^*) = k_{p,+,i^*}(x^*) \neq 0$. By continuity of the weights, $|\hat{\tau}_p^Y| \geq c > 0$ on a neighborhood of $x^*$ when $Y = y^*$, and this occurs with conditional probability at least $\eta^{n_h}$. Since $\hat{\tau}_p^D$ depends only on $(X,D)$, the argument in the proof of Lemma \ref{lem bounded away} gives $\mbP\{|\hat{\tau}_p^D| < s,\, |\hat{\tau}_p^Y| \geq c\} \geq C s$ for small $s > 0$. Therefore, $\mbE[|\hat{\tau}_p|^r] \geq c^r\,\mbE[|\hat{\tau}_p^D|^{-r}1\{|\hat{\tau}_p^Y| \geq c\}] \geq c^r Cr\int_0^{s_0} s^{-r}\,ds = \infty$ for $r \geq 1$. The finite-moment conclusion of Theorem \ref{thm lambda moments} continues to hold under Assumption \ref{as deltasample}, as $\hat{\tau}_{\lambda,p}$ is bounded for fixed $\lambda < 1$.} Fixing $\lambda < 1$ bounds the denominator away from zero and regularizes the estimator. This is because the term $\tilde{D}'M_{\tilde{V}_p}\tilde{D}$ is bounded away from zero (see Lemma \ref{lem positivity}), and is added to the non-negative term $\lambda \Gamma_p (\hat{\tau}_p^D)^2$. 

Alternative methods for restoring finite moments such as Tikhonov regularization and trimming could also be considered, but the class in \eqref{eq lambda class estimators} confers some distinct advantages. Firstly, replacing $Y$ with $Y + cD$ for some constant $c$ shifts $\hat{\tau}_p$ by exactly $c$, and $\hat{\tau}_{\lambda,p}$ and the bias-corrected estimator $\hat{\tau}_{\lambda,p,q}^{\BC}$ developed in Section \ref{sec asymptotics} also inherit this equivariance property. In contrast, Tikhonov regularization shrinks the estimator toward zero, and trimming will in general not satisfy this property if the trimming threshold binds. Secondly, the selection rule for $\lambda$ proposed in Section \ref{sec asymptotics} depends only on the observed $X$ and the choice of $K(\cdot)$ and $p$, without requiring any additional user-chosen penalty parameters, and also is unitless and scale-invariant. This selection rule is also motivated by explicit leading-order conditional bias expansions. Thirdly, this class interpolates between the FRD and SRD estimators, and provides a direct connection to the classical linear IV literature and the $k$-class estimators. 

\section{Asymptotics, bias correction and the selection of $\lambda$}\label{sec asymptotics}

The estimator $\hat{\tau}_{p}$ has well-studied asymptotic properties \citep[e.g.,][]{fan1996local,hahn2001identification,imbens2012optimal}. Only mild assumptions are needed to establish asymptotic equivalence between $\hat{\tau}_{\lambda,p}$ and $\hat{\tau}_{p}$.
\begin{assum}\label{as asymptotics}
     (i) $h \to 0$, $nh \to \infty$, and $f_X$ is continuous at $x_0$. (ii) $1-\lambda = o_p(1)$. (iii) Assumption \ref{as dgp}(iii)(b) is satisfied with $r^* \in (2,\infty]$. (iv) $1-\lambda = o_p((nh)^{-1/2})$, $nh^{2p+3} = O(1)$, and $\mu(\cdot),\pi(\cdot) \in \mcC_b^{p+1}(\mcB_{\ve}(x_0)\setminus\{x_0\})$ for some $\ve > 0$, where $\mcB_{\ve}(x_0)$ is an open neighborhood of $x_0$.
\end{assum}
The asymptotic analysis treats the bandwidths as deterministic sequences. In the simulations and empirical application, I use the MSE- and coverage-error-rate-optimal bandwidths as implemented by \texttt{rdrobust} \citep{calonico2020optimal}.

\begin{theorem} \label{thm lambda consistent}
Let Assumptions \ref{as dgp}--\ref{as kernel}, \ref{as finite sample}(i)--(ii) and \ref{as asymptotics}(i)--(iii) hold. Then,  $\hat{\tau}_{\lambda,p} - \hat{\tau}_{p} \overset{p}{\to} 0$. If Assumption \ref{as asymptotics}(iv) also holds, $\sqrt{nh}(\hat{\tau}_{\lambda,p}-\tau) = \sqrt{nh}(\hat{\tau}_p - \tau) + o_p(1)$. 
\end{theorem}
Theorem \ref{thm lambda consistent} states that $\hat{\tau}_{\lambda,p}$ has the same probability limit and asymptotic distribution as $\hat{\tau}_{p}$ under standard bandwidth and regularity conditions, given mild conditions on $\lambda$. If $nh^{2p+3} = o(1)$, then pointwise asymptotically valid confidence intervals with coverage probability $1-\alpha$ can be constructed with $\hat{\tau}_{\lambda,p} \pm z_{1-\alpha/2}\sqrt{\hat{\mbV}(\hat{\tau}_{\lambda,p})}$ for some appropriate variance estimator $\hat{\mbV}(\hat{\tau}_{\lambda,p})$, where $z_{1-\alpha/2}$ denotes the $1-\alpha/2$ critical value of the $N(0,1)$ distribution. However, these confidence intervals require undersmoothing for pointwise asymptotic validity.

\subsection{Bias-corrected estimation and inference}\label{subsec bias correction}

Bias correction is standard in the FRD literature \citep{calonico2014robust, cattaneo2020practical}. Theorem \ref{thm lambda consistent} implies that $\hat{\tau}_{\lambda,p}$ and $\hat{\tau}_p$ share the same first-order asymptotic bias.  For $W \in \{Y,D\}$, it is standard to estimate the leading smoothing bias of $\hat{\tau}_p^W$ with a $q$-th order local-polynomial regression on each side of the cutoff, for some $q \geq p+1$ and bandwidth $h_b$. Then, the bias estimator is 
\begin{align*}
    \hat{\mcB}_{p,q}^W(h) = \hat{C}_{p,+}\hat{\Phi}_{p+1,q,+}^W - \hat{C}_{p,-}\hat{\Phi}_{p+1,q,-}^W, \,\,\, \hat{C}_{p,s} = \sum_{i \in \mcN_s} k_{p, s,i}\left(\frac{X_i - x_0}{h}\right)^{p+1},
\end{align*}
where $\hat{\Phi}_{p+1,q,s}^{W} = h^{p+1}\,\hat{w}_{q,s}^{(p+1)}(x_0)/[(p+1)!]$, with $\hat{w}_{q,s}^{(p+1)}(x_0)$ the estimated $(p+1)$-th derivative of $\mu_s$ and $\pi_s$ for $W = Y$ and $W = D$, respectively. Then, the plug-in bias estimator and conventional bias-corrected estimator are given by
\begin{align}
    \hat{\mcB}_{p,q}(h) = \frac{\hat{\tau}_p^D\hat{\mcB}_{p,q}^Y(h)  - \hat{\tau}_p^Y\hat{\mcB}_{p,q}^D(h)}{(\hat{\tau}_p^D)^2}, \,\,\, \hat{\tau}_{p,q}^{\BC} = \hat{\tau}_p - \hat{\mcB}_{p,q}(h). \label{eq bias expression}
\end{align}
The conventional bias estimator $\hat{\mcB}_{p,q}(h)$ has the square of the denominator $\hat{\tau}_p^D$ of $\hat{\tau}_p$, which can be arbitrarily close to zero in finite samples. Numerical evidence in Section \ref{sec monte carlo} suggests that the tails of the sampling distribution of $\hat{\tau}_{p,q}^{\BC}$ are substantially heavier than those of $\hat{\tau}_p$, in line with the squared denominator in \eqref{eq bias expression}. Denoting $\hat{\tau}_{IV,p}^W = \Gamma_p \hat{\tau}_p^W$, multiply the numerator and denominator of $\hat{\mcB}_{p,q}(h)$ by $\Gamma_p$ such that $\hat{\mcB}_{p,q}(h) = [\hat{\tau}_{IV,p}^D\hat{\mcB}_{p,q}^Y(h) - \hat{\tau}_{IV,p}^Y \hat{\mcB}_{p,q}^D(h)]/\hat{\tau}_{1,p}^D$. Then, for $\lambda < 1$, replace $\hat{\tau}_{1,p}^D$ with the denominator $\hat{\tau}_{\lambda,p}^D$ and scale by $\lambda$ to obtain
\begin{align}\label{eq new bias estimator}
    \hat{\mcB}_{\lambda,p,q}(h) = \lambda \,\frac{\hat{\tau}_{IV,p}^D\hat{\mcB}_{p,q}^Y(h) - \hat{\tau}_{IV,p}^Y \hat{\mcB}_{p,q}^D(h)}{\hat{\tau}_{\lambda,p}^D}, \,\,\, \hat{\tau}_{\lambda,p,q}^{\BC} = \hat{\tau}_{\lambda,p} - \hat{\mcB}_{\lambda,p,q}(h).
\end{align}
At $\lambda = 1$, $\hat{\tau}_{1,p,q}^{\BC}$ is the conventional bias-corrected estimator $\hat{\tau}_{p,q}^{\BC}$. This new bias-corrected estimator retains the moment properties of the uncorrected $\lambda$-class estimator from Theorem \ref{thm lambda moments} (see Appendix \ref{sec proofs} for the additional regularity conditions on the $q$-th order regression and corresponding observed sample conditions).

\begin{theorem}\label{thm bias corrected moments}
Let Assumptions \ref{as dgp}--\ref{as finite sample}, \ref{as deltasample}--\ref{as tail persistence} and \ref{as bias correction moment} hold, fix \(\lambda \in [0,1)\) and $r > 0$. Then,
$\mbE[|\hat{\tau}_{\lambda,p,q}^{\BC}|^r] < \infty$ for $r < r^*$, $\mbE[|\hat{\tau}_{\lambda,p,q}^{\BC}|^r] = \infty$ otherwise.
\end{theorem}

\noindent For asymptotic normality and the construction of robust bias-corrected (RBC) confidence intervals, define $\bar{W}_i = W_i - \mbE[W_i|X_i]$ for $W \in \{Y,D\}$.

\begin{assum}\label{as bias}
(i) $h_b \to 0$, $nh_b \to \infty$ and $h/h_b \to c_h \in (0,\infty)$. (ii) $\mu(\cdot),\pi(\cdot) \in \mcC_b^{q+1}(\mcB_{\ve}(x_0)\setminus\{x_0\})$ for some $\ve > 0$ and some $q \geq p+1$. (iii) For some $\ve,\delta > 0$, $\sup_{x \in \mcB_{\ve}(x_0)\setminus\{x_0\}} \mbE[|\bar{Y}_i|^{2+\delta} |X_i = x] < \infty$; Assumption \ref{as dgp}(iii)(b) is satisfied for some $r^* \in (2+\delta,\infty]$, for $s \in \{+,-\}$, $\lim_{x\to x_0^s}\mbE[\bar{Y}_i^2|X_i = x]$ and $\lim_{x\to x_0^s}\mbE[\bar{Y}_i\bar{D}_i|X_i = x]$ exist, and $\lim_{x\to x_0^s} \mbE[(\bar{Y}_i - \tau\bar{D}_i)^2|X_i = x] > 0$. (iv) $0 < V_{\BC} < \infty$.
\end{assum}
Assumption \ref{as bias} imposes standard restrictions on an auxiliary bandwidth $h_b$ used to compute bias estimators with higher-order local-polynomials, as well as the minimal smoothness, moment and nondegeneracy conditions required for constructing RBC confidence intervals.

\begin{theorem}\label{thm bias corrected}
Let Assumptions \ref{as dgp}--\ref{as kernel}, \ref{as finite sample}(i)--(ii) and \ref{as asymptotics}--\ref{as bias} hold. Then, $\sqrt{nh}(\hat{\tau}_{\lambda,p,q}^{\BC} - \tau) \overset{d}{\to} N(0,V_{\BC})$.
Further, if $nh\hat{\mbV}(\hat{\tau}_{\lambda,p,q}^{\BC}) \overset{p}{\to} V_{\BC}$, then $(\hat{\tau}_{\lambda,p,q}^{\BC} - \tau)/\,\sqrt{\hat{\mbV}(\hat{\tau}_{\lambda,p,q}^{\BC})}  \overset{d}{\to} N(0,1)$.
\end{theorem}
By Theorem \ref{thm bias corrected}, $\hat{\tau}_{\lambda,p,q}^{\BC}$ and $\hat{\tau}_{p,q}^{\BC}$ share the same influence function, so $V_{\BC}$ is the asymptotic variance of the conventional bias-corrected estimator of \citet{calonico2014robust}, and does not depend on $\lambda$. An immediate corollary of Theorem \ref{thm bias corrected} is that $\hat{\tau}_{\lambda,p,q}^{\BC} \pm z_{1-\alpha/2}\sqrt{\hat{\mbV}(\hat{\tau}_{\lambda,p,q}^{\BC})}$ is a pointwise asymptotically valid $100(1-\alpha)\%$ confidence interval for $\tau$ with bandwidth chosen at the MSE-optimal rate.\footnote{For feasible variance estimation, I combine the $\lambda$-regularized first-order score with the nearest-neighbor covariance estimates of \citet{calonico2014robust}. Consistency of this variance estimator is established in the proof of Theorem \ref{thm bias corrected} under the additional conditions of \citet[Supplementary Material, Theorem A.3]{calonico2014robust}, which for brevity I do not restate here.}

\subsection{Selection of $\lambda$}
Theorems \ref{thm lambda moments} and \ref{thm bias corrected moments} require $\lambda <1$, and Theorems \ref{thm lambda consistent} and \ref{thm bias corrected} require $\sqrt{nh}(1-\lambda) \to 0$, but they provide no guidance for the choice of $\lambda$ beyond these simple conditions. Section \ref{subsec bias correction} focuses on removing the first-order smoothing bias that comes from approximating the conditional expectation functions via local-polynomial regressions. However, another source of bias stems from the ratio structure of the FRD estimator, and I motivate a selection rule for $\lambda$ to remove the leading order finite-sample ratio bias.

\begin{assum}\label{as lambda selection}
Let $\mcX_b = \mcX_h \cup \mcX_{h_b}$. (i) On $\mcX_b$, $\tau(x) \equiv \tau$, and $\mu(x)$ and $\pi(x)$ are side-specific polynomials of degree at most $p$. The function $\mu(x) - \tau\pi(x)$ is continuous at $x_0$, and $\mbE[U_i\bar{D}_i|X_i] = \sigma_{U\bar{D}}$ is constant and finite. Assumption \ref{as dgp}(iii)(b) holds for some $r^* > 4$. (ii) $nh \to \infty$, and \(\lambda \in [0, 1)\) is $X$-measurable with $1 - \lambda = O((nh)^{-1})$ uniformly on the sample events of Assumption \ref{as nagar design regularity}.
\end{assum}

\noindent This assumption imposes more structure than previous assumptions. However, this allows for a simple leading ratio-bias expression to be derived, which in turn leads to a simple selection rule for $\lambda$ that performs well in simulations. For $\mcN_b = \mcN_h \cup \mcN_{h_b}$ and $W \in \{Y,D\}$, write $k_{p,q,i}^{\BC}$ for the bias-corrected weights such that $\hat{\tau}_p^W - \hat{\mcB}_{p,q}^W =\sum_{i \in \mcN_b} k_{p,q,i}^{\BC}W_i$.

\begin{prop}\label{prop lambda selection}
Let Assumptions \ref{as dgp}--\ref{as kernel}, \ref{as lambda selection} and \ref{as nagar design regularity} hold. Then, 
\begin{align}
    \mbE[\hat{\tau}_{\lambda,p} - \tau|X] &= \frac{\sigma_{U\bar{D}}}{(\tau^D)^2\Gamma_p}\left[(1-\lambda)\traceoperator(QK) - \traceoperator(P_zK)\right] + o_p((nh)^{-1}), \label{eq nagar bias unBC} \\
    \mbE[\hat{\tau}_{\lambda,p,q}^{\BC} - \tau|X] &= \frac{\sigma_{U\bar{D}}}{(\tau^D)^2\Gamma_p}\left[(1-\lambda)\traceoperator(QK) - (2\vr_{p,q}-1)\traceoperator(P_zK)\right] + o_p((nh)^{-1}), \label{eq nagar bias BC} 
\end{align}
where $\traceoperator(\cdot)$ is the trace of a matrix, and $\vr_{p,q} = \sum_{i \in \mcN_b}k_{p,i}k_{p,q,i}^{\BC}/\sum_{i \in \mcN_b}k_{p,i}^2$.
\end{prop}
Proposition \ref{prop lambda selection} motivates the selection rules
\begin{align}\label{eq lambda selection rule}
    \lambda^* = \max\left\{1 - \frac{\traceoperator(P_zK)}{\traceoperator(QK)},0\right\}, \,\,\, \lambda_{\BC}^* = \max\left\{1 - \max\{2\vr_{p,q}-1,1\} \frac{\traceoperator(P_zK)}{\traceoperator(QK)},0\right\},
\end{align}
where the outer $\max\{\cdot,\cdot\}$ constraints ensure $\lambda^*, \lambda_{\BC}^* \in [0,1]$. The inner $\max\{\cdot,\cdot\}$ for $\lambda_{\BC}^*$ ensures $\lambda_{\BC}^* \leq \lambda^*$, as bias-correction introduces additional terms that share the near-zero denominator for the conventional $\hat{\tau}_{1,p,q}^{\BC}$. This also connects to the IV literature; when using the uniform kernel, $K$ is proportional to the $n_h$-identity matrix. Since $P_z$ has rank one and $Q$ has rank $n_h - 2(p+1)$, $\lambda^* = 1 - 1/(n_h - 2(p+1))$, which is equivalent to a Fuller-corrected version of $\hat{\tau}_{IV,p}$ on the unweighted effective sample within the bandwidth \citep{fuller1977some}. The term $\vr_{p,q}$ is the least-squares coefficient from regressing the bias-corrected equivalent-kernel weights on the conventional ones. For the triangular kernel with $p = 1$, $q = 2$ and $h = h_b$, $\vr_{p,q} \to 5/4$, so $\lambda^*_{\BC} \approx 1- \frac32(1-\lambda^*)$. With the uniform kernel and $h=h_b$, $\varrho_{p,q} = 1$ exactly. The ratio-bias justification for these rules is conditional on Assumptions \ref{as lambda selection} and \ref{as nagar design regularity}. The data-generating processes in Section \ref{sec monte carlo} do not satisfy these simplifying conditions, yet these rules lead to good estimator and confidence interval properties.

The proofs of Theorems \ref{thm lambda moments} and \ref{thm bias corrected moments} assume fixed $\lambda < 1$, as for fixed $n$ and $h$ this guarantees $\hat{\tau}_{\lambda,p}^D \geq (1-\lambda)C_{\zeta}^* > 0$, where $C_{\zeta}^* > 0$ is defined in Lemma \ref{lem positivity}. Using Cauchy--Schwarz and Lemma \ref{lem Gamma finite}, Assumptions \ref{as kernel} and \ref{as finite sample}(iv) imply that $\traceoperator(P_zK) \geq 2C_{\mcS}/n > 0$ and $\traceoperator(QK) \leq nC_K/h < \infty$, and Assumption \ref{as deltasample} additionally implies $\traceoperator(QK) > 0$. Together, $\lambda^* \leq 1 - 2C_\mcS h/(n^2 C_K) =: \bar{\lambda} < 1$. The bound becomes $\hat{\tau}_{\lambda^*,p}^D \geq (1-\bar{\lambda})C_{\zeta}^* > 0$, and the bound in \eqref{eq taulambda minus lambda} remains finite. The same argument extends to $\hat{\tau}_{\lambda_{\BC}^*,p,q}^{\BC}$ since $\lambda_{\BC}^* \leq \lambda^* \leq \bar{\lambda}$, and therefore the fixed-$\lambda$ moment theorems remain valid for the data-driven selection rules proposed in \eqref{eq lambda selection rule}.

\section{Monte Carlo simulations}\label{sec monte carlo}

I consider the model $Y_i = m(X_i) +  \tau(X_i) D_i + U_i$ with
\begin{equation*}
    m_{\textup{Lee}}(x) = 
    \begin{cases}
        0.48 + 1.27x + 7.18x^2 + 20.21x^3 + 21.54x^4 + 7.33x^5 & \textup{ if} \,\, x < 0, \\
        0.48 + 0.84x - 3.00x^2 + 7.99x^3 - 9.01x^4 + 3.56x^5 & \textup{ if} \,\, x \geq 0,
    \end{cases}
\end{equation*}
and 
\begin{equation*}
    m_{\textup{LM}}(x) = 
    \begin{cases}
        3.70 + 2.99x + 3.28x^2 + 1.45x^3 + 0.22x^4 + 0.03x^5 & \textup{ if} \,\, x < 0, \\
        3.70 + 18.49x - 54.80x^2 + 74.30x^3 - 45.02x^4 + 9.83x^5 & \textup{ if} \,\, x \geq 0,
    \end{cases}
\end{equation*}
which are calibrated to \citet{lee2008randomized} and \citet{ludwig2007does} (referred to as the Lee and LM designs henceforth). Let $\eta_{1i}$, $\eta_{2i}$, and $X_i$ be mutually independent, $\eta_{1i}, \eta_{2,i} \sim N(0,1)$ and $X_i \sim U[-1,1]$, $\ve_i = \rho\,\eta_{1i} + \sqrt{1 - \rho^2}\,\eta_{2i}$, $U_i = 0.3\eta_{1i}$, and $x_0 = 0$. For simplicity, I consider the homogeneous treatment effect $\tau(x) = 1$. The treatment is $D_i = 1\{\Phi^{-1}(\pi(X_i)) + \ve_i \geq 0 \}$, $\pi(x) = \pi_- \exp(0.2x) 1\{x < 0 \} + (\pi_+ + \pi_-(1-\exp(-0.2x)))1\{x \geq 0 \}$, and $\Phi(\cdot)$ is the standard normal cdf. I set $\rho = \pm\, 0.5$ and $\pi_+ = 0.7$, $\pi_- = 0.3$, giving $\tau^D = 0.4$. Each experiment uses $n\in \{500, 1500\}$, and 10,000 replications. Additional simulations with $\rho = \pm \, 0.2$ (weak endogenous selection) and $\rho = \pm \, 0.8$ (strong endogenous selection) are presented in the Online Appendix.

For each parameter configuration, I study the conventional $\hat{\tau}_{1}$ and new $\hat{\tau}_{\lambda^*,1}$, and the BC estimators $\hat{\tau}_{1,2}^{\BC}$ and $\hat{\tau}_{\lambda_{\BC}^*,1,2}^{\BC}$, where $\lambda^*$ and $\lambda_{\BC}^*$ are selected according to \eqref{eq lambda selection rule}. All estimators use the triangular kernel, $p=1$, and $q=2$ and $h=h_b$ where appropriate, so subscripts denoting $p$ and $q$ are dropped. All estimators and their corresponding confidence intervals are computed using the MSE- and coverage-error-rate-optimal (CER hereafter) bandwidths \citep{calonico2020optimal}, implemented with \texttt{rdbwselect}; these bandwidths are commonly used in practice, and serve as good benchmarks for performance. Bias-aware Anderson--Rubin (denoted $AR_{\BA}$) confidence sets from \citet{noack2024bias} are also reported, which build on the bias-aware inference methods of \citet{armstrong2018optimal} and \citet{armstrong2020simple}, and use a separate bandwidth calculated internally within the \texttt{RDHonest} package. All confidence intervals use nearest-neighbor variance estimators. For the estimators, I consider bias, median bias, median absolute error (MAE), and root mean squared error (RMSE). For the confidence intervals, I consider coverage (at the nominal 95\% level), median length, and mean length (for the AR confidence sets, I report coverage, median length excluding unbounded intervals and the percentage of simulation draws where the interval is unbounded).

\subsection{Results}

% =====================================================================
%  Medium endogeneity: \citet{lee2008randomized} design
% =====================================================================
\begin{table}[htbp]
\centering
\begin{threeparttable}
\caption{Medium endogeneity: \citet{lee2008randomized} design}
\label{table dgp1_rho05}
\footnotesize
\setlength{\tabcolsep}{3pt}
\begin{tabular}{l c c c c c c c c c c c c c c}
\toprule
& \multicolumn{7}{c}{MSE-optimal bandwidth} & \multicolumn{7}{c}{CER-optimal bandwidth} \\
\cmidrule(lr){2-8}\cmidrule(lr){9-15}
& \multicolumn{4}{c}{Point estimator} & \multicolumn{3}{c}{Nominal 95\% CI} & \multicolumn{4}{c}{Point estimator} & \multicolumn{3}{c}{Nominal 95\% CI} \\
\cmidrule(lr){2-5}\cmidrule(lr){6-8}\cmidrule(lr){9-12}\cmidrule(lr){13-15}
Est. & {B} & {MB} & {MAE} & {RMSE} & {Cov.} & {MdL} & {MeL} & {B} & {MB} & {MAE} & {RMSE} & {Cov.} & {MdL} & {MeL} \\
\midrule
\multicolumn{15}{l}{\underline{\emph{$\rho = 0.5$, $n = 500$, $\tilde{F}_{\MSE} = 9.8$, $\tilde{F}_{\CER} = 7.0$}}} \\[1pt]
  $\hat{\tau}_1$ & 0.06 & 0.10 & 0.20 & 2.95 & 94.0 & 1.04 & 63.4 & 0.05 & 0.08 & 0.23 & 7.24 & 95.6 & 1.25 & 245 \\
  $\hat{\tau}_{\lambda^*}$ & 0.12 & 0.12 & 0.17 & 0.25 & 90.5 & 0.85 & 0.91 & 0.12 & 0.13 & 0.18 & 0.27 & 91.5 & 0.93 & 1.01 \\
  $\hat{\tau}_1^{\BC}$ & $-33.8$ & 0.05 & 0.27 & \multicolumn{1}{c}{$>10^{3}$} & 95.9 & 1.53 & 97.5 & $-32.2$ & 0.05 & 0.32 & \multicolumn{1}{c}{$>10^{3}$} & 96.1 & 1.85 & 404 \\
  $\hat{\tau}_{\lambda^*}^{\BC}$ & 0.06 & 0.09 & 0.23 & 0.40 & 94.2 & 1.27 & 1.41 & 0.06 & 0.10 & 0.25 & 0.43 & 94.0 & 1.40 & 1.55 \\
  $\hat{\tau}_{\lambda^*_{\BC}}^{\BC}$ & 0.08 & 0.10 & 0.21 & 0.34 & 93.1 & 1.16 & 1.23 & 0.09 & 0.12 & 0.23 & 0.36 & 92.8 & 1.24 & 1.32 \\
  $AR_{\BA}$ & \multicolumn{1}{c}{--} & \multicolumn{1}{c}{--} & \multicolumn{1}{c}{--} & \multicolumn{1}{c}{--} & 96.8 & \multicolumn{1}{c}{1.61} & \multicolumn{1}{c}{58.3\%} & \multicolumn{1}{c}{--} & \multicolumn{1}{c}{--} & \multicolumn{1}{c}{--} & \multicolumn{1}{c}{--} & 96.8 & \multicolumn{1}{c}{1.61} & \multicolumn{1}{c}{58.3\%} \\
\addlinespace[3pt]
\multicolumn{15}{l}{\underline{\emph{$\rho = 0.5$, $n = 1500$, $\tilde{F}_{\MSE} = 28.9$, $\tilde{F}_{\CER} = 20.0$}}} \\[1pt]
  $\hat{\tau}_1$ & 0.07 & 0.08 & 0.12 & 0.21 & 89.9 & 0.59 & 0.71 & 0.03 & 0.06 & 0.14 & 0.94 & 93.4 & 0.72 & 10.0 \\
  $\hat{\tau}_{\lambda^*}$ & 0.08 & 0.09 & 0.12 & 0.18 & 88.1 & 0.55 & 0.59 & 0.07 & 0.07 & 0.13 & 0.19 & 91.6 & 0.66 & 0.70 \\
  $\hat{\tau}_1^{\BC}$ & 0.00 & 0.04 & 0.15 & 0.44 & 94.6 & 0.87 & 1.04 & $-4.14$ & 0.01 & 0.19 & 394 & 95.1 & 1.06 & 14.8 \\
  $\hat{\tau}_{\lambda^*}^{\BC}$ & 0.03 & 0.05 & 0.15 & 0.25 & 94.1 & 0.82 & 0.89 & 0.01 & 0.04 & 0.17 & 0.30 & 94.4 & 0.97 & 1.08 \\
  $\hat{\tau}_{\lambda^*_{\BC}}^{\BC}$ & 0.04 & 0.06 & 0.14 & 0.23 & 93.8 & 0.79 & 0.85 & 0.03 & 0.04 & 0.17 & 0.27 & 94.1 & 0.93 & 1.00 \\
  $AR_{\BA}$ & \multicolumn{1}{c}{--} & \multicolumn{1}{c}{--} & \multicolumn{1}{c}{--} & \multicolumn{1}{c}{--} & 96.7 & \multicolumn{1}{c}{1.10} & \multicolumn{1}{c}{18.8\%} & \multicolumn{1}{c}{--} & \multicolumn{1}{c}{--} & \multicolumn{1}{c}{--} & \multicolumn{1}{c}{--} & 96.7 & \multicolumn{1}{c}{1.10} & \multicolumn{1}{c}{18.8\%} \\
\addlinespace[3pt]
\multicolumn{15}{l}{\underline{\emph{$\rho = -0.5$, $n = 500$, $\tilde{F}_{\MSE} = 9.8$, $\tilde{F}_{\CER} = 7.0$}}} \\[1pt]
  $\hat{\tau}_1$ & $-0.83$ & 0.08 & 0.20 & 97.8 & 98.1 & 1.12 & \multicolumn{1}{c}{$>10^{4}$} & 0.31 & 0.06 & 0.23 & 17.9 & 97.9 & 1.31 & 915 \\
  $\hat{\tau}_{\lambda^*}$ & 0.03 & 0.02 & 0.15 & 0.24 & 95.7 & 0.88 & 0.98 & $-0.02$ & $-0.02$ & 0.16 & 0.25 & 95.1 & 0.95 & 1.06 \\
  $\hat{\tau}_1^{\BC}$ & $-997$ & 0.03 & 0.28 & \multicolumn{1}{c}{$>10^{5}$} & 96.4 & 1.65 & \multicolumn{1}{c}{$>10^{4}$} & 255 & $-0.01$ & 0.33 & \multicolumn{1}{c}{$>10^{4}$} & 96.2 & 1.94 & \multicolumn{1}{c}{$>10^{3}$} \\
  $\hat{\tau}_{\lambda^*}^{\BC}$ & 0.02 & $-0.02$ & 0.23 & 0.43 & 95.0 & 1.36 & 1.58 & $-0.02$ & $-0.06$ & 0.25 & 0.46 & 94.0 & 1.47 & 1.68 \\
  $\hat{\tau}_{\lambda^*_{\BC}}^{\BC}$ & $-0.01$ & $-0.04$ & 0.21 & 0.36 & 94.4 & 1.23 & 1.36 & $-0.05$ & $-0.08$ & 0.23 & 0.38 & 93.1 & 1.29 & 1.41 \\
  $AR_{\BA}$ & \multicolumn{1}{c}{--} & \multicolumn{1}{c}{--} & \multicolumn{1}{c}{--} & \multicolumn{1}{c}{--} & 96.7 & \multicolumn{1}{c}{1.65} & \multicolumn{1}{c}{57.9\%} & \multicolumn{1}{c}{--} & \multicolumn{1}{c}{--} & \multicolumn{1}{c}{--} & \multicolumn{1}{c}{--} & 96.7 & \multicolumn{1}{c}{1.65} & \multicolumn{1}{c}{57.9\%} \\
\addlinespace[3pt]
\multicolumn{15}{l}{\underline{\emph{$\rho = -0.5$, $n = 1500$, $\tilde{F}_{\MSE} = 29.2$, $\tilde{F}_{\CER} = 20.0$}}} \\[1pt]
  $\hat{\tau}_1$ & 0.09 & 0.08 & 0.12 & 1.79 & 97.9 & 0.64 & 7.26 & $-0.17$ & 0.06 & 0.13 & 26.3 & 98.1 & 0.76 & \multicolumn{1}{c}{$>10^{3}$} \\
  $\hat{\tau}_{\lambda^*}$ & 0.08 & 0.06 & 0.11 & 0.19 & 96.7 & 0.59 & 0.64 & 0.04 & 0.03 & 0.12 & 0.19 & 96.9 & 0.68 & 0.74 \\
  $\hat{\tau}_1^{\BC}$ & 0.51 & 0.03 & 0.15 & 39.9 & 96.6 & 0.94 & 11.0 & \multicolumn{1}{c}{$>10^{3}$} & 0.00 & 0.19 & \multicolumn{1}{c}{$>10^{5}$} & 96.0 & 1.12 & \multicolumn{1}{c}{$>10^{3}$} \\
  $\hat{\tau}_{\lambda^*}^{\BC}$ & 0.04 & 0.01 & 0.15 & 0.27 & 96.2 & 0.88 & 0.99 & 0.02 & $-0.02$ & 0.17 & 0.32 & 95.3 & 1.02 & 1.16 \\
  $\hat{\tau}_{\lambda^*_{\BC}}^{\BC}$ & 0.03 & 0.01 & 0.14 & 0.25 & 96.0 & 0.86 & 0.94 & 0.00 & $-0.03$ & 0.17 & 0.29 & 94.9 & 0.98 & 1.07 \\
  $AR_{\BA}$ & \multicolumn{1}{c}{--} & \multicolumn{1}{c}{--} & \multicolumn{1}{c}{--} & \multicolumn{1}{c}{--} & 96.9 & \multicolumn{1}{c}{1.16} & \multicolumn{1}{c}{17.2\%} & \multicolumn{1}{c}{--} & \multicolumn{1}{c}{--} & \multicolumn{1}{c}{--} & \multicolumn{1}{c}{--} & 96.9 & \multicolumn{1}{c}{1.16} & \multicolumn{1}{c}{17.2\%} \\
\bottomrule
\end{tabular}
\begin{tablenotes}[flushleft]
\footnotesize
\item \textit{Notes:} 10,000 replications per parameter configuration. Weak, medium and strong endogeneity correspond to $|\rho| = 0.2$, $0.5$, and $0.8$ respectively. The probability discontinuity at the cutoff is $\tau^D = 0.4$, $n$ is the sample size and $\tilde{F}$ is the median first-stage $F$-statistic within the bandwidth window. All estimators use the triangular kernel, $p=1$, and $q=2$ and $h=h_b$ where appropriate. $\hat\tau_1$ and $\hat\tau_1^{\BC}$ denote the conventional and bias-corrected estimators, $\hat\tau_{\lambda^*}$ and $\hat\tau_{\lambda^*}^{\BC}$ denote the conventional and bias-corrected $\lambda$-class estimators with $\lambda = \lambda^*$, and $\hat\tau_{\lambda^*_{\BC}}^{\BC}$ is the bias-corrected $\lambda$-class estimator with $\lambda = \lambda^*_{\BC}$. All are computed with MSE- and CER-optimal bandwidths. B, MB, MAE, and RMSE denote bias, median bias, median absolute error, and root mean squared error. Cov., MdL and MeL denote coverage, median length, and mean length. For $AR_{\BA}$, MdL is the median length over bounded, nonempty sets only, and so understates typical length when many sets are unbounded, MeL reports the percentage of unbounded sets. $AR_{\BA}$ uses its own bandwidth, so its results do not depend on  either the MSE or CER-optimal bandwidth rule. The smoothness bound is the \texttt{RDHonest} default rule of thumb. Values of absolute size 1000 or more are reported as $>10^k$ (or $<-10^k$), where $k$ is their order of magnitude.
\end{tablenotes}
\end{threeparttable}
\end{table}

% =====================================================================
%  Medium endogeneity: \citet{ludwig2007does} design
% =====================================================================
\begin{table}[htbp]
\centering
\begin{threeparttable}
\caption{Medium endogeneity: \citet{ludwig2007does} design}
\label{table dgp2_rho05}
\footnotesize
\setlength{\tabcolsep}{3pt}
\begin{tabular}{l c c c c c c c c c c c c c c}
\toprule
& \multicolumn{7}{c}{MSE-optimal bandwidth} & \multicolumn{7}{c}{CER-optimal bandwidth} \\
\cmidrule(lr){2-8}\cmidrule(lr){9-15}
& \multicolumn{4}{c}{Point estimator} & \multicolumn{3}{c}{Nominal 95\% CI} & \multicolumn{4}{c}{Point estimator} & \multicolumn{3}{c}{Nominal 95\% CI} \\
\cmidrule(lr){2-5}\cmidrule(lr){6-8}\cmidrule(lr){9-12}\cmidrule(lr){13-15}
Est. & {B} & {MB} & {MAE} & {RMSE} & {Cov.} & {MdL} & {MeL} & {B} & {MB} & {MAE} & {RMSE} & {Cov.} & {MdL} & {MeL} \\
\midrule
\multicolumn{15}{l}{\underline{\emph{$\rho = 0.5$, $n = 500$, $\tilde{F}_{\MSE} = 9.3$, $\tilde{F}_{\CER} = 6.6$}}} \\[1pt]
  $\hat{\tau}_1$ & 0.86 & 0.86 & 0.88 & 15.5 & 44.4 & 1.44 & 619 & 0.32 & 0.53 & 0.56 & 36.4 & 72.4 & 1.39 & \multicolumn{1}{c}{$>10^{3}$} \\
  $\hat{\tau}_{\lambda^*}$ & 0.78 & 0.74 & 0.74 & 0.88 & 21.5 & 1.03 & 1.20 & 0.47 & 0.45 & 0.45 & 0.56 & 52.9 & 0.98 & 1.10 \\
  $\hat{\tau}_1^{\BC}$ & 187 & 0.15 & 0.32 & \multicolumn{1}{c}{$>10^{4}$} & 97.7 & 2.15 & 960 & 528 & 0.08 & 0.35 & \multicolumn{1}{c}{$>10^{4}$} & 96.8 & 2.08 & \multicolumn{1}{c}{$>10^{4}$} \\
  $\hat{\tau}_{\lambda^*}^{\BC}$ & 0.24 & 0.17 & 0.27 & 0.59 & 96.3 & 1.69 & 2.14 & 0.15 & 0.12 & 0.26 & 0.50 & 94.9 & 1.53 & 1.78 \\
  $\hat{\tau}_{\lambda^*_{\BC}}^{\BC}$ & 0.23 & 0.18 & 0.26 & 0.50 & 95.7 & 1.52 & 1.81 & 0.16 & 0.14 & 0.24 & 0.42 & 93.9 & 1.34 & 1.49 \\
  $AR_{\BA}$ & \multicolumn{1}{c}{--} & \multicolumn{1}{c}{--} & \multicolumn{1}{c}{--} & \multicolumn{1}{c}{--} & 96.9 & \multicolumn{1}{c}{1.94} & \multicolumn{1}{c}{65.2\%} & \multicolumn{1}{c}{--} & \multicolumn{1}{c}{--} & \multicolumn{1}{c}{--} & \multicolumn{1}{c}{--} & 96.9 & \multicolumn{1}{c}{1.94} & \multicolumn{1}{c}{65.2\%} \\
\addlinespace[3pt]
\multicolumn{15}{l}{\underline{\emph{$\rho = 0.5$, $n = 1500$, $\tilde{F}_{\MSE} = 25.2$, $\tilde{F}_{\CER} = 17.4$}}} \\[1pt]
  $\hat{\tau}_1$ & 0.85 & 0.71 & 0.71 & 7.27 & 12.3 & 0.78 & 94.5 & 0.43 & 0.40 & 0.40 & 0.98 & 50.8 & 0.79 & 4.99 \\
  $\hat{\tau}_{\lambda^*}$ & 0.72 & 0.67 & 0.67 & 0.79 & 8.0 & 0.70 & 0.77 & 0.40 & 0.38 & 0.38 & 0.46 & 43.0 & 0.70 & 0.75 \\
  $\hat{\tau}_1^{\BC}$ & 27.8 & 0.10 & 0.19 & \multicolumn{1}{c}{$>10^{3}$} & 96.9 & 1.15 & 141 & $-0.28$ & 0.04 & 0.20 & 24.5 & 95.7 & 1.16 & 7.40 \\
  $\hat{\tau}_{\lambda^*}^{\BC}$ & 0.13 & 0.11 & 0.18 & 0.32 & 96.6 & 1.06 & 1.22 & 0.06 & 0.06 & 0.18 & 0.32 & 95.2 & 1.04 & 1.15 \\
  $\hat{\tau}_{\lambda^*_{\BC}}^{\BC}$ & 0.13 & 0.12 & 0.18 & 0.31 & 96.4 & 1.01 & 1.15 & 0.07 & 0.07 & 0.17 & 0.29 & 94.9 & 0.99 & 1.06 \\
  $AR_{\BA}$ & \multicolumn{1}{c}{--} & \multicolumn{1}{c}{--} & \multicolumn{1}{c}{--} & \multicolumn{1}{c}{--} & 96.6 & \multicolumn{1}{c}{1.51} & \multicolumn{1}{c}{39.2\%} & \multicolumn{1}{c}{--} & \multicolumn{1}{c}{--} & \multicolumn{1}{c}{--} & \multicolumn{1}{c}{--} & 96.6 & \multicolumn{1}{c}{1.51} & \multicolumn{1}{c}{39.2\%} \\
\addlinespace[3pt]
\multicolumn{15}{l}{\underline{\emph{$\rho = -0.5$, $n = 500$, $\tilde{F}_{\MSE} = 9.2$, $\tilde{F}_{\CER} = 6.7$}}} \\[1pt]
  $\hat{\tau}_1$ & 1.55 & 0.87 & 0.89 & 28.3 & 79.3 & 2.07 & \multicolumn{1}{c}{$>10^{3}$} & 0.39 & 0.52 & 0.57 & 36.9 & 95.8 & 1.91 & \multicolumn{1}{c}{$>10^{3}$} \\
  $\hat{\tau}_{\lambda^*}$ & 0.68 & 0.65 & 0.66 & 0.84 & 46.2 & 1.29 & 1.57 & 0.33 & 0.32 & 0.34 & 0.50 & 80.2 & 1.15 & 1.42 \\
  $\hat{\tau}_1^{\BC}$ & 253 & 0.10 & 0.39 & \multicolumn{1}{c}{$>10^{4}$} & 98.3 & 3.09 & \multicolumn{1}{c}{$>10^{3}$} & 108 & 0.02 & 0.40 & \multicolumn{1}{c}{$>10^{4}$} & 97.1 & 2.86 & \multicolumn{1}{c}{$>10^{3}$} \\
  $\hat{\tau}_{\lambda^*}^{\BC}$ & 0.20 & 0.05 & 0.31 & 0.80 & 97.2 & 2.40 & 3.08 & 0.07 & $-0.04$ & 0.30 & 0.65 & 95.4 & 2.04 & 2.55 \\
  $\hat{\tau}_{\lambda^*_{\BC}}^{\BC}$ & 0.14 & 0.03 & 0.28 & 0.62 & 96.8 & 2.14 & 2.55 & 0.01 & $-0.05$ & 0.27 & 0.50 & 94.7 & 1.78 & 2.06 \\
  $AR_{\BA}$ & \multicolumn{1}{c}{--} & \multicolumn{1}{c}{--} & \multicolumn{1}{c}{--} & \multicolumn{1}{c}{--} & 98.5 & \multicolumn{1}{c}{2.21} & \multicolumn{1}{c}{64.0\%} & \multicolumn{1}{c}{--} & \multicolumn{1}{c}{--} & \multicolumn{1}{c}{--} & \multicolumn{1}{c}{--} & 98.5 & \multicolumn{1}{c}{2.21} & \multicolumn{1}{c}{64.0\%} \\
\addlinespace[3pt]
\multicolumn{15}{l}{\underline{\emph{$\rho = -0.5$, $n = 1500$, $\tilde{F}_{\MSE} = 25.3$, $\tilde{F}_{\CER} = 17.3$}}} \\[1pt]
  $\hat{\tau}_1$ & 0.81 & 0.72 & 0.72 & 1.14 & 34.9 & 1.13 & 2.17 & 0.54 & 0.40 & 0.40 & 2.36 & 82.8 & 1.06 & 11.8 \\
  $\hat{\tau}_{\lambda^*}$ & 0.71 & 0.65 & 0.65 & 0.81 & 23.1 & 0.96 & 1.07 & 0.37 & 0.33 & 0.34 & 0.47 & 70.9 & 0.88 & 0.97 \\
  $\hat{\tau}_1^{\BC}$ & 0.35 & 0.09 & 0.22 & 7.94 & 98.7 & 1.66 & 3.22 & 2.76 & 0.02 & 0.23 & 175 & 97.5 & 1.56 & 18.3 \\
  $\hat{\tau}_{\lambda^*}^{\BC}$ & 0.15 & 0.07 & 0.20 & 0.46 & 98.5 & 1.51 & 1.81 & 0.07 & 0.00 & 0.21 & 0.43 & 96.8 & 1.38 & 1.67 \\
  $\hat{\tau}_{\lambda^*_{\BC}}^{\BC}$ & 0.13 & 0.06 & 0.20 & 0.41 & 98.3 & 1.44 & 1.68 & 0.04 & $-0.01$ & 0.20 & 0.37 & 96.5 & 1.30 & 1.50 \\
  $AR_{\BA}$ & \multicolumn{1}{c}{--} & \multicolumn{1}{c}{--} & \multicolumn{1}{c}{--} & \multicolumn{1}{c}{--} & 98.1 & \multicolumn{1}{c}{1.76} & \multicolumn{1}{c}{33.3\%} & \multicolumn{1}{c}{--} & \multicolumn{1}{c}{--} & \multicolumn{1}{c}{--} & \multicolumn{1}{c}{--} & 98.1 & \multicolumn{1}{c}{1.76} & \multicolumn{1}{c}{33.3\%} \\
\bottomrule
\end{tabular}
\begin{tablenotes}[flushleft]
\footnotesize
\item \textit{Notes:} 10,000 replications per parameter configuration. Weak, medium and strong endogeneity correspond to $|\rho| = 0.2$, $0.5$, and $0.8$ respectively. The probability discontinuity at the cutoff is $\tau^D = 0.4$, $n$ is the sample size and $\tilde{F}$ is the median first-stage $F$-statistic within the bandwidth window. All estimators use the triangular kernel, $p=1$, and $q=2$ and $h=h_b$ where appropriate. $\hat\tau_1$ and $\hat\tau_1^{\BC}$ denote the conventional and bias-corrected estimators, $\hat\tau_{\lambda^*}$ and $\hat\tau_{\lambda^*}^{\BC}$ denote the conventional and bias-corrected $\lambda$-class estimators with $\lambda = \lambda^*$, and $\hat\tau_{\lambda^*_{\BC}}^{\BC}$ is the bias-corrected $\lambda$-class estimator with $\lambda = \lambda^*_{\BC}$. All are computed with MSE- and CER-optimal bandwidths. B, MB, MAE, and RMSE denote bias, median bias, median absolute error, and root mean squared error. Cov., MdL and MeL denote coverage, median length, and mean length. For $AR_{\BA}$, MdL is the median length over bounded, nonempty sets only, and so understates typical length when many sets are unbounded, MeL reports the percentage of unbounded sets. $AR_{\BA}$ uses its own bandwidth, so its results do not depend on  either the MSE or CER-optimal bandwidth rule. The smoothness bound is the \texttt{RDHonest} default rule of thumb. Values of absolute size 1000 or more are reported as $>10^k$ (or $<-10^k$), where $k$ is their order of magnitude.
\end{tablenotes}
\end{threeparttable}
\end{table}

The implications of Theorems \ref{thm no moments} and \ref{thm lambda moments} are visible in Tables \ref{table dgp1_rho05} and \ref{table dgp2_rho05}. The heavy tails of the sampling distribution of $\hat{\tau}_1$ are evident, with RMSE values substantially larger than those of $\hat{\tau}_{\lambda}$, with the difference especially noticeable in the negative correlation design. In the positive correlation design, $\hat{\tau}_1$ happens to have a small bias, but as $\hat{\tau}_1$ has no finite mean, this is an artifact of the particular simulation draws rather than a reliable property of the estimator. The new $\hat{\tau}_{\lambda}$ controls RMSE across designs, sample sizes, and endogeneity signs. Over all parameter configurations, $\hat{\tau}_{\lambda}$ weakly dominates $\hat{\tau}_1$ for MAE. Median bias between the estimators is similar, with $\hat{\tau}_1$ having the slightly lower values with positive endogeneity, and vice versa with negative endogeneity, in the Lee designs. In the LM designs however, $\hat{\tau}_{\lambda}$ improves on $\hat{\tau}_1$ for this metric. Confidence intervals for both estimators undercover severely in the LM design; while the undercoverage of $\hat{\tau}_1$ when smoothing bias is large is well known \citep[e.g.,][]{calonico2014robust}, this table gives clear evidence that the problem extends to the $\lambda$-class, arguably to an even more serious extent. This clearly motivates the bias-corrected extension of the $\lambda$-class presented in Section \ref{sec asymptotics}.

Although the moments of $\hat{\tau}_1^{\BC}$ are not formally treated in this paper, Tables \ref{table dgp1_rho05} and \ref{table dgp2_rho05} provide numerical evidence that the heavy tails of the sampling distribution of the conventional FRD estimator become more severe after bias correction. Across both tables, only the Lee design with $\rho = 0.5$, $n=1500$ and MSE-optimal bandwidth gives an RMSE value on the same order as the $\lambda$-class, but is still substantially larger. The RMSE values for bias-corrected $\lambda$-class are consistent with the finite-moment results of Theorem \ref{thm bias corrected moments}. When looking at median bias as a robust metric of the sampling distribution location, all bias-corrected estimators are in general effective at appropriately centering the distribution, but both $\hat{\tau}_{\lambda^*}^{\BC}$ and $\hat{\tau}_{\lambda^*_{\BC}}^{\BC}$ weakly dominate $\hat{\tau}_1^{\BC}$ for MAE at the reported precision (and strictly dominate at greater precision), the outlier-robust alternative to RMSE. The estimators $\hat{\tau}_{\lambda^*}$ and $\hat{\tau}_{\lambda^*_{\BC}}^{\BC}$ provide a combination of moment control and good median absolute error, and therefore should be considered the estimators of choice for researchers prioritizing accurate point estimation, depending on whether or not bias-correction is desired. 

All bias-corrected confidence intervals have in general good coverage, with unsurprisingly the best coverage coming with the CER-optimal bandwidth for all three confidence intervals, at the cost of interval length. With the CER-optimal bandwidth, the $\hat{\tau}_{\lambda^*}^{\BC}$ and $\hat{\tau}_{\lambda^*_{\BC}}^{\BC}$ confidence intervals both have median absolute coverage error (MACE) at 1.0 percentage points (pp) over all designs, compared to 1.5pp for $\hat{\tau}_1^{\BC}$. However, the worst case coverage is worse for $\hat{\tau}_{\lambda^*}^{\BC}$ and $\hat{\tau}_{\lambda^*_{\BC}}^{\BC}$ than $\hat{\tau}_1^{\BC}$, at 90.8\%, 88.6\% and 93.2\%, respectively, with these values all appearing in the high endogeneity design. The $\hat{\tau}_1^{\BC}$ is more conservative than the other confidence intervals, with median coverage at 96.2\% with the CER-optimal bandwidth, compared to 95.2\% and 94.8\% for $\hat{\tau}_{\lambda^*}^{\BC}$ and $\hat{\tau}_{\lambda^*_{\BC}}^{\BC}$, respectively. The mean lengths of the $\hat{\tau}_1^{\BC}$ confidence intervals suggest that they are occasionally uninformatively wide. For median length, $\hat{\tau}_{\lambda^*}^{\BC}$ and $\hat{\tau}_{\lambda^*_{\BC}}^{\BC}$ provide an average of 18\% and 25\% reduction in median length across all designs; in the low and medium endogeneity designs, these length reductions are complemented by MACEs of 0.6pp and 0.9pp respectively, compared to a 1.8pp MACE for $\hat{\tau}_1^{\BC}$ in the same designs. With the MSE-optimal bandwidth, the confidence intervals have MACEs similar to one another at 1.6pp, 1.7pp and 1.65pp for $\hat{\tau}_{\lambda^*}^{\BC}$, $\hat{\tau}_{\lambda^*_{\BC}}^{\BC}$ and $\hat{\tau}_1^{\BC}$, respectively. The bias-aware Anderson-Rubin confidence sets \citep{noack2024bias} are robust to identification strength, and they have coverage between 96.1\% and 98.8\% in all designs. By construction, this comes at the cost of informativeness, as these sets are frequently unbounded, with at least 16.2\% of sets unbounded in any given design.\footnote{In unreported simulations, I consider the smoothness bound ROT2 from \citet{noack2024bias}. These confidence sets can have substantially shorter lengths, but coverage can also collapse in high-curvature models such as the LM designs.} These confidence sets are a sensible complement for researchers to report in conjunction with bias-corrected confidence intervals for the $\lambda$-class.

\section{Empirical application} \label{sec empirical}

Eligibility thresholds often lead to FRD settings in health economics, where certain health criteria determine treatment assignment, but doctors and patients may retain discretion over realized treatment receipt. I illustrate the new $\lambda$-class estimators and confidence intervals using data from \citet{bor2017treatment}, studying HIV treatment and retention in care in South Africa.\footnote{The data were obtained from the replication materials of \citet{cattaneo2023guide}, available at \url{https://github.com/rdpackages-replication/CKT_2023_SIM}.} The treatment of interest here is the initiation of antiretroviral therapy (ART) within six months of the first measurement of the CD4 count of the patient, with treatment being offered to patients with a count below 350/$\mu$L.\footnote{CD4 counts constitute a continuous running variable in principle, but rounding in reporting leads to a discrete variable. The MSE- and CER-optimal bandwidth windows contain 246 and 162 unique CD4 count values. Further, the mass point adjustment from \texttt{rdrobust} changes the selected bandwidths by less than 2\% and has negligible effect on reported point estimates and standard errors. Therefore, the application illustrates the sensitivity of the estimators and confidence intervals to bandwidth choices on a densely populated discrete score.} Continued engagement with care is important for realizing the benefits of treatment, making its relationship with treatment access an economically relevant question about healthcare utilization. The original study uses eligibility as an instrument for ART initiation when studying retention at 12 months, which I extend by studying retention at 18--24 months as an outcome.\footnote{Retention here indicates a recorded CD4 or viral-load test, or ART initiation, during that interval.}

\begin{table}[htbp]
\centering
\begin{threeparttable}
\caption{Effect of ART initiation on retention in HIV care at 18--24 months}
\label{table hiv main}
\footnotesize
\setlength{\tabcolsep}{12pt}
\begin{tabular}{l c c c c}
\toprule
& \multicolumn{2}{c}{MSE-optimal bandwidth} & \multicolumn{2}{c}{CER-optimal bandwidth} \\
\cmidrule(lr){2-3}\cmidrule(lr){4-5}
Estimator & Estimate & 95\% CI & Estimate & 95\% CI \\
\midrule
$\hat{\tau}_1$ & 0.487 & [-0.079, 1.053] & 0.456 & [-0.375, 1.287] \\
$\hat{\tau}_{\lambda^*}$ & 0.465 & [-0.040, 0.971] & 0.418 & [-0.241, 1.077] \\
$\hat{\tau}_1^{\BC}$ & 0.481 & [-0.361, 1.322] & 0.789 & [-0.473, 2.052]\\
$\hat{\tau}_{\lambda^*}^{\BC}$ & 0.460 & [-0.298, 1.217] & 0.689 & [-0.327, 1.704] \\
$\hat{\tau}_{\lambda^*_{\BC}}^{\BC}$ & 0.451 & [-0.271, 1.173] & 0.651 & [-0.272, 1.573] \\
$AR_{\BA}$ & & $(-\infty, \, \infty)$ & & \\
\midrule
Bandwidth $h = h_b$ & \multicolumn{2}{c}{123.01} & \multicolumn{2}{c}{81.94} \\
Effective sample size $n_{eff}$ & \multicolumn{2}{c}{1,357} & \multicolumn{2}{c}{895} \\
First-stage $F$ & \multicolumn{2}{c}{9.402}   & \multicolumn{2}{c}{4.142} \\
\bottomrule
\end{tabular}
\begin{tablenotes}[flushleft]
\footnotesize
\item \textit{Notes:} The outcome is retention in HIV care at 18--24 months and the treatment is ART initiation within six months. Effects are differences in probabilities. All estimators use the triangular kernel, $p = 1$, and $q = 2$ and $h = h_b$ where appropriate, and nearest-neighbor standard errors. MSE- and CER-optimal bandwidths are used, and $AR_{\BA}$ uses its own FLCI-selected bandwidth, $h_{\textup{AR}}=99.82$. $\hat{\tau}_1$ and $\hat{\tau}_1^{\BC}$ denote the conventional and bias-corrected FRD estimators, with a robust bias-corrected interval for $\hat\tau_1^{\BC}$. $\hat{\tau}_{\lambda^*}$ and $\hat{\tau}_{\lambda^*}^{\BC}$ denote the conventional and bias-corrected $\lambda$-class estimators with $\lambda = \lambda^*$, and $\hat{\tau}_{\lambda^*_{\BC}}^{\BC}$ the bias-corrected $\lambda$-class estimator with $\lambda = \lambda^*_{\BC}$. The first-stage $F$-statistic is computed within the bandwidth window and $n_{eff}$ is the effective sample size.
\end{tablenotes}
\end{threeparttable}
\end{table}

Table \ref{table hiv main} shows the estimators and confidence intervals considered in Section \ref{sec monte carlo}, with the MSE- and CER-optimal bandwidths. The CER-optimal bandwidth selects a smaller window than the MSE-optimal bandwidth, and therefore leads to a smaller effective sample size and a lower $F$-statistic. At the MSE-optimal bandwidth, both the uncorrected and bias-corrected point estimates are quite similar. Here, the new bias-corrected $\lambda$-class confidence intervals are approximately 10--14\% shorter than the conventional BC confidence intervals. Bias-correction has a more substantial impact at the CER-optimal bandwidth, where the point estimate jumps from 0.456 to 0.789 for $\hat{\tau}_1^{\BC}$, and from 0.418 to 0.689 and 0.651 for $\hat{\tau}_{\lambda^*}^{\BC}$ and $\hat{\tau}_{\lambda_{\BC}^*}^{\BC}$, respectively. The reductions in confidence interval length are more substantial at approximately 20--27\%. Nevertheless, all reported bias-corrected intervals include zero, so regularization does not change the qualitative conclusion concerning the treatment effect. Estimated correlations between the structural and first-stage residuals are $-0.26$ and $-0.24$. Simulation results suggest that $\hat{\tau}_{\lambda_{\BC}^*}^{\BC}$ has good coverage properties at these weaker levels of endogeneity, and motivate exploring both positively and negatively signed endogeneity in the simulations. Since these are estimated correlations, the true sign and magnitude of these correlations is uncertain.

\begin{figure}[ht!] 
    \caption{Bandwidth sensitivity of point estimates and confidence intervals}
    \centering
    \subfloat {{\includegraphics[width=16.5cm]{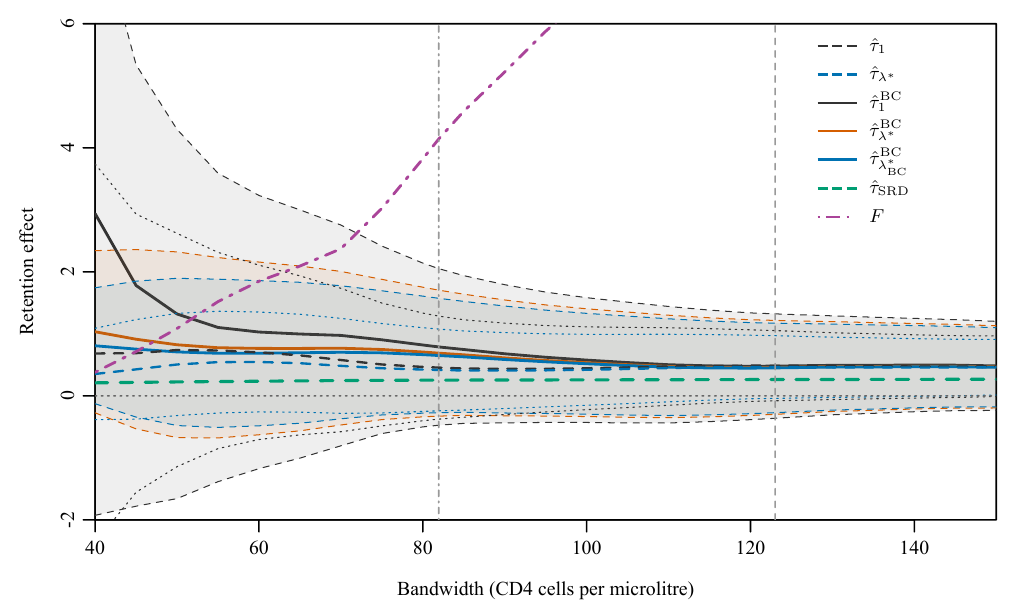} }}
    \label{figure bandwidth sensitivity}
\end{figure}

Figure \ref{figure bandwidth sensitivity} shows bandwidth sensitivity of the point estimates and confidence intervals. The uncorrected point estimates are relatively stable across the bandwidths considered, although the confidence intervals for $\hat{\tau}_1$ become markedly wider as the bandwidth (and therefore the $F$-statistic) decreases. The conventional bias-corrected estimator shows a strong sensitivity as bandwidth decreases. The new bias-corrected $\lambda$-class estimators demonstrate better stability with tighter confidence intervals as bandwidth decreases. This stability is potentially important for two reasons; first, researchers frequently consider different bandwidths as a sensitivity check, and second, these bandwidth selectors are data-driven, and therefore dependent on the observed sample.\footnote{A small patient-level pairs bootstrap $(B=999)$ gives 95\% bootstrap confidence intervals for the MSE- and CER-optimal bandwidths of $[77.53, 150.44]$ and $[51.65,100.22]$, suggesting that bandwidths at which $\hat{\tau}_1^{\BC}$ exhibits substantial sensitivity are plausible in this application.}

\section{Conclusion}

This paper shows that the conventional FRD estimator can perform poorly in finite samples, owing to the non-existence of finite integer moments. By construction, the $\lambda$-class bounds the denominator away from zero, so the new estimators preserve the moments of the error distribution. This class is computationally simple and delivers significant improvements in estimator performance and inference over conventional FRD procedures, particularly with small samples or a small treatment probability discontinuity. Bias-corrected confidence intervals provide substantial simultaneous improvements in median absolute coverage error and median interval length. Across a large parameter grid in simulations, the bias-corrected $\lambda$-class estimator with $\lambda = \lambda_{\BC}^*$ gives the best overall performance, considering the trade-off between estimator and confidence interval performance, and therefore is recommended as a sensible default choice for practice.

An important avenue for future research is to explore how $\lambda$ may be selected to optimize some desired criterion such as root mean squared error or coverage error. A second natural extension is to explore finite-sample bandwidth selection that can exploit the finite-sample properties of the $\lambda$-class estimators, instead of employing standard bandwidth selection algorithms that are based on asymptotic results for the conventional FRD estimator. The generalized IV representation of conventional FRD estimators and the development of the $\lambda$-class estimators are also of interest, as they create a direct bridge to the well-studied and mature linear IV literature, with many tools and results potentially being useful in FRD contexts. Extensions to clustered errors, panel settings, or FRD designs with multiple cutoffs would also widen the applicability of the results of the paper.

\vspace{5mm}

{\small \singlespacing
\noindent \textbf{Declaration of AI use in the manuscript preparation process}

\noindent The author used ChatGPT (OpenAI) and Claude (Anthropic) to assist with mathematical arguments, refining the exposition, developing computational code, and checking the consistency and clarity of the manuscript. The author critically evaluated, reviewed and revised the content as appropriate, and takes full responsibility for the contents of the manuscript.
}

\bibliographystyle{apalike}
\bibliography{references}

\appendix

\setcounter{assum}{0}
\renewcommand{\theassum}{A.\arabic{assum}}
\def\theHassum{A.\arabic{assum}}

\section{Proofs}\label{sec proofs}

\begin{proof}[Proof of Theorem \ref{thm no moments}]
From \eqref{eq main model} and \eqref{eq main estimators}, for $s \in \{+,-\}$
\begin{align*}
    \hat{\mu}_{p,s}(x_0) = \displaystyle\sum_{i\in \mcN_s}\left[m(X_i) + \tau(X_i) D_i\right]k_{p,s,i} +  \displaystyle\sum_{i\in \mcN_s} U_i k_{p,s,i} =: \tilde{G}_s + \tilde{U}_s.
\end{align*}
Therefore, $\hat{\tau}_p^Y = \mu_{\hat{\tau}_p^Y} + (\tilde{U}_+ - \tilde{U}_-)$, where $\mu_{\hat{\tau}_p^Y} = (\tilde{G}_+ - \tilde{G}_-)$ is constant conditional on $(X,D)$. As $m(\cdot)$ and $\tau(\cdot)$ are continuous over $\mcX_h$ by Assumption \ref{as finite sample}(i), and $\mcX_h$ is compact, denote $C_M = \sup_{x \in \mcX_h} |m(x)| < \infty$ and $C_T = \sup_{x \in \mcX_h} |\tau(x)| < \infty$. Therefore, $|\mu_{\hat{\tau}_p^Y}| \leq 2n(C_M + C_T)C_k =: C_{\mu} < \infty$ (where the bound follows as $|\tilde{G}_s| \leq (C_M + C_T) \sum_{i \in \mcN_s}|k_{p,s,i}| \leq n(C_M+C_T)C_k$ by Lemma \ref{lem gradient upper}, with $|k_{p,i}| \leq C_KC_{\omega}/h =: C_k$).

Fix $i^* \in \mcN_+$ as the smallest index such that $|k_{p,+,i^*}| \geq 1/n$ (where such $i^*$ exists as $\sum_{i \in \mcN_+}k_{p,+,i} = 1$ by Assumptions \ref{as kernel} and \ref{as finite sample}, and since the weights are measurable functions of $X$, so too is $i^*$), and write $\hat{\tau}_p^Y = k_{p,+,i^*}U_{i^*} + V$, where $V$ is fixed conditional on $(X,D,U_{-i^*})$, so a change of variables gives
\begin{align*}
    f_{\hat{\tau}_p^Y|X,D,U_{-i^*}}(t_y|X,D,U_{-i^*}) = \frac{1}{|k_{p,+,i^*}|}f_{U_{i^*}|X,D,U_{-i^*}}\left(\frac{t_y - V}{k_{p,+,i^*}}\bigg| X,D,U_{-i^*}\right).
\end{align*}
Fix $\delta_y, \delta_w > 0$, define $\delta_u = n(\delta_y + C_{\mu} + \delta_w)$, and set $b= \min\{\delta_w/(nC_k),\delta_u\}$, and let $\mcA= \bigcap_{j \in\mcN_h \setminus \{i^*\}} \{|U_j| \leq b\}$. On event $\mcA$ and for $W= V - \mu_{\hat{\tau}_p^Y}$, $|W| \leq (n-1)C_kb \leq \delta_w$, so $|V| \leq C_{\mu} + \delta_w$, and for $|t_y| \leq \delta_y$, $|(t_y - V)/k_{p,+,i^*}| \leq \delta_u$. By Assumption \ref{as dgp}(iii)(c), $C_{\delta_u} = \inf_{[-\delta_u,\delta_u] \times \mcX} f_{U|X}(u|x) > 0$, so Lemma \ref{lem fu conditional}(i) at $B = \delta_u$ gives for every $j$, $f_{U_{j}|X,D,U_{-j}}(u|X,D,U_{-j}) \geq \pi_{\delta_u}C_{\delta_u}$ for $|u| \leq \delta_u$, and hence for $j = i^*$, it follows that on $\mcA$, for every $|t_y| \leq \delta_y$, $f_{\hat{\tau}_p^Y|X,D,U_{-i^*}}(t_y|X,D,U_{-i^*}) \geq \pi_{\delta_u}C_{\delta_u} / |k_{p,+,i^*}| \geq \pi_{\delta_u}C_{\delta_u}/C_k >0.$ 

Then, Lemma \ref{lem fu conditional}(iii) implies that the $U_j$ are mutually independent conditional on $(X,D)$, with $f_{U_j|X,D} = f_{U_j|X,D,U_{-j}}$, and so with $b \leq \delta_u$ by construction,
\begin{align*}
    \mbP\{\mcA| X,D\} = \prod_{j\neq i^*} \mbP\{|U_j| \leq b|X, D\} \geq (2b \pi_{\delta_u} C_{\delta_u})^{n_h-1}=: \delta_{\mcA} > 0
\end{align*}
uniformly over admissible $(X,D)$. Therefore, integrating over $U_{-i^*}$ gives
\begin{align*}
    f_{\hat{\tau}_p^Y|X,D}(t_y|X,D) \geq \frac{\pi_{\delta_u}C_{\delta_u}}{|k_{p,+,i^*}|} \mbP\{\mcA| X,D\} \geq \frac{\pi_{\delta_u}C_{\delta_u}}{C_k} \delta_{\mcA} =: \eta_y > 0,
\end{align*}
for $|t_y| \leq \delta_y$, uniformly over admissible $(X,D)$. Since $\hat{\tau}_p^D$ is $(X,D)$-measurable, 
\begin{align*}
    f_{\hat{\tau}_p^Y|\hat{\tau}_p^D}(t_y|t_d) = \mbE\left[f_{\hat{\tau}_p^Y|X,D}(t_y|X,D)| \hat{\tau}_p^D = t_d\right] \geq \eta_y,
\end{align*}
for every $|t_y| \leq \delta_y$ and almost every $t_d$.

By Lemma \ref{lem bounded away}, $f_{\hat{\tau}_p^D}(t_d) \geq \underline{C}_{\tau} > 0$ for $|t_d| < \delta_\tau$, and so $f_{\hat{\tau}_p^Y,\hat{\tau}_p^D}(t_y,t_d) \geq \eta_y \underline{C}_{\tau}$ for $|t_y| < \delta_y$ and  almost every $|t_d| < \delta_\tau$. So, for the $r$-th moment with $r\geq 1$, 
\begin{align}\label{eq tau diverging}
    \mbE[|\hat{\tau}_p|^r]  \geq \int_{-\delta_\tau}^{\delta_\tau} \int_{\delta_y/2\leq |t_y| \leq \delta_y} \left|\frac{t_y}{t_d} \right|^r \eta_y \underline{C}_{\tau}  dt_y dt_d
    \geq  \eta_y \underline{C}_{\tau}\delta_y (\delta_y/2)^r \int_{-\delta_\tau}^{\delta_\tau}|t_d|^{-r}dt_d = \infty,
\end{align}
which diverges since $\int_{-a}^{a}|v|^{-r}dv \geq \int_0^{a}v^{-r}dv = \infty$ for any $a > 0$ if $r\geq 1$. 

Now consider $r < 1$. Then,
\begin{align*}
    \mbE\left[|\hat{\tau}_p^D|^{-r}\right] \leq \int_{-\delta_d}^{\delta_d}|t_d|^{-r}f_{\hat{\tau}_p^D}(t_d)dt_d + \int_{|t_d| \geq \delta_d}|t_d|^{-r}f_{\hat{\tau}_p^D}(t_d)dt_d 
    \leq \frac{2C_D\delta_d^{1-r}}{1-r} + \delta_d ^{-r} < \infty,
\end{align*}
where the first integral converges since $r <1$ and $f_{\hat{\tau}_p^D}$ is bounded near zero by Assumption \ref{as gradient}(ii), and the second converges as $|t_d|^{-r} \leq \delta_d^{-r}$ for $|t_d| \geq \delta_d$. By Lemma \ref{lem fu conditional}(ii) and Assumption \ref{as dgp}(iii)(b),
\begin{align*}
    \mbE[|U_i|^r|X,D,U_{-i}] \leq C_B \mbE[|U_i|^r|X_i] \leq C_B \sup_{x \in \mcX}\mbE[|U_i|^r|X_i = x] < \infty.
\end{align*}
Integrating over $U_{-i}$ gives the same bound for $\mbE[|U_i|^r|X,D]$ uniformly over $(X,D)$. As $\hat{\tau}_p^Y$ is a finite weighted sum of the $U_i$ and $\mu_{\hat{\tau}_p^Y}$,
\begin{align}\label{eq numerator conditional bound}
    \mbE[|\hat{\tau}_p^Y|^{r}|X,D] &\leq 2^{(r-1)_+}\left[C_{\mu}^r + n^{1+(r-1)_+}C_k^rC_B \sup_{x \in \mcX} \mbE[|U_i|^r| X_i = x]\right] =: C_Y < \infty
\end{align}
for $(r-1)_+ = \max\{r-1,0\}$, uniformly over $(X,D)$ for $r < r^*$. Therefore, using the $(X,D)$-measurability of $\hat{\tau}_p^D$ and $\mbE[|\hat{\tau}_p^Y|^{r}|X,D] \leq C_Y < \infty$ uniformly over $(X,D)$ for $r < r^*$ from \eqref{eq numerator conditional bound}, $\mbE[|\hat{\tau}_p|^r]  = \mbE[|\hat{\tau}_p^D|^{-r}\, \mbE[|\hat{\tau}_p^Y|^r|X,D]] \leq C_{Y}\,\mbE[|\hat{\tau}_p^D|^{-r}]< \infty$. Combined with \eqref{eq tau diverging}, this establishes that $\mbE[|\hat{\tau}_p|^r] < \infty$ if and only if $r < 1$. \QEDB
\end{proof}

\begin{proof}[Proof of Proposition \ref{prop equivalence iv}]
Consider $\hat{\tau}_{IV,p}^W = Z'KW - Z'KV_p(V_p'KV_p)^{-1}V_p'KW$ from the generalized IV representation. Write 
\begin{align*}
    V_p'KV_p =    
    \begin{pmatrix}
        S_0 & \mcR_{p,+}' & \mcR_{p,-}' \\
        \mcR_{p,+} & \mcU_{p,+} & 0_{p\times p} \\
        \mcR_{p,-} & 0_{p \times p} & \mcU_{p,-} \\
    \end{pmatrix},
\end{align*}
where for $s \in \{+,-\}$, $\mcR_{p,s}$ is the $p$-vector with $j$-th element $S_{j,s}=\sum_{i\in \mcN_s}K_h(X_i - x_0)(X_i - x_0)^j$, $\mcU_{p,s}$ is the $p\times p$ matrix with $ij$-th element $S_{i+j,s}$, and $0_{p \times p}$ is the $p \times p$ matrix of zeros. For $p = 0$, $\mcR_{p,s}$ and $\mcU_{p,s}$ are empty, and all terms involving them are interpreted as zero henceforth, and the block expressions below reduce to their scalar counterparts throughout the Appendix. By iterating the formula for the partitioned inverse of block matrices, $(V_p'KV_p)^{-1}$ is the matrix
\begin{align*}
    \begin{pmatrix}
        \mcG_p^{-1} & 
        -\mcG_p^{-1} \mcR_{p,+}'\mcU_{p,+}^{-1} &
        -\mcG_p^{-1}\mcR_{p,-}'\mcU_{p,-}^{-1} \\
        -\mcU_{p,+}^{-1}\mcR_{p,+}\mcG_p^{-1} & \mcU_{p,+}^{-1}+\mcU_{p,+}^{-1}\mcR_{p,+}\mcG_p^{-1}\mcR_{p,+}'\mcU_{p,+}^{-1} & \mcU_{p,+}^{-1}\mcR_{p,+}\mcG_p^{-1}\mcR_{p,-}'\mcU_{p,-}^{-1} \\
        -\mcU_{p,-}^{-1}\mcR_{p,-}\mcG_p^{-1} & \mcU_{p,-}^{-1}\mcR_{p,-}\mcG_p^{-1}\mcR_{p,+}'\mcU_{p,+}^{-1} &\mcU_{p,-}^{-1}+\mcU_{p,-}^{-1}\mcR_{p,-}\mcG_p^{-1}\mcR_{p,-}'\mcU_{p,-}^{-1}
    \end{pmatrix},
\end{align*}
for $\mcG_p = S_0 - \mcR_{p,+}'\mcU_{p,+}^{-1}\mcR_{p,+} - \mcR_{p,-}'\mcU_{p,-}^{-1}\mcR_{p,-}$ (as a principal submatrix of $\mcS_{p,s}$, $\mcU_{p,s}$ is also positive-definite and so invertible). Express $Z'KV_p$ as $Z'KV_p = \begin{pmatrix} S_{0,+} & \mcR_{p,+}' & 0_p \end{pmatrix}$, where $0_p$ is the $p$-vector of zeros, and denote $\Upsilon_{p,s} = \mcR_{p,s}'\mcU_{p,s}^{-1}\mcR_{p,s}$ for $s \in \{+,-\}$. Then,
\begin{align*}
     (V_p'KV_p)^{-1}V_p'KZ = 
    \begin{pmatrix}
        S_{0,+} \mcG_p^{-1} - \mcG_p^{-1}\Upsilon_{p,+} \\
        -S_{0,+} \mcU_{p,+}^{-1}\mcR_{p,+}\mcG_p^{-1} + \mcU_{p,+}^{-1}\mcR_{p,+} +   \mcU_{p,+}^{-1}\mcR_{p,+} \mcG_p^{-1}\Upsilon_{p,+} \\
        -S_{0,+} \mcU_{p,-}^{-1}\mcR_{p,-}\mcG_p^{-1} + \mcU_{p,-}^{-1}\mcR_{p,-}\mcG_p^{-1}\Upsilon_{p,+} 
    \end{pmatrix}.
\end{align*}
Further, $(V_p'KW)' = \begin{pmatrix}\sum_{i \in \mcN_h} W_iK_h(X_i - x_0) & (\mcK_{p,+}^W)' & (\mcK_{p,-}^W)'\end{pmatrix}$, where $\mcK_{p,s}^W$ is the $p$-vector with $j$-th element $\sum_{i\in \mcN_s}W_iK_h(X_i - x_0)(X_i - x_0)^j$. Then,
\begin{align}\label{eq ZMvpD main}
    \hat{\tau}_{IV,p}^W =& \displaystyle \sum_{i\in \mcN_+}W_iK_h(X_i - x_0)
    - \left(S_{0,+} \mcG_p^{-1} - \Upsilon_{p,+}\mcG_p^{-1}\right) \displaystyle \sum_{i\in \mcN_h}W_iK_h(X_i - x_0) \notag \\
    &+ S_{0,+} \mcG_p^{-1}\mcR_{p,+}'\mcU_{p,+}^{-1}\mcK_{p,+}^W - \mcR_{p,+}'\mcU_{p,+}^{-1}\mcK_{p,+}^W - \Upsilon_{p,+} \mcG_p^{-1} \mcR_{p,+}'\mcU_{p,+}^{-1}\mcK_{p,+}^W \notag \\
    &+ S_{0,+} \mcG_p^{-1}\mcR_{p,-}'\mcU_{p,-}^{-1}\mcK_{p,-}^W - \Upsilon_{p,+} \mcG_p^{-1} \mcR_{p,-}'\mcU_{p,-}^{-1}\mcK_{p,-}^W.
\end{align}
Consider now the analogue object for the standard formulation $\hat{\tau}_p^W = e_{p,1}'[\mcS_{p,+}^{-1}\mcH_{p,+}'K_+W_+ - \mcS_{p,-}^{-1}\mcH_{p,-}'K_-W_-]$ of $\hat{\tau}_p$, where
\begin{align*}
    \mcS_{p,s} = 
    \begin{pmatrix}
        S_{0,s} & \mcR_{p,s}' \\
        \mcR_{p,s} & \mcU_{p,s} \\
    \end{pmatrix}, \,\,\,\,
    \mcH_{p,s}'K_sW_s = 
    \begin{pmatrix}
        \displaystyle \sum_{i \in \mcN_s} W_iK_h(X_i - x_0) \\
        \mcK_{p,s}^W
    \end{pmatrix},
\end{align*}
for $s \in\{+,-\}$. Denote $\Gamma_{p,s} = S_{0,s} - \Upsilon_{p,s}$, then
\begin{align*}
    e_{p,1}'\mcS_{p,s}^{-1}\mcH_{p,s}'K_sW_s &= \Gamma_{p,s}^{-1}\displaystyle \sum_{i \in \mcN_s} W_iK_h(X_i - x_0) -\Gamma_{p,s}^{-1} \mcR_{p,s}'\mcU_{p,s}^{-1}\mcK_{p,s}^W = \Gamma_{p,s}^{-1}\mcW_{p,s},
\end{align*}
for $\mcW_{p,s}^W = \sum_{i \in \mcN_s} W_iK_h(X_i - x_0) - \mcR_{p,s}'\mcU_{p,s}^{-1}\mcK_{p,s}^W$. Then, $\hat{\tau}_p^W = \,  \Gamma_{p,+}^{-1}\mcW_{p,+}^W - \Gamma_{p,-}^{-1}\mcW_{p,-}^W$. As $\mcG_p \equiv \Gamma_{p,+} + \Gamma_{p,-}$, \eqref{eq ZMvpD main} can be rewritten as
\begin{align}\label{eq tauivden gammatheta}
    \hat{\tau}_{IV,p}^W =& \displaystyle \sum_{i\in \mcN_+}W_iK_h(X_i - x_0) - \frac{\Gamma_{p,+}}{\Gamma_{p,+} + \Gamma_{p,-}}\displaystyle \sum_{i\in \mcN_h}W_iK_h(X_i - x_0) \notag \\
    &+ \frac{\Gamma_{p,+}}{\Gamma_{p,+} + \Gamma_{p,-}} \mcR_{p,+}'\mcU_{p,+}^{-1}\mcK_{p,+}^W + \frac{\Gamma_{p,+}}{\Gamma_{p,+} + \Gamma_{p,-}}\mcR_{p,-}'\mcU_{p,-}^{-1}\mcK_{p,-}^W - \mcR_{p,+}'\mcU_{p,+}^{-1}\mcK_{p,+}^W \notag \\
    =& \,\frac{\Gamma_{p,-}}{\Gamma_{p,+} + \Gamma_{p,-}} \mcW_{p,+}^W
    - \frac{\Gamma_{p,+}}{\Gamma_{p,+} + \Gamma_{p,-}} \mcW_{p,-}^W,
\end{align}
where the second equality uses $1 - a/(a+b) = b/(a+b)$ for scalar $a,b$. Then, using $\Gamma_p = \Gamma_{p,+}\Gamma_{p,-}/(\Gamma_{p,+} + \Gamma_{p,-})$, \eqref{eq tauivden gammatheta} can be expressed as
\begin{align}\label{eq equiv den}
    \hat{\tau}_{IV,p}^W = \frac{\Gamma_{p,+}\Gamma_{p,-}}{\Gamma_{p,+} + \Gamma_{p,-}} \left(\Gamma_{p,+}^{-1}\mcW_{p,+}^W - \Gamma_{p,-}^{-1}\mcW_{p,-}^W \right)= \Gamma_p\hat{\tau}_p^W.
\end{align} 
As $\Gamma_p$ is strictly positive and finite by Lemma \ref{lem Gamma finite},
\begin{align}
    \hat{\tau}_{IV,p} = \frac{\hat{\tau}_{IV,p}^Y}{\hat{\tau}_{IV,p}^D} = \frac{\Gamma_p \hat{\tau}_p^Y}{\Gamma_p \hat{\tau}_p^D} = \frac{ \hat{\tau}_p^Y}{ \hat{\tau}_p^D} = \hat{\tau}_{p}. \tag*{\QEDB}
\end{align}

\end{proof}

\begin{proof}[Proof of Proposition \ref{prop equivalence 2sls}]
As $Q = M_{\tilde{V}_p} - P_z$, for $W \in \{Y,D\}$, then
\begin{align*}
    \hat{\tau}_{\lambda,p}^W = \lambda\tilde{D}'P_z\tilde{W} + (1-\lambda)\tilde{D}'M_{\tilde{V}_p}\tilde{W}, \,\,\, \tilde{D}'P_z\tilde{W} = \frac{(z'\tilde{D})'(z'\tilde{W})}{z'z} = \frac{\hat{\tau}_{IV,p}^D\hat{\tau}_{IV,p}^W}{z'z}.
\end{align*}
As $\hat{\tau}_{IV,p}^W = \Gamma_p\hat{\tau}_p^W$ by \eqref{eq equiv den}, all that remains to show is $z'z = \Gamma_p$, where $z'z = Z'KZ - Z'KV_p(V_p'KV_p)^{-1}V_p'KZ$. Then,
\begin{align*}
    \tilde{Z}'P_{\tilde{V}_p}\tilde{Z} &= \, S_{0,+}^2 \mcG_p^{-1} - S_{0,+}\Upsilon_{p,+}\mcG_p^{-1} -S_{0,+} \mcG_p^{-1}\Upsilon_{p,+}
    + \Upsilon_{p,+} + \Upsilon_{p,+} \mcG_p^{-1} \Upsilon_{p,+} \\
    &= \, S_{0,+} \mcG_p^{-1}[S_{0,+} - \Upsilon_{p,+}]
    - [S_{0,+} - \Upsilon_{p,+}] \mcG_p^{-1}\Upsilon_{p,+}
    + \Upsilon_{p,+}.
\end{align*}
As $\Gamma_{p,+} = S_{0,+} - \Upsilon_{p,+}$ and $\mcG_p = \Gamma_{p,+} + \Gamma_{p,-}$, this becomes
\begin{align*}
    \tilde{Z}'P_{\tilde{V}_p}\tilde{Z} = S_{0,+} \Gamma_{p,+} \mcG_p^{-1}
    - \Gamma_{p,+} \mcG_p^{-1}\Upsilon_{p,+}
    + \Upsilon_{p,+} = \frac{\Gamma_{p,+}^2}{\Gamma_{p,+} + \Gamma_{p,-}}
    + \Upsilon_{p,+}.
\end{align*}
Given $\tilde{Z}'\tilde{Z} =  \sum_{i \in \mcN_h} Z_iK_h(X_i - x_0) = \sum_{i \in \mcN_+} K_h(X_i - x_0) = S_{0,+}$,
it follows that
\begin{align}\label{eq ZMZ final}
    z'z = \, S_{0,+} - \left(\frac{\Gamma_{p,+}^2}{\Gamma_{p,+} + \Gamma_{p,-}} + \Upsilon_{p,+}\right) = \Gamma_{p,+}  - \frac{\Gamma_{p,+}^2}{\Gamma_{p,+} + \Gamma_{p,-}} = \Gamma_p,
\end{align}
where $0 < \Gamma_p < \infty$ by Lemma \ref{lem Gamma finite}. Combining \eqref{eq ZMZ final} with $\hat{\tau}_{IV,p}^W = \Gamma_p\hat{\tau}_p^W$ from \eqref{eq equiv den}, then $\tilde{D}'P_z\tilde{W} = \Gamma_p^{-1}(\Gamma_p \hat{\tau}_p^D)(\Gamma_p \hat{\tau}_p^W) = \Gamma_p \hat{\tau}_p^D\hat{\tau}_p^W$, and so $\hat{\tau}_{\lambda,p}^W = \lambda\Gamma_p\hat{\tau}_p^D\hat{\tau}_p^W + (1-\lambda)\tilde{D}'M_{\tilde{V}_p}\tilde{W}$. Then, at $\lambda=1$,
\begin{align}\label{eq equiv num 2sls main}
    \hat{\tau}_{1,p} = \frac{\Gamma_p\hat{\tau}_p^D\hat{\tau}_p^Y}{\Gamma_p(\hat{\tau}_p^D)^2}= \frac{\hat{\tau}_p^Y}{\hat{\tau}_p^D} = \hat{\tau}_p. \tag*{\QEDB}
\end{align}
\end{proof}

\begin{proof}[Proof of Theorem \ref{thm lambda moments}]
Write $\epsilon_i = U_i + \xi_i$ for $\xi_i = (m(X_i) - V_{p,i}'\delta) + (\tau(X_i) - \tau)D_i$, with $\delta \in \Delta$ for some compact $\Delta \subset \mbR^{2p+1}$. Substitute $\tilde{Y} = \tilde{V}_p\delta + \tau\tilde{D} + \tilde{\epsilon}$ into \eqref{eq lambda class estimators}, note that $M_{\tilde{V}_p}\tilde{V}_p = 0$, and define $\Phi = (P_z + (1-\lambda)Q)\tilde{D}$. Then, $\hat{\tau}_{\lambda,p} - \tau = (\Phi'\tilde{D})^{-1}\Phi'\tilde{\epsilon}$. By the properties of projection matrices, $\|M_{\tilde{V}_p}\|_{op}\leq 1$ and $\|I_{n_h} - \lambda M_z\|_{op} \leq 1$ for $\| \cdot \|_{op}$ the operator norm, so it follows that $\|\Phi\|_2 \leq \| \tilde{D} \|_2 \leq \sqrt{nC_{K,h}}$, for $C_{K,h}=C_K/h$. Define $\vp_i = \Phi_i K^{1/2}_h(X_i - x_0)$, and so $|\vp_i| \leq \|\Phi\|_2 \cdot \sqrt{C_{K,h}} \leq \sqrt{n}C_{K,h} =: C_{\vp} < \infty$, and thus the coefficients are uniformly bounded. Also note that $\Phi'\tilde{\epsilon} = \sum_{i\in\mcN_h} \vp_i \epsilon_i = \sum_{i\in\mcN_h} \vp_i U_i + \sum_{i\in\mcN_h} \vp_i \xi_i$. Since $m(\cdot)$, $\tau(\cdot)$, and the local-polynomial matrix are bounded over $\mcX_h$, and the coefficient vector $\delta$ has bounded norm, there exists some $C_{\xi} < \infty$ such that $|\xi_i| \leq C_{\xi}$ for every $i \in \mcN_h$. Hence, $\left| \sum_{i \in \mcN_h} \vp_i \xi_i \right| \leq n C_{\vp} C_{\xi} < \infty$.

Fix $r < r^*$. By Lemma \ref{lem positivity}, there exists $C_\zeta^*>0$ such that $\hat{\tau}_{\lambda,p}^D \equiv \Phi'\tilde{D} \geq (1-\lambda)C_\zeta^* > 0$. From Lemma \ref{lem fu conditional}(ii), $\mbE[|U_i|^r|X, D] \leq C_B \mbE[|U_i|^r|X_i] < \infty$. Since $\sup_{x \in \mcX} \mbE[|U_i|^r|X_i = x] < \infty$ by Assumption \ref{as dgp}(iii)(b) for $r < r^*$, this bound holds uniformly over $(X,D)$. The finite-sum inequality gives
\begin{align}\label{eq phiu bound}
    \mbE\left[\left| \sum_{i \in \mcN_h} \vp_i U_i\right|^r \big| X,D \right] \leq n^{1+(r-1)_+}C_{\vp}^r C_B \sup_{x \in \mcX}\mbE[|U_i|^r|X_i=x] < \infty.
\end{align}
Combining \eqref{eq phiu bound} and the uniform bound on $\left| \sum_{i \in \mcN_h} \vp_i \xi_i \right|$ gives $\mbE[|\Phi'\tilde{\epsilon}|^r | X,D] < \infty$ uniformly in $(X,D)$. This and the denominator lower bound give
\begin{align}\label{eq taulambda minus lambda}
    \mbE[|\hat{\tau}_{\lambda,p} - \tau|^r] = \mbE[|\Phi'\tilde{D}|^{-r}|\Phi'\tilde{\epsilon}|^r]  &\leq [(1-\lambda)C_\zeta^*]^{-r} \mbE[|\Phi'\tilde{\epsilon}|^r] < \infty.
\end{align}
As $\tau$ is finite, \eqref{eq taulambda minus lambda} gives $\mbE[|\hat{\tau}_{\lambda,p}|^r] < \infty$ for all $r < r^*$.

Now fix $r \geq r^*$. Since $\Phi'\tilde{D} = \sum_{i \in \mcN_h} \vp_iD_i \geq (1-\lambda)C_\zeta^* > 0$, there exists $j \in \mcN_h$ such that $\vp_{j} \neq 0$, and let $i^*$ denote the smallest index satisfying this. Conditional on $(X,D)$, both $i^*$ and $\vp_{i^*}$ are fixed. Write $\Phi'\tilde{\epsilon} = \vp_{i^*}U_{i^*} + R_{i^*}$ for $R_{i^*} = \sum_{j \neq i^*} \vp_j U_j + \sum_j \vp_j \xi_j$. By Lemma \ref{lem fu conditional}(iii) $U_{i^*}|(X,D,U_{-i^*}) \sim U_{i^*}|(X_{i^*},D_{i^*})$, and so Assumption \ref{as tail persistence} gives $\mbE[|U_{i^*}|^r|X,D,U_{-i^*}] = \infty$. Given $\vp_{i^*} \neq 0$, on the event $|U_{i^*}| > 2|R_{i^*}|/|\vp_{i^*}|$, the reverse triangle inequality gives
\begin{gather*}
    |\Phi'\tilde{\epsilon}| = |\vp_{i^*}U_{i^*} + R_{i^*}| \geq |\vp_{i^*}||U_{i^*}| - |R_{i^*}| \geq \frac{|\vp_{i^*}|}{2}|U_{i^*}| \\
    \implies \mbE[|\Phi'\tilde{\epsilon}|^r|X,D,U_{-i^*}] \geq \left(\frac{|\vp_{i^*}|}{2}\right)^r    \mbE\left[ |U_{i^*}|^r 1\left\{|U_{i^*}| > \frac{2|R_{i^*}|}{|\vp_{i^*}|} \right\} \bigg| X, D, U_{-i^*}\right] = \infty.
\end{gather*}
The final equality holds as $2|R_{i^*}|/|\vp_{i^*}|$ is a finite constant conditional on $(X,D,U_{-i^*})$; the excluded region contributes at most $(2|R_{i^*}|/|\vp_{i^*}|)^r < \infty$ to $\mbE[|U_{i^*}|^r|X,D,U_{-i^*}] = \infty$. Therefore, taking expectations over $(X,D,U_{-i^*})$ gives $\mbE[|\Phi'\tilde{\epsilon}|^r] = \infty$. The bound $\Phi'\tilde{D} \leq nC_{K,h}$ holds by Cauchy--Schwarz, and so $\mbE[|\hat{\tau}_{\lambda,p} - \tau|^r] \geq (nC_{K,h})^{-r}\mbE[|\Phi'\tilde{\epsilon}|^r] = \infty$. Since $\tau$ is a finite constant, on the event $|\hat{\tau}_{\lambda,p} - \tau| > 2|\tau|$, the reverse triangle inequality gives $|\hat{\tau}_{\lambda,p}| \geq |\hat{\tau}_{\lambda,p} - \tau|/2$, while $\mbE\left[|\hat{\tau}_{\lambda,p} - \tau|^r 1\{|\hat{\tau}_{\lambda,p} - \tau| \leq 2|\tau|\}\right] \leq (2|\tau|)^r < \infty$. Therefore, for all $r \geq r^*$
\begin{align*}
    \mbE[|\hat{\tau}_{\lambda,p}|^r] \geq 2^{-r}\,\mbE\left[|\hat{\tau}_{\lambda,p} - \tau|^r 1\{|\hat{\tau}_{\lambda,p} - \tau| > 2|\tau|\}\right] = \infty. \tag*{\QEDB}
\end{align*}
\end{proof}

\begin{proof}[Proof of Theorem \ref{thm lambda consistent}]
The difference of the estimators is given by
\begin{align*}
    \hat{\tau}_{\lambda,p} - \hat{\tau}_p = \frac{(1-\lambda)(\tilde{D}'M_{\tilde{V}_p}\tilde{Y} - \hat{\tau}_p\tilde{D}'M_{\tilde{V}_p}\tilde{D})}{\lambda\Gamma_p(\hat{\tau}_p^D)^{2} + (1-\lambda)\tilde{D}'M_{\tilde{V}_p}\tilde{D}}.
\end{align*}
By Assumption \ref{as asymptotics}(i) and standard local-polynomial regression results \citep[e.g.,][]{fan1996local}, $\hat{\mu}_{p,s}(x_0) \overset{p}{\to} \mu_{s}(x_0)$ and $\hat{\pi}_{p,s}(x_0) \overset{p}{\to} \pi_{s}(x_0)$ for $s \in \{+,-\}$, so $\hat{\tau}_p^Y \overset{p}{\to} \mu_{+}(x_0) - \mu_{-}(x_0) = \tau^Y$ and $\hat{\tau}_p^D \overset{p}{\to} \pi_{+}(x_0) - \pi_{-}(x_0) = \tau^D$. As $\tau^D\neq 0$ by Assumption \ref{as limits}(ii), $\hat{\tau}_p \overset{p}{\to} \tau^Y/\tau^D = \tau$ by Slutsky's theorem. As $M_{\tilde{V}_p}$ is a projection matrix, $\tilde{D}'M_{\tilde{V}_p}\tilde{D} \leq \tilde{D}'\tilde{D} = O_p(n)$, and
\begin{align*}
    |\tilde{D}'M_{\tilde{V}_p}\tilde{Y}| \leq \left(\sumin D_iK_h(X_i-x_0) \right)^{1/2} \left(\sumin Y_i^2K_h(X_i-x_0) \right)^{1/2} = O_p(n),
\end{align*}
by Cauchy--Schwarz, since $\mbE[Y_i^2K_h(X_i - x_0)] = O(1)$ as $m(\cdot)$ and $\tau(\cdot)$ are bounded on $\mcX_h$, $\sup_{x \in \mcX} \mbE[U_i^2|X_i = x] <\infty$ by Assumption \ref{as asymptotics}(iii), and $f_X$ is bounded on $\mcX_h$. This implies that $\tilde{D}'M_{\tilde{V}_p}\tilde{Y} - \hat{\tau}_p\tilde{D}'M_{\tilde{V}_p}\tilde{D} = O_p(n)$, since $\hat{\tau}_p = O_p(1)$. Then, $(1-\lambda)(\tilde{D}'M_{\tilde{V}_p}\tilde{Y} - \hat{\tau}_p\tilde{D}'M_{\tilde{V}_p}\tilde{D})/n = (1-\lambda)O_p(1)$. By Lemma \ref{lem Gamma limit} and $\lambda \overset{p}{\to} 1$,
\begin{align*}
    \frac{\lambda\Gamma_p(\hat{\tau}_p^D)^{2}}{n} + \frac{(1-\lambda)\tilde{D}'M_{\tilde{V}_p}\tilde{D}}{n} \overset{p}{\to} C_{\Gamma,p}(\tau^D)^2 > 0,
\end{align*}
and so $\hat{\tau}_{\lambda,p} - \hat{\tau}_p = O_p(1-\lambda)$ follows. So, $\hat{\tau}_{\lambda,p} - \hat{\tau}_p = O_p(1-\lambda) = o_p(1)$ under Assumptions \ref{as asymptotics}(i)--(iii). If additionally Assumption \ref{as asymptotics}(iv) holds, then $1-\lambda = o_p((nh)^{-1/2})$, and so $\sqrt{nh}(\hat{\tau}_{\lambda,p} - \hat{\tau}_p) = o_p(1)$, such that $\sqrt{nh}(\hat{\tau}_{\lambda,p} - \tau) = \sqrt{nh}(\hat{\tau}_p - \tau) + o_p(1)$. \QEDB
\end{proof}

\begin{assum}\label{as bias correction moment}
Let $h_b > 0$ be fixed, $q \geq p+1$, and $\mcX_b = \mcX_h \cup \mcX_{h_b}$. (i) The functions $m(\cdot)$ and $\tau(\cdot)$ are bounded on $\mcX_b$, and Assumption \ref{as tail persistence} holds on $\mcX_b$; (ii) there exists $C_{\mcS,q}>0$ such that the minimum eigenvalue of the $q$-th order sample moment matrix $\mcS_{q,s}$ satisfies $\mineig(\mcS_{q,s}) \geq C_{\mcS,q}$ for $s \in \{+,-\}$.  
\end{assum}
Let $\mcE''$ denote the intersection of $\mcE'$ with the event that the $q$-th order sample moment matrices with bandwidth $h_b$ satisfy Assumption \ref{as bias correction moment}(ii) on both sides of the cutoff. For this theorem, assume $\mbP\{\mcE''\} > 0$ and take expectations conditional on $\mcE''$.

\begin{proof}[Proof of Theorem \ref{thm bias corrected moments}]
Conditional on $(X,D)$, write $\hat{\tau}_{\lambda,p,q}^{\BC} = \sum_{i \in \mcN_b} w_i(X, D)Y_i$. Lemma \ref{lem positivity} yields $\hat{\tau}_{\lambda,p}^D \geq (1 - \lambda)C_{\zeta}^* > 0$, so together with the kernel moment matrix eigenvalue bounds, then the coefficients $w_i(X,D)$ are uniformly bounded. For fixed $n$, $h$, $h_b$ and $\lambda < 1$, the uniform coefficient bounds and boundedness of $m$ and $\tau$ on $\mcX_b$ imply that, on $\mcE''$,
\begin{align*}
    |\hat{\tau}_{\lambda,p,q}^{\BC}|^r \leq C_r\left(1+\sum_{i=1}^n|U_i|^r\right).
\end{align*}
For $r < r^*$, the right-hand side has finite expectation under the original sampling law by Assumption \ref{as dgp}(iii)(b). Since $\mbP\{\mcE''\} > 0$, it follows that $\mbE[|\hat{\tau}_{\lambda,p,q}^{\BC}|^r] < \infty$. Shift equivariance implies that $\sum_{i \in \mcN_b} w_i(X, D)D_i = 1$, so at least one coefficient $w_i(X, D)$ is non-zero. Conditional independence and Assumption \ref{as tail persistence} then imply $\mbE[|\hat{\tau}_{\lambda,p,q}^{\BC}|^r] = \infty$ for $r \geq r^*$, again by the same argument as in the proof of Theorem \ref{thm lambda moments}. \QEDB
\end{proof}

\begin{proof}[Proof of Theorem \ref{thm bias corrected}]
Standard local-polynomial results give $\hat{\tau}_p^D \overset{p}{\to} \tau^D \neq 0$ under the stated assumptions \citep[e.g.,][]{hahn2001identification}. Therefore, a first-order delta method expansion of $\hat{\tau}_p$ around $(\tau^Y, \tau^D)$ gives $\hat{\tau}_p - \tau = [(\hat{\tau}_p^Y - \tau^Y) - \tau(\hat{\tau}_p^D - \tau^D)]/\tau^D + O_p((nh)^{-1})$.
Define $\theta = (\tau^Y, \tau^D, \Phi_{p+1,+}^Y, \Phi_{p+1,-}^Y, \Phi_{p+1,+}^D,  \Phi_{p+1,-}^D)'$, where $\Phi_{p+1,s}^W$ is the scaled population counterpart of $\hat{\Phi}_{p+1,q,s}^W$ for $W \in \{Y,D\}$ and $s \in \{+,-\}$, and denote by $\hat{\theta}$ the empirical counterpart of $\theta$. Further, define the population bias object $\mcB_p^W(h) = C_{p,+}\Phi_{p+1,+}^W - C_{p,-}\Phi_{p+1,-}^W$ for $W \in \{Y,D\}$, where $C_{p,s}$ is the probability limit of $\hat{C}_{p,s}$ for $s \in \{+,-\}$, and denote $\beta = (\mcB_p^Y(h), \mcB_p^D(h), 0, 0, 0, 0)'$. Then $\Delta = \hat{\theta} - \theta - \beta$. With this notation, express $\hat{\tau}_p^Y - \tau^Y = \mcB_p^Y(h) + \Delta_1$ and  $\hat{\tau}_p^D - \tau^D = \mcB_p^D(h) + \Delta_2$, and so $\hat{\tau}_p - \tau = \mcB_p(h) + \vp_\tau'\Delta + O_p((nh)^{-1})$ for $\vp_\tau = (1/\tau^D, -\tau/\tau^D, 0, 0, 0, 0)'$. By Theorem \ref{thm lambda consistent},
\begin{align}\label{eq taulambda bias expansion}
    \sqrt{nh}(\hat{\tau}_{\lambda,p} - \tau - \mcB_{p}(h)) = \vp_\tau'\sqrt{nh}\Delta + o_p(1).
\end{align}
To linearize, define $\gamma:\mbR^6\to \mbR$ (which is smooth on $\{\theta: \theta_2 \neq 0\}$) such that
\begin{align*}
    \gamma(\theta) = \frac{\theta_2(C_{p,+}\theta_3 - C_{p,-}\theta_4) - \theta_1(C_{p,+}\theta_5 - C_{p,-}\theta_6)}{\theta_2^2}.
\end{align*}
By construction, $\gamma(\theta) \equiv \mcB_p(h)$. At $\lambda = 1$, the conventional plug-in bias estimate is $\hat{\mcB}_{1,p,q}(h) = [\hat{\mcB}_{p,q}^Y(h) - \hat{\tau}_p \, \hat{\mcB}_{p,q}^D(h)]/\hat{\tau}_p^D$. Replacing $\hat{C}_{p,s}$ by $C_{p,s}$ in this expression gives $\gamma(\hat{\theta})$. By Lemma \ref{lem influence functions}, $\hat{\Phi}_{p+1,q,s}^W-\Phi_{p+1,s}^W = O_p((nh)^{-1/2})$. Since $\Phi_{p+1,s}^W = O(h^{p+1})$ and $nh^{2p+3} = O(1)$, it follows that $\hat{\Phi}_{p+1,q,s}^W = O_p((nh)^{-1/2})$, and so$(\hat{C}_{p,s}-C_{p,s}) \hat{\Phi}_{p+1,q,s}^W = o_p((nh)^{-1/2})$. As $\hat{\tau}_p^D \overset{p}{\to} \tau^D\neq 0$ and $\hat{\tau}_p = O_p(1)$, $\hat{\mcB}_{1,p,q}(h) = \gamma(\hat{\theta})+o_p((nh)^{-1/2})$. 

It remains to show that under Assumption \ref{as asymptotics}, using the $\lambda$-class BC estimator has no first-order effect. For $\lambda < 1$, write $\hat{\tau}_{\lambda,p}^D = \hat{\tau}_{1,p}^D(1 + r_n)$ for $r_n = (1 - \lambda)[\tilde{D}'M_{\tilde{V}_p} \tilde{D}/n - \hat{\tau}_{1,p}^D/n]/(\hat{\tau}_{1,p}^D/n)$. Since $\hat{\tau}_p^D \overset{p}{\to} \tau^D$ and $\Gamma_p/n \overset{p}{\to} C_{\Gamma,p} > 0$ by Lemma \ref{lem Gamma limit}, it follows that $\hat{\tau}_{1,p}^D/n \overset{p}{\to} C_{\Gamma,p}(\tau^D)^2 > 0$, and $\tilde{D}'M_{\tilde{V}_p} \tilde{D}/n = O_p(1)$ by the proof of Theorem \ref{thm lambda consistent}. As $1-\lambda = o_p((nh)^{-1/2}) = o_p(1)$ by Assumption \ref{as asymptotics}(iv), $r_n  = o_p(1)$. Therefore,
\begin{align*}
    \hat{\mcB}_{\lambda,p,q}(h) = \frac{\lambda}{1+r_n} \left(\gamma(\hat{\theta}) + o_p((nh)^{-1/2})\right) = \frac{\lambda}{1+r_n}\gamma(\hat{\theta}) + o_p((nh)^{-1/2}),
\end{align*}
where the second equality follows since
$\lambda/(1+r_n)=O_p(1)$. Lemma \ref{lem influence functions} gives $\Delta = O_p((nh)^{-1/2})$, and $\Phi_{p+1,s}^W = O(h^{p+1})$ by definition. Since Assumption \ref{as asymptotics}(iv) implies $h^{p+1} = O((nh)^{-1/2})$, $\gamma(\hat{\theta}) = O_p((nh)^{-1/2})$. So, it follows that $\sqrt{nh}(\lambda/(1+r_n) - 1)\gamma(\hat{\theta}) = o_p(1)$, as $r_n = o_p(1)$ implies $\lambda/(1+r_n) - 1 = o_p(1)$ by Assumption \ref{as asymptotics}(ii). Expanding $\gamma(\hat{\theta})$ around $\theta$ gives $\gamma(\hat{\theta}) - \gamma(\theta) = \nabla \gamma(\theta)'\Delta + \nabla \gamma(\theta)'\beta + O_p((nh)^{-1})$,
since $\hat{\theta} - \theta = \Delta + \beta$. Define $\Xi^Y = C_{p,+}\theta_3 - C_{p,-}\theta_4$ and $\Xi^D = C_{p,+}\theta_5 - C_{p,-}\theta_6$. Then the gradient is
\begin{align*}
    \nabla \gamma(\theta)' =
    \begin{pmatrix}
        -\displaystyle\frac{\Xi^D}{\theta_2^2}, &
        \displaystyle\frac{2\theta_1\Xi^D - \theta_2\Xi^Y}{\theta_2^3}, &
        \displaystyle\frac{C_{p,+}}{\theta_2}, &
        -\displaystyle\frac{C_{p,-}}{\theta_2}, &
        -\displaystyle\frac{C_{p,+}\theta_1}{\theta_2^2}, &
        \displaystyle\frac{C_{p,-}\theta_1}{\theta_2^2}
    \end{pmatrix}.
\end{align*}
Since the first two components of $\nabla \gamma(\theta)$ and $\beta$ are both $O(h^{p+1})$, it follows that $\nabla \gamma(\theta)'\beta = O(h^{2(p+1)})$. As $\sqrt{nh}h^{2(p+1)} = (\sqrt{nh}h^{p+1})h^{p+1} \to 0$, it follows that
\begin{align}\label{eq lambdabias diff}
    \sqrt{nh}(\hat{\mcB}_{\lambda,p,q}(h) - \mcB_p(h)) = \nabla \gamma(\theta)'\sqrt{nh}\Delta + o_p(1).
\end{align}
As $h \to 0$, $\Phi_{p+1,s}^Y$ and $\Phi_{p+1,s}^D$ are $o(1)$, and therefore 
\begin{align*}
    \nabla \gamma(\theta)' \to
    \begin{pmatrix}
        0, &
        0, &
        \displaystyle\frac{C_{p,+}}{\tau^D}, &
        -\displaystyle\frac{C_{p,-}}{\tau^D}, &
        -\displaystyle\frac{\tau C_{p,+}}{\tau^D}, &
        \displaystyle\frac{\tau C_{p,-}}{\tau^D}
    \end{pmatrix}.
\end{align*}
\eqref{eq taulambda bias expansion} and \eqref{eq lambdabias diff} combine to give $\sqrt{nh}(\hat{\tau}_{\lambda,p,q}^{\BC} - \tau) = \phi_n' \sqrt{nh} \Delta + o_p(1)$, for $\phi_n = \vp_\tau - \nabla \gamma(\theta)$, and so $\phi_n \to \phi = (1/\tau^D)(1,-\tau,-C_{p,+},C_{p,-},\tau C_{p,+}, -\tau C_{p,-})'$. Applying Lemma \ref{lem influence functions} therefore gives the asymptotic linear representation $\sqrt{nh}(\hat{\tau}_{\lambda,p,q}^{\BC} - \tau) = (nh)^{-1/2}\sum_{i=1}^n \phi_n's_{n,i} + o_p(1)$. The asymptotic variance satisfies $\phi_n'\Sigma_n \phi_n \to \phi'\Sigma \phi = V_{\BC}$. Assumption \ref{as kernel}(i) and the uniform conditional $2+\delta$ moments imply $(nh)^{-1}\sumin \mbE[|\phi_n's_{n,i}|^{2+\delta}] = O(1)$, so
\begin{align*}
    \frac{1}{(nh)^{1+\delta/2}}\sumin \mbE[|\phi_n's_{n,i}|^{2+\delta}] = O((nh)^{-\delta/2}) = o(1).
\end{align*}
Therefore, the Lyapunov condition holds and the central limit theorem for triangular arrays yields $\sqrt{nh}(\hat{\tau}_{\lambda,p,q}^{\BC} - \tau) \overset{d}{\to} N(0, V_{\BC})$. Let $\hat{V}_1$ denote the conventional BC variance estimator, using the same bandwidths and nearest-neighbor residuals, and define $a_{\lambda,i} = K_{ii}^{1/2}[(P_z + (1-\lambda)Q)\tilde{D}]_i$ for $i \in \mcN_h$ (and zero for $i \notin \mcN_h$). With $\tilde{k}_{p,q,i} = k_{p,i} - k_{p,q,i}^{\BC}$, the implemented first-order score coefficient vector is
\begin{align*}
    c_{\lambda,i} = \frac{1}{\hat{\tau}_{\lambda,p}^D}
    \begin{pmatrix}
        a_{\lambda,i} - \lambda \hat{\tau}_{IV,p}^D \tilde{k}_{p,q,i} \\
        -\hat{\tau}_{\lambda,p}a_{\lambda,i} + \lambda \hat{\tau}_{IV,p}^Y \tilde{k}_{p,q,i}
    \end{pmatrix}.
\end{align*}
At $\lambda = 1$, this reduces to the conventional BC score coefficient $c_{1,i} = k_{p,q,i}^{\BC}(1,-\hat{\tau}_p)'/\hat{\tau}_p^D$. Convergence of the normalized sample local-polynomial moment matrices and $\hat{\tau}_p^D \overset{p}{\to} \tau^D$ then imply that
\begin{align*}
    \max_{i \in \mcN_b} \| c_{\lambda,i} - c_{1,i}\| = O_p\left(\frac{1-\lambda}{nh}\right),\,\,\, \max_{i \in \mcN_b} \| c_{1,i} \| = O_p((nh)^{-1}).
\end{align*}
Let $\nu_i$ denote the nearest-neighbor residual vector such that $\hat{\mbV}(\hat{\tau}_{\lambda,p,q}^{\BC}) = \sum_{i \in\mcN_b}(c_{\lambda,i}'\nu_i)^2$. The conditional moment bounds and a fixed number of neighbors imply $\sum_{i \in \mcN_b}\|\nu_i\|^2 = O_p(nh)$, so
\begin{align*}
    nh|\hat{\mbV}(\hat{\tau}_{\lambda,p,q}^{\BC}) - \hat{V}_1| \leq nh \max_{i \in \mcN_b}\| c_{\lambda,i} - c_{1,i}\| \max_{i \in \mcN_b}(\| c_{\lambda,i}\| + \|c_{1,i}\|)\sum_{i \in \mcN_b}\|\nu_i\|^2 = o_p(1).
\end{align*}
Theorem A.3 of \citet[Supplementary Material]{calonico2014robust} together with consistency of the plug-in coefficients gives $nh\hat{V}_1 \overset{p}{\to} V_{\BC}$, and therefore $nh\hat{\mbV}(\hat{\tau}_{\lambda,p,q}^{\BC}) \overset{p}{\to} V_{\BC}$, so studentized asymptotic normality follows from Slutsky's theorem. \QEDB
\end{proof}

\begin{assum}\label{as nagar design regularity}
(i) There exist $X$-measurable design events $\mcF_n$, $\mbP\{\mcF_n\} \to 1$, such that on $\mcF_n$, the side-specific local polynomial moment matrices of orders $p$ and $q$ are nonsingular at bandwidths $h$ and $h_b$ respectively. Further, the following bounds hold uniformly over $X \in \mcF_n$: $|\mcN_b| = O(nh)$, $\Gamma_p \asymp n$, $\traceoperator(P_zK) \asymp h^{-1}$, $\traceoperator(QK) \asymp n$, $\max_{i \in\mcN_h}|z_i| K_{ii}^{1/2} = O(h^{-1})$, and $\max_{i \in \mcN_b}(|k_{p,i}| + |k_{p,q,i}^{\BC}|) = O((nh)^{-1})$. (ii) There exist constants $C_\lambda$ and $c_{\lambda}$ such that uniformly over $\mcF_n$ and treatment realizations with $\hat{\tau}_{\lambda,p}^D > 0$, then $(\hat{\tau}_{\lambda,p}^D)^{-1} \leq C_{\lambda}h(nh)^{c_{\lambda}}$.
\end{assum}

\begin{proof}[Proof of Proposition \ref{prop lambda selection}]
Throughout the proof, replace $K$ by $hK$ and retain the same notation for the rescaled quantities. This multiplies $\Gamma_p$, $\hat{\tau}_{\lambda,p}^D$, and both traces by $h$, while leaving the estimators and the displayed leading bias terms unchanged. Thus $\Gamma_p\asymp nh$, and the following bounds hold uniformly on $\mcF_n$. Denote $U = \bar{Y} - \tau\bar{D}$, and let $v = K^{1/2}\bar{D}$ and $u = K^{1/2}U$, and define $S_U = z'u$ and $S_V = z'v$. Conditional independence, boundedness of $\bar{D_i}$, uniform boundedness of the fourth moments of $U_i$, and the design assumptions then give
\begin{gather}\label{eq Su Sv bounds}
    \mbE[S_U|X] = \mbE[S_V|X] = 0, \,\,\, \mbE[S_U^2+S_V^2|X] = O(nh), \\ \mbE[S_U^4+S_V^4|X] = O((nh)^2), \,\,\, \mbE[S_V^8|X] = O((nh)^4).
\end{gather}
Further, as $Q$ is an orthogonal projection, $\|v\|^2 = O(nh)$, $\mbE[\|u\|^2|X] = O(nh)$ and $1-\lambda = O((nh)^{-1})$, then $(1-\lambda)v'Qv = O(1)$ and $\mbE[((1-\lambda)v'Qu)^2|X] = O(1)$ uniformly.
Note that $\hat{\tau}_{IV,p}^D = \tau^D \Gamma_p + S_V$, $\hat{\tau}_{IV,p}^Y - \tau \hat{\tau}_{IV,p}^D = S_U$, and  $QK^{1/2}\mu(X) = QK^{1/2}\pi(X) = 0$. Therefore,
\begin{align*}
    \hat{\tau}_{\lambda,p} - \tau = \frac{\tilde{D}'(P_z + (1-\lambda)Q)(\tilde{Y} - \tau\tilde{D})}{\tilde{D}'(P_z + (1-\lambda)Q)\tilde{D}} = \frac{\tau^D S_U + S_US_V/\Gamma_p + (1-\lambda)v'Qu}{(\tau^D)^2\Gamma_p + 2\tau^D S_V + S_V^2/\Gamma_p + (1-\lambda)v'Qv}.
\end{align*}
Denote $\mcJ_1 = \tau^DS_U$, $\mcJ_2 = S_US_V/\Gamma_p + (1-\lambda)v'Qu$, $\mcJ_3 = (\tau^D)^2\Gamma_p$, and $\mcJ_4 = 2\tau^D S_V + S_V^2/\Gamma_p + (1-\lambda)v'Qv$, such that $\hat{\tau}_{\lambda,p} - \tau = (\mcJ_1 + \mcJ_2)/(\mcJ_3 + \mcJ_4)$. Then,
\begin{align}\label{eq J bounds}
   \mbE[\mcJ_1^2|X] = O(nh), \,\,\, \mbE[\mcJ_2^2|X] = O(1), \,\,\,  \mbE[\mcJ_4^2|X] = O(nh), \,\,\,  \mbE[\mcJ_4^4|X] = O((nh)^2).
\end{align}
In the following generic expansion, define $M/(\mcJ_3+\mcJ_4)$ to be zero when $\mcJ_3+\mcJ_4=0$, and denote $M_D = \mcJ_3+\mcJ_4$. I now establish that for any numerator $M = \tilde{\mcJ}_1 + \tilde{\mcJ}_2$ such that $\mbE[\tilde{\mcJ}_1|X] = 0$, $\mbE[\tilde{\mcJ}_1^2|X] = O(nh)$, and $\mbE[\tilde{\mcJ}_2^2|X] = O(1)$, then 
\begin{align}\label{eq arbitrary numerator bound}
    \mbE\left[\frac{M}{M_D}\bigg|X\right] = \frac{\mbE[\tilde{\mcJ}_2|X]}{\mcJ_3} - \frac{\mbE[\tilde{\mcJ}_1\mcJ_4|X]}{\mcJ_3^2} + o((nh)^{-1}). 
\end{align}
Let $\mcA = \{|\mcJ_4| \leq \mcJ_3/2\}$. As $\mcJ_3 \asymp nh$, $(1-\lambda)v'Qv = O(1)$ uniformly, and $\Gamma_p \asymp nh$, then there exists some $c_\mcA > 0$ such that $|S_V| \leq c_\mcA nh$ implies the event $\mcA$ for sufficiently large $nh$. Then $\mbP\{\mcA^c|X\} \leq \mbP\{|S_V| > c_\mcA nh|X\} \leq 2\exp(-c_\mcA' nh)$ for some $c_\mcA'>0$, by Hoeffding's inequality. On event $\mcA$, the exact identity
\begin{align*}
    \frac{M}{M_D} = M\left[\frac{1}{\mcJ_3} - \frac{\mcJ_4}{\mcJ_3^2} + \frac{\mcJ_4^2}{\mcJ_3^2M_D}\right]
\end{align*}
holds. As $M_D \geq \mcJ_3 /2$, \eqref{eq J bounds} and Cauchy--Schwarz give
\begin{align*}
    \mbE\left[\frac{|M|\mcJ_4^2}{\mcJ_3^2M_D}1\{\mcA\}\bigg|X\right] \leq \frac{2\mbE[M^2|X]^{1/2}\mbE[\mcJ_4^4|X]^{1/2}}{\mcJ_3^3} = O((nh)^{-3/2}) = o((nh)^{-1}).
\end{align*}
Since $\mcJ_3>0$ on $\mcF_n$, a zero denominator implies $\mcA^c$. On these draws either centered estimator equals $-\tau$ so their contribution to either conditional bias is at most $|\tau|\mbP\{\mcA^c \mid X\}=o((nh)^{-1})$ in absolute value. On the event $\mcA^c \cap \{M_D > 0\}$, Assumption \ref{as nagar design regularity}(ii) implies that
\begin{align*}
    \mbE\left[\frac{|M|}{M_D}1\{\mcA^c, M_D > 0\}\bigg|X\right] \leq C_\lambda (nh)^{c_{\lambda}} \mbE[M^2|X]^{1/2}\mbP\{\mcA^c, M_D > 0|X\}^{1/2} = o((nh)^{-1}).
\end{align*}
The same argument shows that $M/\mcJ_3$ and $M\mcJ_4/\mcJ_3^2$ on $\mcA^c$ are also $o((nh)^{-1})$. Finally, 
\begin{align*}
    \frac{|\mbE[\tilde{\mcJ}_2\mcJ_4|X]|}{\mcJ_3^2} \leq \frac{(\mbE[\tilde{\mcJ}_2^2|X]\mbE[\mcJ_4^2|X])^{1/2}}{\mcJ_3^2} = O((nh)^{-3/2}),
\end{align*}
so together with $\mbE[\tilde{\mcJ}_1|X] = 0$, this implies \eqref{eq arbitrary numerator bound}.

For the conventional $\lambda$-class estimators, conditional independence and constant covariance imply $\mbE[S_US_V|X] = \sigma_{U\bar{D}}\Gamma_p\traceoperator(P_zK)$ and $\mbE[v'Qu|X] = \sigma_{U\bar{D}}\traceoperator(QK)$. Therefore, $\mbE[\mcJ_2|X] = \sigma_{U\bar{D}}[\traceoperator(P_zK) + (1-\lambda)\traceoperator(QK)]$. Also,
\begin{align*}
    \mbE[\mcJ_1\mcJ_4|X] = 2(\tau^D)^2\mbE[S_US_V|X] + \frac{\tau^D}{\Gamma_p}\mbE[S_US_V^2|X] + (1-\lambda)\tau^D \mbE[S_Uv'Qv|X].
\end{align*}
By Cauchy--Schwarz and \eqref{eq J bounds}, $|\mbE[S_US_V^2|X]|/\Gamma_p = O((nh)^{1/2})$ and $(1-\lambda)|\mbE[S_Uv'Qv|X]| = O((nh)^{1/2})$, so $\mbE[\mcJ_1\mcJ_4|X] = 2(\tau^D)^2\sigma_{U\bar{D}}\Gamma_p\traceoperator(P_zK) + o(nh)$. Therefore, applying \eqref{eq arbitrary numerator bound} yields
\begin{align*}
    \mbE[\hat{\tau}_{\lambda,p} - \tau|X] &= \frac{\sigma_{U\bar{D}}}{(\tau^D)^2\Gamma_p}\left[(1-\lambda)\traceoperator(QK) - \traceoperator(P_zK)\right] + o((nh)^{-1}).
\end{align*}
Define $\tilde{k}_{p,q,i} = k_{p,i} - k_{p,q,i}^{\BC}$ and $\mcK_W = \sum_{i \in \mcN_b}\tilde{k}_{p,q,i} W_i$ for $W \in \{U,\bar{D}\}$. Since $\pi(\cdot)$ and $m(\cdot)$ are side-specific polynomials of degree at most $p$ by Assumption \ref{as lambda selection}(i), local-polynomial regressions of order $q \geq p+1$ fit these exactly, and so $\hat{\mcB}_{p,q}^D(h) = \mcK_{\bar{D}}$ and $\hat{\mcB}_{p,q}^Y(h)-\tau\hat{\mcB}_{p,q}^D(h) = \mcK_U$. Then $\hat{\tau}_{\lambda,p,q}^{\BC} - \tau = (\mcJ_1^{\BC} + \mcJ_2^{\BC})/M_D$ for $\mcJ_1^{\BC} = \mcJ_1 -\lambda\tau^D\Gamma_p \mcK_U$ and $\mcJ_2^{\BC} = \mcJ_2 - \lambda(S_V\mcK_U - S_U\mcK_{\bar{D}})$. The equivalent-kernel weight bounds, conditional independence, and the fourth-moment assumptions yield
\begin{align*}
    \mbE[\mcK_U^2|X] + \mbE[\mcK_{\bar{D}}^2|X] = O((nh)^{-1}), \,\,\, \mbE[\mcK_U^4|X] + \mbE[\mcK_{\bar{D}}^4|X] = O((nh)^{-2}),
\end{align*}
so $\mbE[\mcJ_1^{\BC}|X] = 0$ and $\mbE[(\mcJ_1^{\BC})^2|X] = O(nh)$. Cauchy--Schwarz also gives
$\mbE[S_V^2\mcK_U^2|X] \leq (\mbE[S_V^4|X]\mbE[\mcK_U^4|X])^{1/2} = O(1)$ and $\mbE[S_U^2\mcK_{\bar{D}}^2|X] \leq (\mbE[S_U^4|X]\mbE[\mcK_{\bar{D}}^4|X])^{1/2} = O(1)$. These together with \eqref{eq J bounds} imply $\mbE[(\mcJ_2^{BC})^2|X] = O(1)$. The expansion \eqref{eq arbitrary numerator bound} applies, with
\begin{align*}
    \mbE[\hat{\tau}_{\lambda,p,q}^{\BC} - \tau|X] = \frac{\mbE[\mcJ_2^{\BC}|X]}{\mcJ_3} - \frac{\mbE[\mcJ_1^{\BC}\mcJ_4|X]}{\mcJ_3^2} + o((nh)^{-1}).
\end{align*}
To bound these terms, note that $k_{p,i} = z_iK_{ii}^{1/2}/\Gamma_p$ and $\Gamma_p \sum_i k_{p,i}^2 = \traceoperator(P_zK)$, and so
\begin{align*}
    \mbE[S_V\mcK_U|X] = \mbE[S_U\mcK_{\bar{D}}|X] = \sigma_{U\bar{D}}\Gamma_p \sum_{i \in \mcN_b} k_{p,i}\tilde{k}_{p,q,i} = \sigma_{U\bar{D}}(1 - \vr_{p,q})\traceoperator(P_zK).
\end{align*}
From this, it follows that $\mbE[\mcJ_2^{\BC}|X] = \sigma_{U\bar{D}}[\traceoperator(P_zK) + (1-\lambda)\traceoperator(QK)]$. Similarly,
\begin{align*}
    \mbE[\mcJ_1^{\BC}\mcJ_4|X] =& \, 2(\tau^D)^2\mbE[(S_U - \lambda \Gamma_p \mcK_U)S_V|X] + \frac{\tau^D}{\Gamma_p}\mbE[(S_U - \lambda \Gamma_p \mcK_U)S_V^2|X] \\ 
    &+ (1-\lambda)\tau^D \mbE[(S_U -\lambda \Gamma_p \mcK_U)v'Qv|X] \\
    =& \, 2(\tau^D)^2\sigma_{U\bar{D}}\Gamma_p\traceoperator(P_zK)[1-\lambda(1-\vr_{p,q})] + o(nh) \\
    =& \, 2(\tau^D)^2\sigma_{U\bar{D}}\Gamma_p\vr_{p,q}\traceoperator(P_zK) + o(nh),
\end{align*}
where the final equality follows as $1-\lambda = O((nh)^{-1})$, $\vr_{p,q} = O(1)$, and $\traceoperator(P_zK) = O(1)$. It therefore follows that
\begin{align*}
    \mbE[\hat{\tau}_{\lambda,p,q}^{\BC} - \tau|X] &= \frac{\sigma_{U\bar{D}}}{(\tau^D)^2\Gamma_p}\left[(1-\lambda)\traceoperator(QK) - (2\vr_{p,q}-1)\traceoperator(P_zK)\right] + o((nh)^{-1}). \tag*{\QEDB}
\end{align*}
\end{proof}

\section{Supporting lemmas}\label{sec lemmas}

\begin{lem}\label{lem fu conditional}
Let Assumptions \ref{as dgp}--\ref{as finite sample} hold. Then (i) for every finite $B >0$ and $X_i \in \mcX_h$, $f_{U_i|X,D,U_{-i}}(u) \geq \pi_B f_{U|X}(u|X_i)$ for all $|u| \leq B$; (ii) there exists $C_B < \infty$ such that $f_{U_i|X,D,U_{-i}}(u) \leq C_B f_{U|X}(u|X_i)$ for all $u \in \mbR$; (iii) $\{U_i\}_{i=1}^n$ are mutually independent conditional on $(X,D)$, and $f_{U_i|X,D}(u) = f_{U_i|X,D,U_{-i}}(u)$ for all $i$.
\end{lem}

\begin{proof}[Proof of Lemma \ref{lem fu conditional}]
Part (i): As $U_i = Y_i - m(X_i) - \tau(X_i)D_i$, Assumption \ref{as dgp}(i) implies that $\{(U_i,X_i,D_i)\}_{i=1}^n$ is also i.i.d., so $f_{U_i|X,D,U_{-i}}(u|X,D,U_{-i}) = f_{U_i|X_i,D_i}(u|X_i,D_i)$. Bayes' rule then gives
\begin{align*}
    f_{U_i|X_i,D_i,U_{-i}}(u) = \frac{\pi_U(X_i,u)^{d_i}(1-\pi_U(X_i,u))^{1-d_i}}{\mbP\{D_i = d_i|X_i\}} f_{U|X}(u|X_i).
\end{align*}
As $\pi_U(X_i,u)^{d_i}(1-\pi_U(X_i,u))^{1-d_i} \geq \pi_B$ for $|u| \leq B$ and $\mbP\{D_i = d_i|X_i\} \leq 1$ by definition, result (i) follows. For part (ii), fix any $B_0 > 0$. By Assumption \ref{as dgp}(iii)(c), there exists $c_{B_0} = \inf_{|u| \leq B_0, x \in \mcX} f_{U|X}(u|x) > 0$, so for every $x \in \mcX_h$ and $d \in \{0,1\}$,
\begin{align*}
    \mbP\{D_i = d|X_i = x\} &= \int \pi_U(x,u)^d(1-\pi_U(x,u))^{1-d} f_{U|X}(u|x) du \\
    & \geq \int_{-B_0}^{B_0} \pi_{B_0} f_{U|X}(u|x) du \geq 2B_0 \pi_{B_0} c_{B_0}.
\end{align*}
Since $\pi_U(X_i,u)^{d_i}(1-\pi_U(X_i,u))^{1-d_i} \leq 1$, the result follows with $C_B = (2B_0 \pi_{B_0} c_{B_0})^{-1}$. For part (iii), as $\{(U_i,X_i,D_i)\}_{i=1}^n$ are i.i.d., $f_{U|X,D}(u|x,d) = \prod_i f_{U_i|X_i,D_i}(u_i|x_i,d_i)$,
so $U_i$ are mutually independent conditional on $(X,D)$, and $f_{U_i|X,D,U_{-i}} = f_{U_i|X,D}$. \QEDB
\end{proof}

\begin{lem}\label{lem gradient upper}
Let Assumptions \ref{as dgp}--\ref{as finite sample} hold, and consider a first-stage effective sample $(x,d)$ with $x \in (\mcX_h^\circ \setminus \{x_0\})^{n_h}$. Then, there exists some $\overline{C}_{\nabla} < \infty$ such that $\norm{\nabla_x g_d(x)}_2 \leq \overline{C}_{\nabla}$, with $g_d(x) = \hat{\tau}_p^D(x,d)$. Further, $|g_d(x)| \leq nC_k < \infty$ uniformly over admissible $(x,d)$. 
\end{lem}
\begin{proof}[Proof of Lemma \ref{lem gradient upper}]
Consider first $|\partial\omega_{p,s,i}/\partial x_j|$ for $j \in \mcN_s$, $s \in \{+,-\}$. Then, 
\begin{align}\label{eq S deriv bound}
    \left|\left[\frac{\partial \mcS_{p,s}}{\partial x_j}\right]_{k,\ell}\right|
    = \left|\frac{\partial}{\partial x_j} K_h(x_j - x_0)(x_j - x_0)^{k+\ell-2}\right|
    \leq [L_K + C_K(k+\ell-2)]h^{k+\ell-4} < \infty,
\end{align}
where $[\cdot]_{k,\ell}$ denotes the $(k,\ell)$ element of a matrix, and $L_K$ is the Lipschitz constant for $K(\cdot)$ of Assumption \ref{as kernel}. As $\left[\mcH_{p,s,i}\right]_{k} = (x_i - x_0)^{k-1}$,
\begin{align}\label{eq deriv H bound}
    \left|\left[\mcH_{p,s,i}\right]_{k}\right| \leq h^{k-1}, \,\,\,
    \left|\left[\frac{\partial \mcH_{p,s,i}}{\partial x_j}\right]_{k}\right| \leq (k-1)h^{k-2} < \infty.
\end{align}
The bounds in \eqref{eq S deriv bound} and \eqref{eq deriv H bound} then give $\norm{\mcH_{p,s,i}}_{\infty} \leq \max\{1,h^{p}\}=: C_H < \infty$ for $\| \cdot \|_{\infty}$ the entrywise maximum, and
\begin{gather}\label{eq H norm bound}
    \norm{\frac{\partial \mcS_{p,s}}{\partial x_j}}_{\infty}\leq \max_{1 \leq k,\ell \leq p+1} [L_K + C_K(k+\ell-2)]h^{k+\ell-4} =: C_{\partial S} < \infty, \,\,\, \\  \norm{\frac{\partial \mcH_{p,s,i}}{\partial x_j}}_{\infty} \leq \max_{1 \leq k \leq p+1}(k-1) h^{k-2} =: C_{\partial H} < \infty.
\end{gather}
By Assumption \ref{as finite sample}(iv), there exists fixed $C_{\mcS} > 0$ such that $\mineig(\mcS_{p,s}) \geq C_{\mcS} > 0$. Then,
\begin{align}\label{eq omega bound}
    |\omega_{p,s,i}| \leq \norm{e_{p,1}}_2 \norm{\mcS_{p,s}^{-1}}_2 \norm{\mcH_{p,s,i}}_2 \leq C_{\mcS}^{-1}\sqrt{p+1} C_H=:C_{\omega} < \infty,
\end{align}
and also
\begin{align}\label{eq omega deriv bound}
    \left|\frac{\partial \omega_{p,s,i}}{\partial x_j}\right| &\leq \norm{e_{p,1}}_2\norm{\mcS_{p,s}^{-1}}_2^2\norm{\frac{\partial \mcS_{p,s}}{\partial x_j}}_2 \norm{\mcH_{p,s,i}}_2 + \norm{e_{p,1}}_2\norm{\mcS_{p,s}^{-1}}_2 \norm{\frac{\partial \mcH_{p,s,i}}{\partial x_j}}_2 \notag \\
    &\leq C_{\mcS}^{-2}(p+1)^{3/2}C_{\partial S} C_H+C_{\mcS}^{-1}\sqrt{p+1}C_{\partial H} =: C_{\partial\omega} < \infty,
\end{align}
where the second line follows as for an $\ell\times \ell$ matrix $A$, $\norm{A}_2 \leq \ell \norm{A}_{\infty}$, and for vector $b\in\mbR^{\ell}$, $\norm{b}_2 \leq \sqrt{\ell}\norm{b}_\infty$. For $\delta_{ij}$ the Kronecker delta,
\begin{align*}
    \left|\frac{\partial}{\partial x_j}d_i K_h(x_i - x_0)\omega_{p,s,i}\right|
    &\leq \left|K_h(x_i - x_0)\right| \left|\frac{\partial \omega_{p,s,i}}{\partial x_j}\right| + \delta_{ij}\left|K'_h(x_j-x_0)\right|\left|\omega_{p,s,j}\right|,
\end{align*}
which given \eqref{eq omega deriv bound}, implies that
\begin{align*}
    \left|\frac{\partial}{\partial x_j}\sum_{i\in \mcN_s}d_i K_h(x_i - x_0)\omega_{p,s,i}\right| \leq nC_{K,h}C_{\partial\omega} + L_KC_{\omega}/h^2 < \infty.
\end{align*}
As $\sum_{i\in \mcN_s}K_h(x_i - x_0)\omega_{p,s,i} = 1$ given Assumptions \ref{as kernel} and \ref{as finite sample}, $\partial\sum_{i\in \mcN_s}K_h(x_i - x_0)\omega_{p,s,i}/\partial x_j = 0$, so
$|\partial g_d/\partial x_j| \leq nC_{K,h}C_{\partial\omega} + L_KC_{\omega}/h^2 < \infty$. Therefore,
\begin{align}\label{eq gradient upper bound}
    \norm{\nabla_x g_d(x)}_2 \leq \sqrt{n}\left(nC_{K,h}C_{\partial\omega} + L_KC_{\omega}/h^2\right) =: \overline{C}_{\nabla} < \infty.
\end{align}
Further, \eqref{eq omega bound} and Assumption \ref{as kernel} give $|g_d(x)| \leq \sum_{i\in \mcN_h}|d_i| |k_{p,i}| \leq nC_k < \infty$. \QEDB
\end{proof}

\begin{lem} \label{lem differentiability}
Let Assumptions \ref{as dgp}--\ref{as finite sample} hold. Consider a first-stage effective sample $(x,d)$ such that $x \in (\mcX_h^\circ \setminus \{x_0\})^{n_h}$, where $A^\circ$ denotes the interior of set $A$. Treating $d$ as fixed, there exists an open neighborhood $\mcB_\ve(x)$ of $x$ with radius $\ve > 0$, such that $g_d \in \mcC_b^\vs (\mcB_\ve(x))$, for $g_d(x) = \hat{\tau}_p^D(x,d)$.
\end{lem}
\begin{proof}[Proof of Lemma \ref{lem differentiability}]
Write $g_d(x) = e_{p,1}'(\mcS_{p,+}^{-1}\mcH_{p,+}'K_+D_+ - \mcS_{p,-}^{-1}\mcH_{p,-}'K_-D_-)$, as in \eqref{eq tau p matrix}. As $x \in (\mcX_h^\circ \setminus \{x_0\})^{n_h}$, choose some $\ve > 0$ sufficiently small such that $\cl(\mcB_\ve(x)) \subset (\mcX_h^\circ \setminus \{x_0\})^{n_h}$, with $\cl(A)$ denoting the closure of set $A$. Term by term, for $s \in \{+,-\}$, $\mcH_{p,s}$ is a matrix with entries polynomial in $x - x_0$, and so is $\mcC_b^{\infty}$ over $\mcB_{\ve}(x)$. Since $\mcB_\ve(x) \subset (\mcX_h^\circ \setminus \{x_0\})^{n_h}$, all the kernel arguments lie within $(-1,1)\setminus \{0\}$, and therefore $K_s$ is $\mcC_b^\vs$ over $\mcB_\ve(x)$ by Assumption \ref{as kernel}. Similarly, entries for $\mcS_{p,s}$ are sums of products of entries from $\mcH_{p,s}$ and $K_s$, and so $\mcS_{p,s}$ is  $\mcC_b^\vs$ over $\mcB_\ve(x)$. Vectors $e_{p,1}$ and $d_s$ are treated as constant vectors and do not affect differentiability. As $|g_d(x)| \leq nC_k < \infty$ by Lemma \ref{lem gradient upper}, $g_d$ is bounded. By Assumption \ref{as finite sample}(iv) and continuity of the minimum eigenvalue, $\inf_{\tilde{x} \in \cl(\mcB_\ve(x))} \mineig(\mcS_{p,s}(\tilde{x})) > 0$, and so matrix inversion preserves bounded derivatives up to order $\vs$. Since $g_d(x)$ is bounded and is constructed via matrix-vector products, matrix inversion, and linear transformations, all of which preserve bounded derivatives up to order $\vs$, $g_d \in \mcC_b^\vs(\mcB_\ve(x))$. \QEDB
\end{proof}

\begin{lem}\label{lem symmetric existence}
Let Assumptions \ref{as dgp}--\ref{as finite sample} hold, with $n \geq 2\max\{2,p+1\}$ and sufficiently
small $C_{\mcS}>0$. Then, there exists an admissible symmetric first-stage effective sample $(x^*,d^*)$ and some $\ve >0$ such that $\mbP\{x \in \mcB_\ve(x^*), d = d^* \} > 0$, for some neighborhood $\mcB_\ve(x^*)  \subset (\mcX_h^\circ\setminus \{x_0\})^{n_h}$ of $x^*$.
\end{lem}
\begin{proof}[Proof of Lemma \ref{lem symmetric existence}]
Select $\bar{p} = \max\{2,p+1\}$ distinct points $x_{+,1}^*, \hdots, x_{+,\bar{p}}^*$ such that $x_{+,j}^* \in \mcX_{h,+}$ for each $j$, with $x_0 < x_{+,1}^* < \hdots < x_{+,\bar{p}}^* < x_0 + h$, and define $x_{-,j}^* = 2x_0 - x_{+,j}^*$ for each $j$. By construction, $x_{-,j}^* \in \mcX_{h,-}$. As the points are distinct, there exists some sufficiently small $\ve$ such that $\mcB_{\ve}({x_s}^*) \subset \mcX_{h,s}^{\bar{p}}$ for $x_s^* = (x_{s,1}^*, \hdots, x_{s,\bar{p}}^*)'$. This then gives $\mcB_{\ve}(x^*) \subset \mcX^{2\bar{p}}$ for $x^* = ((x_+^*)',(x_-^*)')'$. As the points are also interior, $\mcS_{p,s}(x^*)$ is positive-definite. Then, taking $C_{\mcS}>0$ and $\ve>0$ sufficiently small if necessary, continuity implies that Assumption \ref{as finite sample}(iv) holds throughout $\mcB_{\ve}(x^*)$. Assumptions \ref{as dgp}(i) and \ref{as finite sample}(i) imply $\mbP\{x \in \mcB_{\ve}(x^*)\} > 0$. Define the symmetric treatment configuration $d^*$ by $d_{+,1}^* = d_{-,1}^*=1$ and $d_{+,j}^* =d_{-,j}^* = 0, \quad j=2,\ldots,\bar p$. By Assumption \ref{as finite sample}, $\mbP\{d = d^* | x \in \mcB_\ve(x^*)\} > 0$. Finally, by Assumption \ref{as finite sample}(v), the remaining $n-2\bar p$ observations lie outside $\mcX_h$ with positive probability. Consequently, $\mbP\{x \in \mcB_\ve(x^*), d=d^*\}>0$. On this event, $n_h=2\bar p$ and the effective sample satisfies Assumption \ref{as finite sample}(iii)--(iv). The first-stage effective sample $(x^*,d^*)$ is therefore admissible, symmetric, and has positive probability in every sufficiently small neighborhood of $x^*$. \QEDB
\end{proof}

\begin{lem}\label{lem symmetry}
Let Assumptions \ref{as dgp}--\ref{as finite sample} hold, and suppose the first-stage effective sample $(x^*,d^*)$ is symmetric. Then, $\hat{\tau}_p^D = 0$. 
\end{lem}
\begin{proof}[Proof of Lemma \ref{lem symmetry}]
By the symmetry of the kernel in Assumption \ref{as kernel}, $K_h(x_i^* - x_0) = K_h(-(x_{\Pi(i)}^* - x_0)) = K_h(x_{\Pi(i)}^* - x_0)$, and also $(x_i^* - x_0)^j = (-(x_{\Pi(i)}^* - x_0))^j = (-1)^j(x_{\Pi(i)}^* - x_0)^j$. Recall $\mcS_{p,s} = \sum_{i \in \mcN_s} K_h(x_i^* - x_0) \mcH_{p,i}  \mcH_{p,i}'$, for $\mcH_{p,i}$ the $(p+1)$-vector with $j$-th element $(x_i^* - x_0)^{j-1}$. Let $\mcO_{p+1} = \diag(1,-1,1,\hdots,(-1)^p)$, where the following properties of $\mcO_{p+1}$ are immediate: $\mcO_{p+1} = \mcO_{p+1}'$, $\mcO_{p+1} = \mcO_{p+1}^{-1}$, and $\mcO_{p+1}^2 = I_{p+1}$. Then, $\mcH_{p,i} = \mcO_{p+1} \mcH_{p,\Pi(i)}$, $\mcS_{p,+} = \mcO_{p+1}  \mcS_{p,-} \mcO_{p+1}$, and given $\omega_{p,+,i} = e_{p,1}'\mcS_{p,+}^{-1}\mcH_{p,i}$, then
\begin{align*}
    \omega_{p,+,i} = e_{p,1}'\mcS_{p,+}^{-1}\mcH_{p,i} = e_{p,1}' (\mcO_{p+1}\mcS_{p,-}\mcO_{p+1})^{-1}\mcO_{p+1} \mcH_{p,\Pi(i)} = \omega_{p,-,\Pi(i)},
\end{align*}
where the final equality uses $e_{p,1}'\mcO_{p+1} = e_{p,1}'$. As $K_h(x_i^* - x_0) = K_h(x_{\Pi(i)}^* - x_0)$, $d_i^* = d_{\Pi(i)}^*$ and $\omega_{p,+,i} = \omega_{p,-,\Pi(i)}$ for each $i \in \mcN_+$, $\Pi(i) \in \mcN_-$, then
\begin{align*}
    \hat{\tau}_p^D &= \displaystyle\sum_{i \in \mcN_+} d_i^* K_{h}(x_i^* - x_0)\omega_{p,+,i} - \displaystyle\sum_{i \in \mcN_-} d_i^* K_{h}(x_i^* - x_0)\omega_{p,-,i} \\
    &= \displaystyle\sum_{i \in \mcN_+} d_i^* K_{h}(x_i^* - x_0)\omega_{p,+,i} - \displaystyle\sum_{i \in \mcN_+} d_{\Pi(i)}^* K_{h}(x_{\Pi(i)}^* - x_0)\omega_{p,-,\Pi(i)} = 0 \tag*{\QEDB}
\end{align*}
\end{proof}

\begin{lem}\label{lem bounded away}
Let Assumptions \ref{as dgp}--\ref{as gradient} hold. Then, there exist $\delta_{\tau} > 0$ and $\underline{C}_{\tau} > 0$ such that $f_{\hat{\tau}_p^D}(t_d) \geq \underline{C}_{\tau}$ for almost every $|t_d| < \delta_{\tau}$.
\end{lem}

\begin{proof}[Proof of Lemma \ref{lem bounded away}]
Fix $(x^*,d^*)$ as in Assumption \ref{as gradient}(i). Since $x^*$ lies in the interior of the admissible configuration set, there exists $\varepsilon > 0$ such that $(x,d^*)$ is admissible for every $x \in \mcB_\ve(x^*)$. By the argument in the proof of Lemma \ref{lem symmetric existence}, $\mbP\{x \in \mcB_\ve(x^*),\, d = d^*\} > 0$. Denote the corresponding index sets by $I_+$ and $I_-$, where symmetry implies $|I_+|=|I_-|=n_h/2$. Define $\mcA = \{\mcN_+=I_+,\, \mcN_-=I_-,\, D_{I_+\cup I_-}=d^*\}$. Since $\{x\in\mcB_\ve(x^*),\, d=d^*\}\subseteq\mcA$, it follows that $\mbP\{\mcA\}>0$. For this fixed $d^*$, define $g_{d^*}(x) = \hat{\tau}_{p}^D(x,d^*)$. Then, $g_{d^*}\in \mcC_b^1(\mcB_{\ve}(x^*))$ for some $\ve > 0$ by Lemma \ref{lem differentiability}, $g_{d^*}(x^*) = 0$ by Lemma \ref{lem symmetry}, and $\nabla_x g_{d^*}(x^*) \neq 0$ by Assumption \ref{as gradient}(i).

Suppose without loss of generality that $\partial g_{d^*}(x^*)/\partial x_1 \neq 0$. By continuity, there exists some ball $\mcB_{\vr}(x^*)$ with sufficiently small $\vr > 0$ such that $\cl(\mcB_{\vr}(x^*))\subset \mcB_{\ve}(x^*)$, with $\cl(A)$ denoting the closure of set $A$, where $|g_{d^*}(x)|< \delta_g$ for $x\in \mcB_{\vr}(x^*)$ for some $\delta_g > 0$, and no coordinate crosses the cutoff or bandwidth window edge (this ensures that the index sets $I_+$ and $I_-$ remain constant over $\mcB_{\vr}(x^*)$). Also note that
\begin{align}\label{eq gradient bounds}
    \inf_{x\in \cl(\mcB_{\vr}(x^*))} \norm{\nabla_x g_{d^*}(x)}_2 \geq \underline{C}_{\nabla} > 0, \,\,\, \sup_{x\in \cl(\mcB_{\vr}(x^*))} \norm{\nabla_x g_{d^*}(x)}_2 \leq \overline{C}_{\nabla} < \infty,
\end{align}
where $\underline{C}_{\nabla} = \|\nabla_x g_{d^*}(x^*)\|_2 /2 > 0$, taking $\vr$ smaller if necessary, since $\nabla_x g_{d^*}(x^*) \neq 0$ and $\nabla_x g_{d^*}$ is continuous on $\mcB_\ve(x^*)$ by Lemma \ref{lem differentiability}, and $\overline{C}_{\nabla}$ is the bound in \eqref{eq gradient upper bound} from Lemma \ref{lem gradient upper}. 

Define $\mcM_t = \{x \in \mcB_{\vr}(x^*): g_{d^*}(x) = t\}$ as the level sets of $g_{d^*}$ for $|t|< \delta_g$ within $\mcB_{\vr}(x^*)$. Since $g_{d^*} \in \mcC_b^1(\mcB_\vr(x^*))$ with gradient bounded away from zero over $\mcB_\vr(x^*)$ by \eqref{eq gradient bounds}, the regular value theorem implies that $\mcM_t$ is a $\mcC^1$ $(n_h-1)$-dimensional submanifold (see Theorem 3.2 of \citet{hirsch1976differential}). Consider the map $\Psi: \mcB_{\vr}(x^*)\to \mbR^{n_h}$ defined as $\Psi(x) = (g_{d^*}(x),x_2,\hdots,x_{n_h})$. By the inverse function theorem, there exists some open neighborhood $\mcV$ of $x^*$ on which $\Psi$ is a (local) $\mcC^1$ diffeomorphism \citep{lafontaine2015introduction}. Choose some open neighborhood $\mcB_{\vr'}(x^*)$ with $\vr'<\vr$, and so $\cl(\mcB_{\vr'}(x^*))\subset \mcV \cap \mcB_{\vr}(x^*)$, and set
\begin{align*}
    C_{\Psi} = \sup_{x\in \cl(\mcB_{\vr'}(x^*))} \norm{J\Psi(x)}_{op} < \infty,
\end{align*}
where $J f$ is the Jacobian of $f$, and $\norm{\cdot}_{op}$ is the operator norm. As $\cl(\mcB_{\vr'}(x^*))$ is convex, the mean value inequality gives $\norm{\Psi(x') - \Psi(x'')}_2 \leq C_{\Psi}\norm{x' - x''}_2$, for $x',x''\in \cl(\mcB_{\vr'}(x^*))$. As $\Psi$ is a $\mcC^1$ diffeomorphism with $\Psi(x^*) = (0, x^*_2, \hdots,x^*_{n_h})$, there exists $\delta_g' < \delta_g$ and a non-empty open set $\mcE \subset \mbR^{n_h-1}$ with $(n_h-1)$-dimensional Lebesgue measure $\mcL^{n_h-1}(\mcE) > 0$, such that $(-\delta_g',\delta_g') \times \mcE \subset \Psi(\mcB_{\vr'}(x^*))$. Then, for every $|t| < \delta_g'$, $\Psi^{-1}(\{t\} \times \mcE) \subset \mcM_t \cap \mcB_{\vr'}(x^*)$ gives
\begin{align}\label{eq hausdorff bound}
    H^{n_h-1}(\mcM_t \cap \mcB_{\vr'}(x^*)) \geq H^{n_h-1}(\Psi^{-1}(\{t\} \times \mcE))
    \geq C_{\Psi}^{-(n_h-1)}\mcL^{n_h-1}(\mcE) =: C_{\mcM} > 0,
\end{align}
where $H^\ell(A)$ denotes the $\ell$-dimensional Hausdorff measure of $A$. Assumptions \ref{as finite sample}(i)--(ii) and Bayes' rule imply that the conditional density of $X$ given $\mcA$ is bounded from above and below on $\cl(\mcB_{\vr'}(x^*))$, and thus there is a constant $0 < \underline{C}_f < \infty$ such that $\underline{C}_f \leq f_{X|\mcA}(x|\mcA) $. On $\mcA \cap \{X \in \mcB_{\vr'}(x^*)\}$, $\hat{\tau}_p^D = g_{d^*}(X)$. Hence, for any Borel set $T \subseteq (-\delta_g',\delta_g')$, the coarea formula \citep{Federer1969} together with \eqref{eq gradient bounds} and \eqref{eq hausdorff bound} gives
\begin{align*}
    \mbP\{\hat{\tau}_p^D \in T\} &\geq \mbP\{\mcA\}\int_T \int_{\mcM_t \cap \mcB_{\vr'}(x^*)} \frac{f_{X|\mcA}(x|\mcA)}{\norm{\nabla_x g_{d^*}(x)}_2}\, dH^{n_h-1}(x)\, dt \\
    &\geq \mbP\{\mcA\}\,\underline{C}_f\,\overline{C}_{\nabla}^{-1}C_{\mcM}\, \mcL^1(T).
\end{align*}
As this holds for every such $T$, $f_{\hat{\tau}_p^D}(t_d) \geq \mbP\{\mcA\}\underline{C}_f\overline{C}_{\nabla}^{-1}C_{\mcM} =: \underline{C}_{\tau} > 0$ for almost every $|t_d| < \delta_\tau := \delta_g'$.\footnote{See \citet{negro2024sample} for finite-sample distribution theory using the coarea formula.} \QEDB
\end{proof} 

\begin{lem}\label{lem Gamma finite}
Let Assumptions \ref{as dgp}--\ref{as finite sample} hold. Then, $\Gamma_p \in [C_{\mcS}/2, nC_{K,h}]$.
\end{lem}

\begin{proof}[Proof of Lemma \ref{lem Gamma finite}]
For $s \in \{+,-\}$, $\Gamma_{p,s}=S_{0,s}-\Upsilon_{p,s}$. When $p=0$, under the empty-block convention discussed in the proof of Proposition \ref{prop equivalence iv}, $\Upsilon_{0,s}=0$ and hence $\Gamma_{0,s}=S_{0,s}>0$. For $p\geq1$, $\Gamma_{p,s} = S_{0,s} - \Upsilon_{p,s}$ is the Schur complement of $\mcU_{p,s}$ in $\mcS_{p,s}$. By Assumption \ref{as finite sample}, $\mcS_{p,s}$ is positive-definite. Therefore, $\Gamma_{p,s} > 0$ (see Theorem 7.7.6 of \citet{horn2012matrix}). Since 
$\mcU_{p,s}$ is a principal submatrix of a positive-definite matrix, $\mcU_{p,s}^{-1}$ is also a positive-definite matrix. Consequently, $\Gamma_{p,s} \leq S_{0,s} \leq nC_{K,h} < \infty$. This latter inequality also holds trivially for $p=0$. As $\Gamma_p = \Gamma_{p,+}\Gamma_{p,-}/(\Gamma_{p,+} + \Gamma_{p,-})$, $0 < \Gamma_p \leq nC_{K,h} < \infty$. Further, by the block matrix inversion formula and Assumption \ref{as finite sample}(iv), $\Gamma_{p,s}^{-1} = e_{p,1}'\mcS_{p,s}^{-1}e_{p,1} \leq C_{\mcS}^{-1}$, so $\Gamma_{p,s} \geq C_{\mcS}$ for $s \in \{+,-\}$. As $\Gamma_p = \Gamma_{p,+}\Gamma_{p,-}/ (\Gamma_{p,+} + \Gamma_{p,-})$, it follows that $\Gamma_p^{-1} = \Gamma_{p,+}^{-1} + \Gamma_{p,-}^{-1} \leq 2/C_{\mcS}$. Therefore, $C_{\mcS}/2 \leq \Gamma_p \leq nC_{K,h}$. \QEDB
\end{proof}

\begin{lem}\label{lem Gamma limit}
Let Assumptions \ref{as dgp}--\ref{as kernel}, \ref{as finite sample}(i) and \ref{as asymptotics}(i) hold. Then, $\Gamma_p/n \overset{p}{\to} C_{\Gamma,p}$ for some $C_{\Gamma,p}>0$.
\end{lem}
\begin{proof}[Proof of Lemma \ref{lem Gamma limit}]
For $s \in \{+,-\}$, define $\nu_{j,s} = \int_{\mcI_s} \upsilon^{\,j}K(\upsilon)d\upsilon$ for $\mcI_+ = [0,1]$ and $\mcI_- = [-1,0)$, and let $\vr_{p,s}$ and $\mcV_{p,s}$ be the $p$-vector and $p\times p$ matrix with $k$-th and $k\ell$-th elements $\nu_{k,s}$ and $\nu_{k+\ell,s}$, respectively, and define $\underline{h}_p = \diag(h,\hdots, h^p)$. For $p=0$, these vector and matrix terms are empty, and the quadratic-form term below is interpreted as zero. With Assumption \ref{as asymptotics}(i), standard kernel results give $S_{0,s}/n\overset{p}{\to} f_X(x_0)\nu_{0,s}$, $\underline{h}_p^{-1}\mcR_{p,s}/n\overset{p}{\to} f_X(x_0)\vr_{p,s}$, and $ \underline{h}_p^{-1}\mcU_{p,s}\underline{h}_p^{-1}/n \overset{p}{\to} f_X(x_0)\mcV_{p,s}$. So, $\Gamma_{p,s}/n \overset{p}{\to} f_X(x_0)\left(\nu_{0,s} - \vr_{p,s}'\mcV_{p,s}^{-1}\vr_{p,s}\right) =: C_{\Gamma,p,s}$. Let $\Theta_{p,s} = \int_{\mcI_s} K(\upsilon)\vt_p(\upsilon)\vt_p(\upsilon)'d\upsilon$ for $\vt_p(\upsilon) = (1,\upsilon,\hdots,\upsilon^p)'$. Since $K(\upsilon) > 0$ on $(-1,1)$ by Assumption \ref{as kernel} and no non-zero polynomial can vanish on an open interval, $a'\Theta_{p,s}a = \int_{\mcI_s} K(\upsilon)(a'\vt_p(\upsilon))^2d\upsilon > 0$ for non-zero $a \in \mbR^{p+1}$, so $\Theta_{p,s}$ is positive-definite. For $p \geq 1$, $\nu_{0,s} - \vr_{p,s}'\mcV_{p,s}^{-1}\vr_{p,s}$ is the Schur complement of $\mcV_{p,s}$ in $\Theta_{p,s}$ and $f_X(x_0)>0$ by Assumption \ref{as finite sample}(i), so $C_{\Gamma,p,s} > 0$. When $p = 0$, $C_{\Gamma,p,s} = f_X(x_0)\nu_{0,s} > 0$ directly.  Therefore,
\begin{align*}
    \frac{\Gamma_p}{n} = \frac{(\Gamma_{p,+}/n)(\Gamma_{p,-}/n)}{\Gamma_{p,+}/n + \Gamma_{p,-}/n} \overset{p}{\to} \frac{C_{\Gamma,p,+}C_{\Gamma,p,-}}{C_{\Gamma,p,+}+C_{\Gamma,p,-}} = C_{\Gamma,p}> 0. \tag*{\QEDB}
\end{align*}
\end{proof}

\begin{lem}\label{lem positivity}
Let Assumptions \ref{as dgp}--\ref{as finite sample} and \ref{as deltasample} hold. Then, $\tilde{D}'M_{\tilde{V}_p}\tilde{D} \in [C_\zeta^*, nC_{K,h}]$ for some $C_\zeta^* > 0$.
\end{lem}

\begin{proof}[Proof of Lemma \ref{lem positivity}]
By definition of $M_{\tilde{V}_{p}}$, $\tilde{D}'M_{\tilde{V}_{p}}\tilde{D} = \min_{\beta \in \mbR^{2p+1}} \|\tilde{D} - \tilde{V}_p\beta\|_2^2$. Denote $\tilde{D}_{\mcN}$ and $\tilde{\mcH}_{p,\mcN_+} = K_{\mcN_+}^{1/2}\mcH_{p,\mcN_+}$ as the weighted treatment vector and polynomial design matrix containing only $i \in \mcN_+(\delta, \kappa)$ satisfying Assumption \ref{as deltasample}. Then
\begin{align*}
    \tilde{D}'M_{\tilde{V}_{p}}\tilde{D} = \min_{\beta \in \mbR^{2p+1}} \|\tilde{D} - \tilde{V}_p\beta\|_2^2 \geq \min_{\alpha \in \mbR^{p+1}} \|\tilde{D}_{\mcN_+} - \tilde{\mcH}_{p,\mcN_+}\alpha\|_2^2 = \tilde{D}_{\mcN_+}'M_{\tilde{\mcH}_{p,\mcN_+}}\tilde{D}_{\mcN_+} \geq 0,
\end{align*}
where the final inequality is satisfied since $M_{\tilde{\mcH}_{p,\mcN_+}}$ is a well-defined projection matrix. Here, equality holds if and only if $\tilde{D}_{\mcN_+} \in \textup{col}(\tilde{\mcH}_{p,\mcN_+})$. As $K_h(X_i - x_0) \geq \kappa$ for each $i \in \mcN_+(\delta, \kappa)$ by Assumption \ref{as deltasample}(ii) (and so $K_{\mcN_+}^{1/2}$ is invertible), this implies that $\tilde{D}_{\mcN_+}'M_{\tilde{\mcH}_{p,\mcN_+}}\tilde{D}_{\mcN_+} = 0$ if and only if $D_{\mcN_+} \in \textup{col}(\mcH_{p,\mcN_+})$. Suppose that this is true, so there exists some non-zero $\alpha = (\alpha_0, \hdots, \alpha_p)'\in \mbR^{p+1}$ such that $D_{\mcN_+} = \mcH_{p,\mcN_+}\alpha$ i.e.,
\begin{align}\label{eq D polynomial}
    D_i = \alpha_0 + \sum_{j=1}^p \alpha_j (X_i - x_0)^j = P_+(X_i - x_0)
\end{align}
for $i \in \mcN_+(\delta, \kappa)$ and polynomial $P_+(X_i - x_0)$. As $D_i \in \{0,1\}$, $X_i - x_0$ is a root of either $P_+(x)$ or $P_+(x)-1$. Let $\tilde{P}_+(x) = P_+(x)(P_+(x)-1)$. By \eqref{eq D polynomial}, $P_+(x)$ has degree at most $p$, and so $\tilde{P}_+(x)$ has degree at most $2p$. But there are $n_p > 2p$ distinct values of $X_i - x_0$ for $i \in \mcN_+(\delta, \kappa)$, each a root of $\tilde{P}_+(x)$. Therefore $\tilde{P}_+ \equiv 0$, and so either $P_+(x) = 0$ or $P_+(x) = 1$. But this implies $D_i = 0$ for all $i \in \mcN_+(\delta, \kappa)$ or $D_i = 1$ for all $i \in \mcN_+(\delta, \kappa)$, contradicting Assumption \ref{as deltasample}(iii), and so $\tilde{D}_{\mcN_+}'M_{\tilde{\mcH}_{p,\mcN_+}}\tilde{D}_{\mcN_+} > 0$. Denote by $\Omega_+(\delta, \kappa) \subset \cl(\mcX_{h,+}^{n_p})$ the set
\begin{align*}
    \Omega_+(\delta, \kappa) = \{(x_1,\hdots, x_{n_p})\in \cl(\mcX_{h,+}^{n_p}): |x_i - x_j| \geq \delta \,\, \forall i \neq j, K_h(x_i - x_0) \geq \kappa \,\, \forall i\}.
\end{align*}
Since $|x_i - x_j|$ and $K_h(x_i - x_0)$ are continuous in $(x_1,\hdots,x_{n_p})$, the sets defined by the separation constraints $|x_i - x_j| \geq \delta$ and kernel constraints $K_h(x_i - x_0) \geq \kappa$ are closed. Hence $\Omega_+(\delta,\kappa)$ is a closed subset of the compact set $\cl(\mcX_{h,+}^{n_p})$, and is therefore compact. $\tilde{D}_{\mcN_+}' M_{\tilde{\mcH}_{p,\mcN_+}} \tilde{D}_{\mcN_+}$ is continuous in $X_{\mcN_+}$, as $\mcH_{p,\mcN}$ is a full column rank matrix of polynomials in $X_i-x_0$ and the kernel $K_h(X_i-x_0)$ is continuous and bounded away from zero. For admissible $\tilde{D}_{\mcN_+}$, $C_{\zeta, D_{\mcN_+}} = \min_{X_{\mcN_+} \in \Omega_+(\delta, \kappa)} \tilde{D}_{\mcN_+}'M_{\tilde{\mcH}_{p,\mcN_+}}\tilde{D}_{\mcN_+} > 0$ is well-defined.
Let $\mathscr{D}_p \subset \{0,1\}^{n_p}$ be the set of admissible values of $D_{\mcN_+}$ given Assumption \ref{as deltasample}. $\mathscr{D}_p$ is finite, so  $C_\zeta^* = \min_{D_{\mcN_+} \in \mathscr{D}_p} C_{\zeta, D_{\mcN_+}} > 0$ is well-defined. Then, $\tilde{D}'M_{\tilde{V}_p}\tilde{D} \geq \tilde{D}_{\mcN_+}'M_{\tilde{\mcH}_{p,\mcN_+}}\tilde{D}_{\mcN_+} \geq C_\zeta^* > 0.$ The upper bound follows as $\tilde{D}'M_{\tilde{V}_p}\tilde{D} \leq \tilde{D}'\tilde{D} \leq nC_{K,h}$. \QEDB
\end{proof}

\noindent Define $s_{n,i} = (k_{p,i}^0 \bar{Y}_i, \,\, k_{p,i}^0 \bar{D}_i, \,\, \kappa_n k_{q,p+1,i}^{B,+,0} \bar{Y}_i, \,\, \kappa_n k_{q,p+1,i}^{B,-,0} \bar{Y}_i, \,\, \kappa_n k_{q,p+1,i}^{B,+,0} \bar{D}_i, \,\, \kappa_n k_{q,p+1,i}^{B,-,0} \bar{D}_i)'$, where $\kappa_n = (h/h_b)^{p+2}$, $k_{p,i}^0$ denotes the signed population-equivalent kernel contribution of the finite-sample generalized kernel weights $k_{p,i}$ (as defined in Section \ref{sec model}), and $k_{q,p+1,i}^{B,+,0}$ and $k_{q,p+1,i}^{B,-,0}$ denote the right- and left-sided population analogues of the generalized kernel weights from the $q$-th order auxiliary regression with bandwidth $h_b$ associated with the coefficient on $(X_i-x_0)^{p+1}$. Explicitly, with $\upsilon_i = (X_i - x_0)/h$ and $v_i = (X_i - x_0)/h_b$,
\begin{gather*}
    k_{p,i}^0 = \frac{K(\upsilon_i)}{f_X(x_0)}\left(e_{p,1}' \Theta_{p,+}^{-1}\vt_p(\upsilon_i)1\{\upsilon_i \geq 0\} - e_{p,1}'\Theta_{p,-}^{-1}\vt_p(\upsilon_i)1\{\upsilon_i < 0\} \right), \\
    k_{q,p+1,i}^{B,s,0} = \frac{K(v_i)}{f_X(x_0)}e_{q,p+2}'\Theta_{q,s}^{-1}\vt_q(v_i)1\{v_i \in \mcI_s\},
\end{gather*}
for $s \in \{+,-\}$, with $\int_{\mcI_s} K(\upsilon)e_{p,1}'\Theta_{p,s}^{-1}\vt_p(\upsilon)d\upsilon = 1$ and $\int_{\mcI_s} K(v)e_{q,p+2}'\Theta_{q,s}^{-1}\vt_q(v)dv = 0$ .

\begin{lem}\label{lem influence functions}
Let Assumptions \ref{as dgp}--\ref{as kernel}, \ref{as finite sample}(i)--(ii) and \ref{as asymptotics}--\ref{as bias} hold. Then, $\Delta$ admits the joint asymptotic linear representation $\sqrt{nh} \Delta = (nh)^{-1/2} \sumin s_{n,i} + o_p(1)$. Further, $\mbE[s_{n,i}] = 0$ and for some finite covariance matrix $\Sigma$, $\Sigma_n = (nh)^{-1} \sumin \mbE[s_{n,i}s_{n,i}'] \to \Sigma$.
\end{lem}

\begin{proof}[Proof of Lemma \ref{lem influence functions}]
A standard boundary local-polynomial expansion gives
\begin{align*}
    \hat{\tau}_p^W - \tau^W = \mcB_p^W(h) + \frac{1}{nh}\sumin k_{p,i}^0 \bar{W}_i + o_p((nh)^{-1/2}),
\end{align*}
as $nh^{2p+3} = O(1)$ by Assumption \ref{as asymptotics}(iv). The $q$-th order estimator of a $(p+1)$-th order derivative has approximation bias of order $O(h_b^{q-p})$, hence for $s \in \{+,-\}$
\begin{align}\label{eq phi bias expansion}
    \hat{\Phi}_{p+1,q,s}^W - \Phi_{p+1,s}^W = \frac{h^{p+1}}{nh_b^{p+2}}\sumin k_{q,p+1,i}^{B,s,0}\bar{W}_i + O(h^{p+1}h_b^{q-p}) + o_p\left(\frac{h^{p+1}}{\sqrt{nh_b^{2p+3}}} \right).
\end{align}
As $h^{p+1}/(nh_b^{p+2}) = \kappa_n/(nh)$, the leading term in \eqref{eq phi bias expansion} is $(nh)^{-1}\sumin \kappa_n k_{q,p+1,i}^{B,s,0}\bar{W}_i$.
Further, $\sqrt{nh}h^{p+1} = O(1)$ and Assumption \ref{as bias}(i) imply $\sqrt{nh}h^{p+1}h_b^{q-p} \to 0$, and $\sqrt{nh}(h^{p+1})(nh_b^{2p+3})^{-1/2} = (h/h_b)^{p+3/2} \to c_h^{(p+3/2)}$, which is finite and positive. Therefore, \eqref{eq phi bias expansion} becomes
\begin{align}\label{eq phi bias expansion simplified}
    \hat{\Phi}_{p+1,q,s}^W - \Phi_{p+1,s}^W = \frac{1}{nh}\sumin \kappa_n k_{q,p+1,i}^{B,s,0}\bar{W}_i + o_p((nh)^{-1/2}).
\end{align}
Stacking these expressions then yields $\sqrt{nh}\Delta = (nh)^{-1/2}\sum_i s_{n,i} + o_p(1)$. Since the residuals are conditionally mean zero, $\mbE[\bar{W}_i|X_i] = 0$, and so $\mbE[s_{n,i}] = 0$. As $h/h_b \to c_h \in (0,\infty)$, $\kappa_n \to c_h^{p+2}$ and each element of $s_{n,i}$ is supported within an $O(h)$ neighborhood of $x_0$. By Assumptions \ref{as kernel}, \ref{as asymptotics}, and \ref{as bias}, each entry of $\mbE[s_{n,i}s_{n,i}']$ is $O(h)$ uniformly in $i$, so $\Sigma_n = O(1)$.  For each element of $\Sigma_n$, a change of variables $v=(X_i-x_0)/h$, together with $h/h_b \to c_h$, expresses the limit as an integral of products of bounded equivalent-kernel functions and the corresponding limiting conditional covariance of $(\bar{Y}_i,\bar{D}_i)$. Assumption \ref{as bias}(iii) and dominated convergence therefore give $\Sigma_n = (nh)^{-1}\sumin \mbE[s_{n,i}s_{n,i}'] \to \Sigma$ for some finite covariance matrix $\Sigma$. \QEDB
\end{proof}

\clearpage

\appendix
\begin{center}
{\Large Online Appendix for}\\
\vspace{0.5cm}
{\Large The moment is here: a generalized class of estimators for fuzzy regression discontinuity designs}\\
\vspace{0.5cm}
{\large Stuart Lane}\\
\vspace{0.5cm}
{University of Bristol, stuart.lane@bristol.ac.uk}
\end{center}

\setcounter{table}{0}
\renewcommand{\thesection}{O\Alph{section}}
\renewcommand{\thetable}{OA\arabic{table}}

\vspace{20mm}

\section{Additional simulation results}\label{sec online appendix}

\noindent Here I present results repeating the experiment in the main paper, but allowing for endogenous selection with $\rho = \pm 0.2$ (weak selection) and $\rho = \pm 0.8$ (strong selection). 

\clearpage

% =====================================================================
%  Weak endogeneity: \citet{lee2008randomized} design
% =====================================================================
\begin{table}[htbp]
\centering
\begin{threeparttable}
\caption{Weak endogeneity: \citet{lee2008randomized} design}
\label{table dgp1_rho02}
\footnotesize
\setlength{\tabcolsep}{3pt}
\begin{tabular}{l c c c c c c c c c c c c c c}
\toprule
& \multicolumn{7}{c}{MSE-optimal bandwidth} & \multicolumn{7}{c}{CER-optimal bandwidth} \\
\cmidrule(lr){2-8}\cmidrule(lr){9-15}
& \multicolumn{4}{c}{Point estimator} & \multicolumn{3}{c}{Nominal 95\% CI} & \multicolumn{4}{c}{Point estimator} & \multicolumn{3}{c}{Nominal 95\% CI} \\
\cmidrule(lr){2-5}\cmidrule(lr){6-8}\cmidrule(lr){9-12}\cmidrule(lr){13-15}
Est. & {B} & {MB} & {MAE} & {RMSE} & {Cov.} & {MdL} & {MeL} & {B} & {MB} & {MAE} & {RMSE} & {Cov.} & {MdL} & {MeL} \\
\midrule
\multicolumn{15}{l}{\underline{\emph{$\rho = 0.2$, $n = 500$, $\tilde{F}_{\MSE} = 9.8$, $\tilde{F}_{\CER} = 7.0$}}} \\[1pt]
  $\hat{\tau}_1$ & 0.10 & 0.09 & 0.20 & 1.84 & 96.9 & 1.08 & 24.4 & $-0.93$ & 0.07 & 0.23 & 98.3 & 98.0 & 1.30 & \multicolumn{1}{c}{$>10^{4}$} \\
  $\hat{\tau}_{\lambda^*}$ & 0.09 & 0.09 & 0.16 & 0.25 & 93.8 & 0.89 & 0.95 & 0.08 & 0.08 & 0.17 & 0.26 & 94.7 & 0.98 & 1.04 \\
  $\hat{\tau}_1^{\BC}$ & 0.44 & 0.05 & 0.28 & 87.1 & 96.7 & 1.61 & 36.9 & \multicolumn{1}{c}{$>10^{4}$} & 0.03 & 0.33 & \multicolumn{1}{c}{$>10^{6}$} & 96.8 & 1.94 & \multicolumn{1}{c}{$>10^{4}$} \\
  $\hat{\tau}_{\lambda^*}^{\BC}$ & 0.05 & 0.06 & 0.23 & 0.41 & 95.6 & 1.33 & 1.45 & 0.04 & 0.05 & 0.24 & 0.44 & 95.4 & 1.46 & 1.59 \\
  $\hat{\tau}_{\lambda^*_{\BC}}^{\BC}$ & 0.05 & 0.06 & 0.21 & 0.35 & 95.2 & 1.21 & 1.27 & 0.04 & 0.06 & 0.22 & 0.37 & 94.9 & 1.29 & 1.36 \\
  $AR_{\BA}$ & \multicolumn{1}{c}{--} & \multicolumn{1}{c}{--} & \multicolumn{1}{c}{--} & \multicolumn{1}{c}{--} & 96.6 & \multicolumn{1}{c}{1.64} & \multicolumn{1}{c}{58.1\%} & \multicolumn{1}{c}{--} & \multicolumn{1}{c}{--} & \multicolumn{1}{c}{--} & \multicolumn{1}{c}{--} & 96.6 & \multicolumn{1}{c}{1.64} & \multicolumn{1}{c}{58.1\%} \\
\addlinespace[3pt]
\multicolumn{15}{l}{\underline{\emph{$\rho = 0.2$, $n = 1500$, $\tilde{F}_{\MSE} = 29.1$, $\tilde{F}_{\CER} = 20.2$}}} \\[1pt]
  $\hat{\tau}_1$ & 0.08 & 0.08 & 0.12 & 0.21 & 93.0 & 0.61 & 0.74 & 0.05 & 0.06 & 0.14 & 0.56 & 96.0 & 0.74 & 3.53 \\
  $\hat{\tau}_{\lambda^*}$ & 0.08 & 0.08 & 0.12 & 0.18 & 91.8 & 0.58 & 0.61 & 0.06 & 0.06 & 0.12 & 0.19 & 94.5 & 0.68 & 0.73 \\
  $\hat{\tau}_1^{\BC}$ & 0.03 & 0.04 & 0.16 & 0.49 & 95.5 & 0.90 & 1.08 & $-0.41$ & 0.01 & 0.19 & 37.7 & 95.8 & 1.08 & 5.23 \\
  $\hat{\tau}_{\lambda^*}^{\BC}$ & 0.03 & 0.04 & 0.15 & 0.25 & 95.3 & 0.85 & 0.91 & 0.01 & 0.02 & 0.17 & 0.30 & 95.2 & 1.00 & 1.09 \\
  $\hat{\tau}_{\lambda^*_{\BC}}^{\BC}$ & 0.04 & 0.04 & 0.14 & 0.24 & 95.1 & 0.82 & 0.87 & 0.02 & 0.02 & 0.17 & 0.28 & 95.1 & 0.96 & 1.02 \\
  $AR_{\BA}$ & \multicolumn{1}{c}{--} & \multicolumn{1}{c}{--} & \multicolumn{1}{c}{--} & \multicolumn{1}{c}{--} & 96.8 & \multicolumn{1}{c}{1.12} & \multicolumn{1}{c}{18.4\%} & \multicolumn{1}{c}{--} & \multicolumn{1}{c}{--} & \multicolumn{1}{c}{--} & \multicolumn{1}{c}{--} & 96.8 & \multicolumn{1}{c}{1.12} & \multicolumn{1}{c}{18.4\%} \\
\addlinespace[3pt]
\multicolumn{15}{l}{\underline{\emph{$\rho = -0.2$, $n = 500$, $\tilde{F}_{\MSE} = 9.8$, $\tilde{F}_{\CER} = 7.1$}}} \\[1pt]
  $\hat{\tau}_1$ & 0.16 & 0.09 & 0.20 & 4.33 & 98.6 & 1.12 & 113 & $-0.37$ & 0.06 & 0.23 & 36.6 & 98.8 & 1.32 & \multicolumn{1}{c}{$>10^{3}$} \\
  $\hat{\tau}_{\lambda^*}$ & 0.06 & 0.05 & 0.16 & 0.25 & 96.3 & 0.91 & 0.98 & 0.02 & 0.02 & 0.16 & 0.25 & 96.1 & 0.99 & 1.06 \\
  $\hat{\tau}_1^{\BC}$ & 13.9 & 0.03 & 0.28 & 734 & 96.7 & 1.65 & 169 & $-410$ & 0.00 & 0.33 & \multicolumn{1}{c}{$>10^{4}$} & 96.8 & 1.96 & \multicolumn{1}{c}{$>10^{3}$} \\
  $\hat{\tau}_{\lambda^*}^{\BC}$ & 0.03 & 0.01 & 0.23 & 0.42 & 95.7 & 1.35 & 1.51 & 0.00 & $-0.02$ & 0.24 & 0.45 & 95.2 & 1.48 & 1.64 \\
  $\hat{\tau}_{\lambda^*_{\BC}}^{\BC}$ & 0.02 & 0.00 & 0.21 & 0.35 & 95.4 & 1.24 & 1.32 & $-0.01$ & $-0.03$ & 0.22 & 0.37 & 95.0 & 1.31 & 1.39 \\
  $AR_{\BA}$ & \multicolumn{1}{c}{--} & \multicolumn{1}{c}{--} & \multicolumn{1}{c}{--} & \multicolumn{1}{c}{--} & 96.7 & \multicolumn{1}{c}{1.66} & \multicolumn{1}{c}{58.0\%} & \multicolumn{1}{c}{--} & \multicolumn{1}{c}{--} & \multicolumn{1}{c}{--} & \multicolumn{1}{c}{--} & 96.7 & \multicolumn{1}{c}{1.66} & \multicolumn{1}{c}{58.0\%} \\
\addlinespace[3pt]
\multicolumn{15}{l}{\underline{\emph{$\rho = -0.2$, $n = 1500$, $\tilde{F}_{\MSE} = 29.1$, $\tilde{F}_{\CER} = 20.2$}}} \\[1pt]
  $\hat{\tau}_1$ & 0.10 & 0.08 & 0.12 & 0.34 & 96.4 & 0.63 & 0.95 & 0.08 & 0.06 & 0.13 & 1.41 & 97.8 & 0.76 & 8.22 \\
  $\hat{\tau}_{\lambda^*}$ & 0.08 & 0.07 & 0.11 & 0.18 & 95.3 & 0.59 & 0.63 & 0.05 & 0.04 & 0.12 & 0.19 & 96.7 & 0.69 & 0.74 \\
  $\hat{\tau}_1^{\BC}$ & 0.16 & 0.03 & 0.15 & 7.22 & 96.1 & 0.92 & 1.41 & 3.16 & 0.01 & 0.19 & 323 & 96.0 & 1.11 & 11.1 \\
  $\hat{\tau}_{\lambda^*}^{\BC}$ & 0.04 & 0.02 & 0.14 & 0.26 & 95.7 & 0.87 & 0.95 & 0.02 & 0.00 & 0.17 & 0.31 & 95.6 & 1.02 & 1.12 \\
  $\hat{\tau}_{\lambda^*_{\BC}}^{\BC}$ & 0.03 & 0.02 & 0.14 & 0.24 & 95.6 & 0.85 & 0.91 & 0.01 & 0.00 & 0.17 & 0.28 & 95.4 & 0.98 & 1.04 \\
  $AR_{\BA}$ & \multicolumn{1}{c}{--} & \multicolumn{1}{c}{--} & \multicolumn{1}{c}{--} & \multicolumn{1}{c}{--} & 96.9 & \multicolumn{1}{c}{1.15} & \multicolumn{1}{c}{17.2\%} & \multicolumn{1}{c}{--} & \multicolumn{1}{c}{--} & \multicolumn{1}{c}{--} & \multicolumn{1}{c}{--} & 96.9 & \multicolumn{1}{c}{1.15} & \multicolumn{1}{c}{17.2\%} \\
\bottomrule
\end{tabular}
\begin{tablenotes}[flushleft]
\footnotesize
\item \textit{Notes:} 10,000 replications per parameter configuration. Weak, medium and strong endogeneity correspond to $|\rho| = 0.2$, $0.5$, and $0.8$ respectively. The probability discontinuity at the cutoff is $\tau^D = 0.4$, $n$ is the sample size and $\tilde{F}$ is the median first-stage $F$-statistic within the bandwidth window. All estimators use the triangular kernel, $p=1$, and $q=2$ and $h=h_b$ where appropriate. $\hat\tau_1$ and $\hat\tau_1^{\BC}$ denote the conventional and bias-corrected estimators, $\hat\tau_{\lambda^*}$ and $\hat\tau_{\lambda^*}^{\BC}$ denote the conventional and bias-corrected $\lambda$-class estimators with $\lambda = \lambda^*$, and $\hat\tau_{\lambda^*_{\BC}}^{\BC}$ is the bias-corrected $\lambda$-class estimator with $\lambda = \lambda^*_{\BC}$. All are computed with MSE- and CER-optimal bandwidths. B, MB, MAE, and RMSE denote bias, median bias, median absolute error, and root mean squared error. Cov., MdL and MeL denote coverage, median length, and mean length. For $AR_{\BA}$, MdL is the median length over bounded, nonempty sets only, and so understates typical length when many sets are unbounded, MeL reports the percentage of unbounded sets. $AR_{\BA}$ uses its own bandwidth, so its results do not depend on  either the MSE or CER-optimal bandwidth rule. The smoothness bound is the \texttt{RDHonest} default rule of thumb. Values of absolute size 1000 or more are reported as $>10^k$ (or $<-10^k$), where $k$ is their order of magnitude.
\end{tablenotes}
\end{threeparttable}
\end{table}

% =====================================================================
%  Weak endogeneity: \citet{ludwig2007does} design
% =====================================================================
\begin{table}[htbp]
\centering
\begin{threeparttable}
\caption{Weak endogeneity: \citet{ludwig2007does} design}
\label{table dgp2_rho02}
\footnotesize
\setlength{\tabcolsep}{3pt}
\begin{tabular}{l c c c c c c c c c c c c c c}
\toprule
& \multicolumn{7}{c}{MSE-optimal bandwidth} & \multicolumn{7}{c}{CER-optimal bandwidth} \\
\cmidrule(lr){2-8}\cmidrule(lr){9-15}
& \multicolumn{4}{c}{Point estimator} & \multicolumn{3}{c}{Nominal 95\% CI} & \multicolumn{4}{c}{Point estimator} & \multicolumn{3}{c}{Nominal 95\% CI} \\
\cmidrule(lr){2-5}\cmidrule(lr){6-8}\cmidrule(lr){9-12}\cmidrule(lr){13-15}
Est. & {B} & {MB} & {MAE} & {RMSE} & {Cov.} & {MdL} & {MeL} & {B} & {MB} & {MAE} & {RMSE} & {Cov.} & {MdL} & {MeL} \\
\midrule
\multicolumn{15}{l}{\underline{\emph{$\rho = 0.2$, $n = 500$, $\tilde{F}_{\MSE} = 9.3$, $\tilde{F}_{\CER} = 6.7$}}} \\[1pt]
  $\hat{\tau}_1$ & 1.24 & 0.87 & 0.89 & 10.6 & 55.9 & 1.66 & 495 & $-0.08$ & 0.53 & 0.56 & 92.0 & 82.6 & 1.58 & \multicolumn{1}{c}{$>10^{4}$} \\
  $\hat{\tau}_{\lambda^*}$ & 0.75 & 0.71 & 0.71 & 0.87 & 29.3 & 1.15 & 1.34 & 0.43 & 0.41 & 0.42 & 0.55 & 63.8 & 1.08 & 1.22 \\
  $\hat{\tau}_1^{\BC}$ & 151 & 0.14 & 0.34 & \multicolumn{1}{c}{$>10^{4}$} & 98.0 & 2.49 & 775 & \multicolumn{1}{c}{$>10^{3}$} & 0.06 & 0.37 & \multicolumn{1}{c}{$>10^{5}$} & 97.1 & 2.36 & \multicolumn{1}{c}{$>10^{4}$} \\
  $\hat{\tau}_{\lambda^*}^{\BC}$ & 0.23 & 0.14 & 0.29 & 0.67 & 97.0 & 1.95 & 2.44 & 0.13 & 0.08 & 0.27 & 0.54 & 96.3 & 1.72 & 2.04 \\
  $\hat{\tau}_{\lambda^*_{\BC}}^{\BC}$ & 0.20 & 0.14 & 0.27 & 0.55 & 96.6 & 1.74 & 2.06 & 0.12 & 0.08 & 0.24 & 0.44 & 95.9 & 1.51 & 1.69 \\
  $AR_{\BA}$ & \multicolumn{1}{c}{--} & \multicolumn{1}{c}{--} & \multicolumn{1}{c}{--} & \multicolumn{1}{c}{--} & 97.4 & \multicolumn{1}{c}{2.08} & \multicolumn{1}{c}{64.7\%} & \multicolumn{1}{c}{--} & \multicolumn{1}{c}{--} & \multicolumn{1}{c}{--} & \multicolumn{1}{c}{--} & 97.4 & \multicolumn{1}{c}{2.08} & \multicolumn{1}{c}{64.7\%} \\
\addlinespace[3pt]
\multicolumn{15}{l}{\underline{\emph{$\rho = 0.2$, $n = 1500$, $\tilde{F}_{\MSE} = 25.5$, $\tilde{F}_{\CER} = 17.6$}}} \\[1pt]
  $\hat{\tau}_1$ & 0.89 & 0.72 & 0.72 & 9.32 & 17.5 & 0.90 & 547 & 0.47 & 0.41 & 0.41 & 1.65 & 60.7 & 0.88 & 11.0 \\
  $\hat{\tau}_{\lambda^*}$ & 0.72 & 0.67 & 0.67 & 0.81 & 11.7 & 0.79 & 0.88 & 0.39 & 0.37 & 0.37 & 0.47 & 51.4 & 0.78 & 0.84 \\
  $\hat{\tau}_1^{\BC}$ & 104 & 0.10 & 0.20 & \multicolumn{1}{c}{$>10^{4}$} & 98.0 & 1.32 & 802 & 3.02 & 0.03 & 0.21 & 269 & 97.0 & 1.29 & 17.0 \\
  $\hat{\tau}_{\lambda^*}^{\BC}$ & 0.14 & 0.10 & 0.19 & 0.37 & 97.8 & 1.22 & 1.42 & 0.06 & 0.04 & 0.19 & 0.35 & 96.7 & 1.16 & 1.32 \\
  $\hat{\tau}_{\lambda^*_{\BC}}^{\BC}$ & 0.14 & 0.10 & 0.18 & 0.34 & 97.6 & 1.17 & 1.33 & 0.07 & 0.05 & 0.18 & 0.31 & 96.5 & 1.11 & 1.21 \\
  $AR_{\BA}$ & \multicolumn{1}{c}{--} & \multicolumn{1}{c}{--} & \multicolumn{1}{c}{--} & \multicolumn{1}{c}{--} & 97.1 & \multicolumn{1}{c}{1.63} & \multicolumn{1}{c}{36.7\%} & \multicolumn{1}{c}{--} & \multicolumn{1}{c}{--} & \multicolumn{1}{c}{--} & \multicolumn{1}{c}{--} & 97.1 & \multicolumn{1}{c}{1.63} & \multicolumn{1}{c}{36.7\%} \\
\addlinespace[3pt]
\multicolumn{15}{l}{\underline{\emph{$\rho = -0.2$, $n = 500$, $\tilde{F}_{\MSE} = 9.2$, $\tilde{F}_{\CER} = 6.7$}}} \\[1pt]
  $\hat{\tau}_1$ & $-0.02$ & 0.87 & 0.89 & 71.8 & 70.6 & 1.92 & \multicolumn{1}{c}{$>10^{4}$} & 1.30 & 0.53 & 0.57 & 55.0 & 92.4 & 1.79 & \multicolumn{1}{c}{$>10^{3}$} \\
  $\hat{\tau}_{\lambda^*}$ & 0.71 & 0.68 & 0.68 & 0.85 & 39.2 & 1.25 & 1.50 & 0.37 & 0.36 & 0.37 & 0.52 & 74.7 & 1.14 & 1.36 \\
  $\hat{\tau}_1^{\BC}$ & \multicolumn{1}{c}{$>10^{3}$} & 0.12 & 0.37 & \multicolumn{1}{c}{$>10^{5}$} & 98.3 & 2.84 & \multicolumn{1}{c}{$>10^{4}$} & \multicolumn{1}{c}{$>10^{3}$} & 0.03 & 0.39 & \multicolumn{1}{c}{$>10^{5}$} & 97.2 & 2.68 & \multicolumn{1}{c}{$>10^{3}$} \\
  $\hat{\tau}_{\lambda^*}^{\BC}$ & 0.21 & 0.09 & 0.30 & 0.76 & 97.3 & 2.22 & 2.82 & 0.08 & 0.01 & 0.29 & 0.60 & 96.3 & 1.92 & 2.33 \\
  $\hat{\tau}_{\lambda^*_{\BC}}^{\BC}$ & 0.16 & 0.08 & 0.27 & 0.60 & 97.1 & 1.98 & 2.35 & 0.05 & 0.00 & 0.26 & 0.47 & 95.9 & 1.67 & 1.91 \\
  $AR_{\BA}$ & \multicolumn{1}{c}{--} & \multicolumn{1}{c}{--} & \multicolumn{1}{c}{--} & \multicolumn{1}{c}{--} & 98.2 & \multicolumn{1}{c}{2.17} & \multicolumn{1}{c}{64.0\%} & \multicolumn{1}{c}{--} & \multicolumn{1}{c}{--} & \multicolumn{1}{c}{--} & \multicolumn{1}{c}{--} & 98.2 & \multicolumn{1}{c}{2.17} & \multicolumn{1}{c}{64.0\%} \\
\addlinespace[3pt]
\multicolumn{15}{l}{\underline{\emph{$\rho = -0.2$, $n = 1500$, $\tilde{F}_{\MSE} = 25.4$, $\tilde{F}_{\CER} = 17.4$}}} \\[1pt]
  $\hat{\tau}_1$ & 0.80 & 0.72 & 0.72 & 1.29 & 26.8 & 1.04 & 3.51 & 0.49 & 0.40 & 0.40 & 2.24 & 74.3 & 0.99 & 21.2 \\
  $\hat{\tau}_{\lambda^*}$ & 0.72 & 0.66 & 0.66 & 0.81 & 17.6 & 0.90 & 1.00 & 0.38 & 0.35 & 0.35 & 0.47 & 63.7 & 0.85 & 0.92 \\
  $\hat{\tau}_1^{\BC}$ & 0.57 & 0.09 & 0.21 & 39.4 & 98.5 & 1.52 & 5.21 & 7.35 & 0.02 & 0.22 & 652 & 97.3 & 1.45 & 30.5 \\
  $\hat{\tau}_{\lambda^*}^{\BC}$ & 0.14 & 0.08 & 0.20 & 0.42 & 98.2 & 1.39 & 1.65 & 0.07 & 0.02 & 0.20 & 0.41 & 96.8 & 1.29 & 1.52 \\
  $\hat{\tau}_{\lambda^*_{\BC}}^{\BC}$ & 0.13 & 0.08 & 0.19 & 0.38 & 98.0 & 1.33 & 1.54 & 0.05 & 0.01 & 0.19 & 0.35 & 96.7 & 1.22 & 1.38 \\
  $AR_{\BA}$ & \multicolumn{1}{c}{--} & \multicolumn{1}{c}{--} & \multicolumn{1}{c}{--} & \multicolumn{1}{c}{--} & 97.6 & \multicolumn{1}{c}{1.73} & \multicolumn{1}{c}{34.1\%} & \multicolumn{1}{c}{--} & \multicolumn{1}{c}{--} & \multicolumn{1}{c}{--} & \multicolumn{1}{c}{--} & 97.6 & \multicolumn{1}{c}{1.73} & \multicolumn{1}{c}{34.1\%} \\
\bottomrule
\end{tabular}
\begin{tablenotes}[flushleft]
\footnotesize
\item \textit{Notes:} 10,000 replications per parameter configuration. Weak, medium and strong endogeneity correspond to $|\rho| = 0.2$, $0.5$, and $0.8$ respectively. The probability discontinuity at the cutoff is $\tau^D = 0.4$, $n$ is the sample size and $\tilde{F}$ is the median first-stage $F$-statistic within the bandwidth window. All estimators use the triangular kernel, $p=1$, and $q=2$ and $h=h_b$ where appropriate. $\hat\tau_1$ and $\hat\tau_1^{\BC}$ denote the conventional and bias-corrected estimators, $\hat\tau_{\lambda^*}$ and $\hat\tau_{\lambda^*}^{\BC}$ denote the conventional and bias-corrected $\lambda$-class estimators with $\lambda = \lambda^*$, and $\hat\tau_{\lambda^*_{\BC}}^{\BC}$ is the bias-corrected $\lambda$-class estimator with $\lambda = \lambda^*_{\BC}$. All are computed with MSE- and CER-optimal bandwidths. B, MB, MAE, and RMSE denote bias, median bias, median absolute error, and root mean squared error. Cov., MdL and MeL denote coverage, median length, and mean length. For $AR_{\BA}$, MdL is the median length over bounded, nonempty sets only, and so understates typical length when many sets are unbounded, MeL reports the percentage of unbounded sets. $AR_{\BA}$ uses its own bandwidth, so its results do not depend on  either the MSE or CER-optimal bandwidth rule. The smoothness bound is the \texttt{RDHonest} default rule of thumb. Values of absolute size 1000 or more are reported as $>10^k$ (or $<-10^k$), where $k$ is their order of magnitude.
\end{tablenotes}
\end{threeparttable}
\end{table}

% =====================================================================
%  Strong endogeneity: \citet{lee2008randomized} design
% =====================================================================
\begin{table}[htbp]
\centering
\begin{threeparttable}
\caption{Strong endogeneity: \citet{lee2008randomized} design}
\label{table dgp1_rho08}
\footnotesize
\setlength{\tabcolsep}{3pt}
\begin{tabular}{l c c c c c c c c c c c c c c}
\toprule
& \multicolumn{7}{c}{MSE-optimal bandwidth} & \multicolumn{7}{c}{CER-optimal bandwidth} \\
\cmidrule(lr){2-8}\cmidrule(lr){9-15}
& \multicolumn{4}{c}{Point estimator} & \multicolumn{3}{c}{Nominal 95\% CI} & \multicolumn{4}{c}{Point estimator} & \multicolumn{3}{c}{Nominal 95\% CI} \\
\cmidrule(lr){2-5}\cmidrule(lr){6-8}\cmidrule(lr){9-12}\cmidrule(lr){13-15}
Est. & {B} & {MB} & {MAE} & {RMSE} & {Cov.} & {MdL} & {MeL} & {B} & {MB} & {MAE} & {RMSE} & {Cov.} & {MdL} & {MeL} \\
\midrule
\multicolumn{15}{l}{\underline{\emph{$\rho = 0.8$, $n = 500$, $\tilde{F}_{\MSE} = 9.8$, $\tilde{F}_{\CER} = 7.1$}}} \\[1pt]
  $\hat{\tau}_1$ & 0.13 & 0.10 & 0.20 & 7.56 & 89.4 & 0.97 & 302 & $-0.05$ & 0.09 & 0.23 & 6.52 & 91.1 & 1.17 & 359 \\
  $\hat{\tau}_{\lambda^*}$ & 0.14 & 0.15 & 0.18 & 0.25 & 83.8 & 0.77 & 0.84 & 0.16 & 0.16 & 0.19 & 0.27 & 84.9 & 0.84 & 0.94 \\
  $\hat{\tau}_1^{\BC}$ & $-18.1$ & 0.06 & 0.27 & \multicolumn{1}{c}{$>10^{3}$} & 94.1 & 1.44 & 458 & 21.6 & 0.06 & 0.32 & \multicolumn{1}{c}{$>10^{3}$} & 94.4 & 1.72 & 508 \\
  $\hat{\tau}_{\lambda^*}^{\BC}$ & 0.07 & 0.12 & 0.23 & 0.40 & 90.9 & 1.18 & 1.36 & 0.08 & 0.14 & 0.26 & 0.43 & 90.8 & 1.30 & 1.50 \\
  $\hat{\tau}_{\lambda^*_{\BC}}^{\BC}$ & 0.11 & 0.14 & 0.23 & 0.34 & 89.2 & 1.07 & 1.18 & 0.13 & 0.16 & 0.24 & 0.36 & 88.6 & 1.14 & 1.26 \\
  $AR_{\BA}$ & \multicolumn{1}{c}{--} & \multicolumn{1}{c}{--} & \multicolumn{1}{c}{--} & \multicolumn{1}{c}{--} & 96.4 & \multicolumn{1}{c}{1.48} & \multicolumn{1}{c}{57.9\%} & \multicolumn{1}{c}{--} & \multicolumn{1}{c}{--} & \multicolumn{1}{c}{--} & \multicolumn{1}{c}{--} & 96.4 & \multicolumn{1}{c}{1.48} & \multicolumn{1}{c}{57.9\%} \\
\addlinespace[3pt]
\multicolumn{15}{l}{\underline{\emph{$\rho = 0.8$, $n = 1500$, $\tilde{F}_{\MSE} = 28.5$, $\tilde{F}_{\CER} = 19.8$}}} \\[1pt]
  $\hat{\tau}_1$ & 0.06 & 0.08 & 0.13 & 0.24 & 86.8 & 0.57 & 0.75 & $-0.01$ & 0.06 & 0.14 & 1.56 & 90.7 & 0.70 & 9.54 \\
  $\hat{\tau}_{\lambda^*}$ & 0.09 & 0.10 & 0.13 & 0.17 & 83.8 & 0.52 & 0.57 & 0.07 & 0.09 & 0.13 & 0.18 & 87.8 & 0.62 & 0.67 \\
  $\hat{\tau}_1^{\BC}$ & $-0.04$ & 0.04 & 0.15 & 1.57 & 93.5 & 0.84 & 1.11 & $-3.90$ & 0.02 & 0.19 & 299 & 94.3 & 1.03 & 14.0 \\
  $\hat{\tau}_{\lambda^*}^{\BC}$ & 0.03 & 0.06 & 0.15 & 0.25 & 92.4 & 0.78 & 0.88 & 0.01 & 0.06 & 0.17 & 0.30 & 93.2 & 0.93 & 1.08 \\
  $\hat{\tau}_{\lambda^*_{\BC}}^{\BC}$ & 0.04 & 0.07 & 0.15 & 0.23 & 91.8 & 0.76 & 0.83 & 0.03 & 0.07 & 0.17 & 0.27 & 92.5 & 0.89 & 0.99 \\
  $AR_{\BA}$ & \multicolumn{1}{c}{--} & \multicolumn{1}{c}{--} & \multicolumn{1}{c}{--} & \multicolumn{1}{c}{--} & 96.4 & \multicolumn{1}{c}{1.09} & \multicolumn{1}{c}{19.2\%} & \multicolumn{1}{c}{--} & \multicolumn{1}{c}{--} & \multicolumn{1}{c}{--} & \multicolumn{1}{c}{--} & 96.4 & \multicolumn{1}{c}{1.09} & \multicolumn{1}{c}{19.2\%} \\
\addlinespace[3pt]
\multicolumn{15}{l}{\underline{\emph{$\rho = -0.8$, $n = 500$, $\tilde{F}_{\MSE} = 9.8$, $\tilde{F}_{\CER} = 7.1$}}} \\[1pt]
  $\hat{\tau}_1$ & 0.53 & 0.08 & 0.19 & 34.1 & 96.7 & 1.12 & \multicolumn{1}{c}{$>10^{3}$} & 2.80 & 0.05 & 0.22 & 269 & 95.6 & 1.29 & \multicolumn{1}{c}{$>10^{5}$} \\
  $\hat{\tau}_{\lambda^*}$ & 0.00 & $-0.01$ & 0.14 & 0.23 & 93.9 & 0.83 & 0.95 & $-0.06$ & $-0.06$ & 0.16 & 0.25 & 91.7 & 0.87 & 1.02 \\
  $\hat{\tau}_1^{\BC}$ & \multicolumn{1}{c}{$>10^{3}$} & 0.02 & 0.27 & \multicolumn{1}{c}{$>10^{5}$} & 95.8 & 1.65 & \multicolumn{1}{c}{$>10^{3}$} & \multicolumn{1}{c}{$>10^{4}$} & $-0.03$ & 0.32 & \multicolumn{1}{c}{$>10^{6}$} & 94.9 & 1.91 & \multicolumn{1}{c}{$>10^{5}$} \\
  $\hat{\tau}_{\lambda^*}^{\BC}$ & 0.01 & $-0.05$ & 0.23 & 0.44 & 93.7 & 1.34 & 1.63 & $-0.04$ & $-0.11$ & 0.26 & 0.46 & 92.1 & 1.43 & 1.70 \\
  $\hat{\tau}_{\lambda^*_{\BC}}^{\BC}$ & $-0.04$ & $-0.08$ & 0.21 & 0.36 & 92.7 & 1.22 & 1.38 & $-0.09$ & $-0.13$ & 0.24 & 0.38 & 90.4 & 1.26 & 1.41 \\
  $AR_{\BA}$ & \multicolumn{1}{c}{--} & \multicolumn{1}{c}{--} & \multicolumn{1}{c}{--} & \multicolumn{1}{c}{--} & 96.4 & \multicolumn{1}{c}{1.58} & \multicolumn{1}{c}{57.9\%} & \multicolumn{1}{c}{--} & \multicolumn{1}{c}{--} & \multicolumn{1}{c}{--} & \multicolumn{1}{c}{--} & 96.4 & \multicolumn{1}{c}{1.58} & \multicolumn{1}{c}{57.9\%} \\
\addlinespace[3pt]
\multicolumn{15}{l}{\underline{\emph{$\rho = -0.8$, $n = 1500$, $\tilde{F}_{\MSE} = 28.6$, $\tilde{F}_{\CER} = 19.9$}}} \\[1pt]
  $\hat{\tau}_1$ & 0.13 & 0.08 & 0.12 & 0.53 & 98.2 & 0.66 & 1.28 & 0.13 & 0.06 & 0.13 & 1.43 & 97.3 & 0.77 & 6.78 \\
  $\hat{\tau}_{\lambda^*}$ & 0.07 & 0.05 & 0.11 & 0.19 & 97.2 & 0.59 & 0.65 & 0.04 & 0.02 & 0.11 & 0.19 & 95.6 & 0.66 & 0.73 \\
  $\hat{\tau}_1^{\BC}$ & 0.16 & 0.03 & 0.16 & 5.18 & 96.6 & 0.96 & 1.92 & 1.39 & 0.00 & 0.19 & 110 & 95.7 & 1.12 & 9.67 \\
  $\hat{\tau}_{\lambda^*}^{\BC}$ & 0.05 & 0.00 & 0.15 & 0.29 & 95.9 & 0.90 & 1.04 & 0.02 & $-0.03$ & 0.18 & 0.34 & 94.6 & 1.02 & 1.20 \\
  $\hat{\tau}_{\lambda^*_{\BC}}^{\BC}$ & 0.03 & $-0.01$ & 0.14 & 0.26 & 95.6 & 0.87 & 0.98 & 0.00 & $-0.05$ & 0.17 & 0.29 & 94.2 & 0.97 & 1.10 \\
  $AR_{\BA}$ & \multicolumn{1}{c}{--} & \multicolumn{1}{c}{--} & \multicolumn{1}{c}{--} & \multicolumn{1}{c}{--} & 97.0 & \multicolumn{1}{c}{1.20} & \multicolumn{1}{c}{16.2\%} & \multicolumn{1}{c}{--} & \multicolumn{1}{c}{--} & \multicolumn{1}{c}{--} & \multicolumn{1}{c}{--} & 97.0 & \multicolumn{1}{c}{1.20} & \multicolumn{1}{c}{16.2\%} \\
\bottomrule
\end{tabular}
\begin{tablenotes}[flushleft]
\footnotesize
\item \textit{Notes:} 10,000 replications per parameter configuration. Weak, medium and strong endogeneity correspond to $|\rho| = 0.2$, $0.5$, and $0.8$ respectively. The probability discontinuity at the cutoff is $\tau^D = 0.4$, $n$ is the sample size and $\tilde{F}$ is the median first-stage $F$-statistic within the bandwidth window. All estimators use the triangular kernel, $p=1$, and $q=2$ and $h=h_b$ where appropriate. $\hat\tau_1$ and $\hat\tau_1^{\BC}$ denote the conventional and bias-corrected estimators, $\hat\tau_{\lambda^*}$ and $\hat\tau_{\lambda^*}^{\BC}$ denote the conventional and bias-corrected $\lambda$-class estimators with $\lambda = \lambda^*$, and $\hat\tau_{\lambda^*_{\BC}}^{\BC}$ is the bias-corrected $\lambda$-class estimator with $\lambda = \lambda^*_{\BC}$. All are computed with MSE- and CER-optimal bandwidths. B, MB, MAE, and RMSE denote bias, median bias, median absolute error, and root mean squared error. Cov., MdL and MeL denote coverage, median length, and mean length. For $AR_{\BA}$, MdL is the median length over bounded, nonempty sets only, and so understates typical length when many sets are unbounded, MeL reports the percentage of unbounded sets. $AR_{\BA}$ uses its own bandwidth, so its results do not depend on  either the MSE or CER-optimal bandwidth rule. The smoothness bound is the \texttt{RDHonest} default rule of thumb. Values of absolute size 1000 or more are reported as $>10^k$ (or $<-10^k$), where $k$ is their order of magnitude.
\end{tablenotes}
\end{threeparttable}
\end{table}

% =====================================================================
%  Strong endogeneity: \citet{ludwig2007does} design
% =====================================================================
\begin{table}[htbp]
\centering
\begin{threeparttable}
\caption{Strong endogeneity: \citet{ludwig2007does} design}
\label{table dgp2_rho08}
\footnotesize
\setlength{\tabcolsep}{3pt}
\begin{tabular}{l c c c c c c c c c c c c c c}
\toprule
& \multicolumn{7}{c}{MSE-optimal bandwidth} & \multicolumn{7}{c}{CER-optimal bandwidth} \\
\cmidrule(lr){2-8}\cmidrule(lr){9-15}
& \multicolumn{4}{c}{Point estimator} & \multicolumn{3}{c}{Nominal 95\% CI} & \multicolumn{4}{c}{Point estimator} & \multicolumn{3}{c}{Nominal 95\% CI} \\
\cmidrule(lr){2-5}\cmidrule(lr){6-8}\cmidrule(lr){9-12}\cmidrule(lr){13-15}
Est. & {B} & {MB} & {MAE} & {RMSE} & {Cov.} & {MdL} & {MeL} & {B} & {MB} & {MAE} & {RMSE} & {Cov.} & {MdL} & {MeL} \\
\midrule
\multicolumn{15}{l}{\underline{\emph{$\rho = 0.8$, $n = 500$, $\tilde{F}_{\MSE} = 9.3$, $\tilde{F}_{\CER} = 6.7$}}} \\[1pt]
  $\hat{\tau}_1$ & 0.89 & 0.85 & 0.87 & 77.3 & 31.9 & 1.19 & \multicolumn{1}{c}{$>10^{4}$} & 0.61 & 0.53 & 0.56 & 4.82 & 58.0 & 1.17 & 175 \\
  $\hat{\tau}_{\lambda^*}$ & 0.80 & 0.76 & 0.76 & 0.88 & 12.6 & 0.87 & 0.98 & 0.51 & 0.49 & 0.49 & 0.57 & 36.7 & 0.84 & 0.91 \\
  $\hat{\tau}_1^{\BC}$ & 129 & 0.16 & 0.30 & \multicolumn{1}{c}{$>10^{4}$} & 95.6 & 1.77 & \multicolumn{1}{c}{$>10^{4}$} & $-19.9$ & 0.10 & 0.32 & \multicolumn{1}{c}{$>10^{3}$} & 94.9 & 1.75 & 262 \\
  $\hat{\tau}_{\lambda^*}^{\BC}$ & 0.24 & 0.20 & 0.26 & 0.51 & 93.0 & 1.39 & 1.74 & 0.18 & 0.17 & 0.25 & 0.43 & 91.5 & 1.27 & 1.45 \\
  $\hat{\tau}_{\lambda^*_{\BC}}^{\BC}$ & 0.25 & 0.22 & 0.26 & 0.45 & 91.7 & 1.25 & 1.48 & 0.20 & 0.20 & 0.24 & 0.38 & 88.8 & 1.12 & 1.22 \\
  $AR_{\BA}$ & \multicolumn{1}{c}{--} & \multicolumn{1}{c}{--} & \multicolumn{1}{c}{--} & \multicolumn{1}{c}{--} & 96.1 & \multicolumn{1}{c}{1.69} & \multicolumn{1}{c}{66.2\%} & \multicolumn{1}{c}{--} & \multicolumn{1}{c}{--} & \multicolumn{1}{c}{--} & \multicolumn{1}{c}{--} & 96.1 & \multicolumn{1}{c}{1.69} & \multicolumn{1}{c}{66.2\%} \\
\addlinespace[3pt]
\multicolumn{15}{l}{\underline{\emph{$\rho = 0.8$, $n = 1500$, $\tilde{F}_{\MSE} = 25.1$, $\tilde{F}_{\CER} = 17.3$}}} \\[1pt]
  $\hat{\tau}_1$ & 0.74 & 0.69 & 0.69 & 0.95 & 7.6 & 0.63 & 2.22 & 0.41 & 0.39 & 0.39 & 3.40 & 41.2 & 0.66 & 60.0 \\
  $\hat{\tau}_{\lambda^*}$ & 0.71 & 0.66 & 0.66 & 0.77 & 4.4 & 0.57 & 0.63 & 0.40 & 0.39 & 0.39 & 0.45 & 32.7 & 0.60 & 0.64 \\
  $\hat{\tau}_1^{\BC}$ & 0.60 & 0.10 & 0.18 & 29.5 & 94.5 & 0.93 & 3.38 & $-4.80$ & 0.04 & 0.19 & 431 & 93.2 & 0.97 & 92.4 \\
  $\hat{\tau}_{\lambda^*}^{\BC}$ & 0.12 & 0.12 & 0.17 & 0.28 & 93.4 & 0.86 & 0.98 & 0.06 & 0.08 & 0.17 & 0.28 & 92.3 & 0.88 & 0.97 \\
  $\hat{\tau}_{\lambda^*_{\BC}}^{\BC}$ & 0.13 & 0.13 & 0.17 & 0.27 & 92.8 & 0.82 & 0.92 & 0.08 & 0.09 & 0.17 & 0.26 & 91.4 & 0.84 & 0.90 \\
  $AR_{\BA}$ & \multicolumn{1}{c}{--} & \multicolumn{1}{c}{--} & \multicolumn{1}{c}{--} & \multicolumn{1}{c}{--} & 96.2 & \multicolumn{1}{c}{1.38} & \multicolumn{1}{c}{42.2\%} & \multicolumn{1}{c}{--} & \multicolumn{1}{c}{--} & \multicolumn{1}{c}{--} & \multicolumn{1}{c}{--} & 96.2 & \multicolumn{1}{c}{1.38} & \multicolumn{1}{c}{42.2\%} \\
\addlinespace[3pt]
\multicolumn{15}{l}{\underline{\emph{$\rho = -0.8$, $n = 500$, $\tilde{F}_{\MSE} = 9.3$, $\tilde{F}_{\CER} = 6.7$}}} \\[1pt]
  $\hat{\tau}_1$ & 1.14 & 0.86 & 0.89 & 14.2 & 86.0 & 2.23 & 387 & 0.72 & 0.50 & 0.56 & 68.4 & 98.0 & 2.06 & \multicolumn{1}{c}{$>10^{4}$} \\
  $\hat{\tau}_{\lambda^*}$ & 0.64 & 0.61 & 0.62 & 0.82 & 53.0 & 1.29 & 1.65 & 0.27 & 0.27 & 0.31 & 0.48 & 84.5 & 1.11 & 1.46 \\
  $\hat{\tau}_1^{\BC}$ & 18.6 & 0.08 & 0.40 & 885 & 98.1 & 3.31 & 597 & \multicolumn{1}{c}{$>10^{3}$} & $-0.01$ & 0.40 & \multicolumn{1}{c}{$>10^{5}$} & 97.0 & 3.09 & \multicolumn{1}{c}{$>10^{4}$} \\
  $\hat{\tau}_{\lambda^*}^{\BC}$ & 0.18 & 0.01 & 0.33 & 0.86 & 96.8 & 2.58 & 3.30 & 0.03 & $-0.09$ & 0.32 & 0.68 & 95.0 & 2.19 & 2.71 \\
  $\hat{\tau}_{\lambda^*_{\BC}}^{\BC}$ & 0.10 & $-0.02$ & 0.30 & 0.66 & 96.4 & 2.30 & 2.71 & $-0.04$ & $-0.11$ & 0.29 & 0.52 & 94.3 & 1.89 & 2.17 \\
  $AR_{\BA}$ & \multicolumn{1}{c}{--} & \multicolumn{1}{c}{--} & \multicolumn{1}{c}{--} & \multicolumn{1}{c}{--} & 98.8 & \multicolumn{1}{c}{2.17} & \multicolumn{1}{c}{64.1\%} & \multicolumn{1}{c}{--} & \multicolumn{1}{c}{--} & \multicolumn{1}{c}{--} & \multicolumn{1}{c}{--} & 98.8 & \multicolumn{1}{c}{2.17} & \multicolumn{1}{c}{64.1\%} \\
\addlinespace[3pt]
\multicolumn{15}{l}{\underline{\emph{$\rho = -0.8$, $n = 1500$, $\tilde{F}_{\MSE} = 25.3$, $\tilde{F}_{\CER} = 17.5$}}} \\[1pt]
  $\hat{\tau}_1$ & 0.79 & 0.71 & 0.71 & 1.49 & 44.9 & 1.21 & 4.17 & 0.53 & 0.40 & 0.40 & 3.80 & 88.3 & 1.13 & 46.9 \\
  $\hat{\tau}_{\lambda^*}$ & 0.69 & 0.63 & 0.63 & 0.80 & 29.9 & 1.01 & 1.13 & 0.35 & 0.31 & 0.31 & 0.47 & 77.1 & 0.90 & 1.00 \\
  $\hat{\tau}_1^{\BC}$ & 0.71 & 0.08 & 0.23 & 32.6 & 98.5 & 1.78 & 6.25 & 9.39 & 0.02 & 0.24 & 686 & 97.4 & 1.67 & 72.2 \\
  $\hat{\tau}_{\lambda^*}^{\BC}$ & 0.15 & 0.05 & 0.21 & 0.50 & 98.1 & 1.62 & 1.94 & 0.07 & $-0.02$ & 0.22 & 0.46 & 96.5 & 1.46 & 1.79 \\
  $\hat{\tau}_{\lambda^*_{\BC}}^{\BC}$ & 0.12 & 0.04 & 0.21 & 0.44 & 97.9 & 1.55 & 1.80 & 0.03 & $-0.04$ & 0.21 & 0.39 & 96.0 & 1.37 & 1.60 \\
  $AR_{\BA}$ & \multicolumn{1}{c}{--} & \multicolumn{1}{c}{--} & \multicolumn{1}{c}{--} & \multicolumn{1}{c}{--} & 98.6 & \multicolumn{1}{c}{1.81} & \multicolumn{1}{c}{32.1\%} & \multicolumn{1}{c}{--} & \multicolumn{1}{c}{--} & \multicolumn{1}{c}{--} & \multicolumn{1}{c}{--} & 98.6 & \multicolumn{1}{c}{1.81} & \multicolumn{1}{c}{32.1\%} \\
\bottomrule
\end{tabular}
\begin{tablenotes}[flushleft]
\footnotesize
\item \textit{Notes:} 10,000 replications per parameter configuration. Weak, medium and strong endogeneity correspond to $|\rho| = 0.2$, $0.5$, and $0.8$ respectively. The probability discontinuity at the cutoff is $\tau^D = 0.4$, $n$ is the sample size and $\tilde{F}$ is the median first-stage $F$-statistic within the bandwidth window. All estimators use the triangular kernel, $p=1$, and $q=2$ and $h=h_b$ where appropriate. $\hat\tau_1$ and $\hat\tau_1^{\BC}$ denote the conventional and bias-corrected estimators, $\hat\tau_{\lambda^*}$ and $\hat\tau_{\lambda^*}^{\BC}$ denote the conventional and bias-corrected $\lambda$-class estimators with $\lambda = \lambda^*$, and $\hat\tau_{\lambda^*_{\BC}}^{\BC}$ is the bias-corrected $\lambda$-class estimator with $\lambda = \lambda^*_{\BC}$. All are computed with MSE- and CER-optimal bandwidths. B, MB, MAE, and RMSE denote bias, median bias, median absolute error, and root mean squared error. Cov., MdL and MeL denote coverage, median length, and mean length. For $AR_{\BA}$, MdL is the median length over bounded, nonempty sets only, and so understates typical length when many sets are unbounded, MeL reports the percentage of unbounded sets. $AR_{\BA}$ uses its own bandwidth, so its results do not depend on  either the MSE or CER-optimal bandwidth rule. The smoothness bound is the \texttt{RDHonest} default rule of thumb. Values of absolute size 1000 or more are reported as $>10^k$ (or $<-10^k$), where $k$ is their order of magnitude.
\end{tablenotes}
\end{threeparttable}
\end{table}

\end{document}